\documentclass[manuscript,screen]{acmart}

\AtBeginDocument{%
  }

\setcopyright{acmcopyright}
\copyrightyear{2018}
\acmYear{2018}
\acmDOI{XXXXXXX.XXXXXXX}

\acmJournal{TOCL}
\acmPrice{15.00}
\acmISBN{978-1-4503-XXXX-X/18/06}

\usepackage{xspace}
\usepackage{booktabs}

\newcommand{\V}{\mathbb{V}}
\newcommand{\Sch}{\mathcal{S}}
\newcommand{\names}{\mathit{Names}}

\newcommand{\fvars}{\mathit{fvars}}

\newcommand{\inst}{\ensuremath{D}\xspace}

\newcommand{\In}{I}
\newcommand{\Out}{O}

\newcommand{\synI}{I^{\mathrm{syn}}}
\newcommand{\synO}{O^{\mathrm{syn}}}

\newcommand{\semI}{I^{\mathrm{sem}}}
\newcommand{\semO}{O^{\mathrm{sem}}}

\newcommand{\inn}[1]{\In(\alpha_{#1})}
\newcommand{\outt}[1]{\Out(\alpha_{#1})}

\newcommand{\interpretation}{interpretation\xspace}
\newcommand{\schema}{vocabulary\xspace}
\newcommand{\withconverse}[1]{#1}
\newcommand{\withoutconverse}[1]{}

\newcommand{\arity}{\mathit{ar}}
\newcommand{\dom}{\mathbf{dom}}
\newcommand{\ar}{\arity{}}
\newcommand{\iar}{\mathit{iar}}
\newcommand{\FV}{\mathrm{FV}}

\newcommand{\compl}[1]{\overline{#1}}
\newcommand{\cyl}[1]{\mathrm{cyl}_{#1}}
\newcommand{\cyll}[1]{\cyl{#1}^{l}}
\newcommand{\cylr}[1]{\cyl{#1}^{r}}

\newcommand{\sel}[1]{\sigma_{#1}}
\newcommand{\sellr}[1]{\sel{#1}^{\mathit{lr}}}
\newcommand{\sell}[1]{\sel{#1}^{\mathit{l}}}
\newcommand{\selr}[1]{\sel{#1}^{\mathit{r}}}

\newcommand{\conv}[1]{#1^{\smallsmile}}

\newcommand{\symdif}{\mathbin{\triangle}}
\newcommand{\val}{valuation}
\newcommand{\pval}{partial valuation}
\newcommand{\sval}{\mathcal{V}}

\newcommand{\semm}[2]{\llbracket #1 \rrbracket_{#2}}
\newcommand{\sem}[1]{\semm{#1}{\inst{}}}
\newcommand{\semNoI}[1]{\semm{#1}{}}

\newcommand{\bbar}{\text{BRV}\xspace}
\newcommand{\gbrv}{global BRV\xspace}

\newcommand{\mvr}[1]{\mathrm{mv}_{#1}^{\mathit{r}}}
\newcommand{\ox}{\overline{x}}
\newcommand{\oy}{\overline{y}}
\newcommand{\oz}{\overline{z}}
\renewcommand{\setminus}{-}
\newcommand{\comp}{\ensuremath \mathbin{;}}

\newcommand{\all}{\mathit{all}}

\newcommand{\id}{\mathit{id}}

\newcommand{\restr}[2]{#1|_{#2}}

\newcommand{\LIF}{\mathrm{LIF}}

\newcommand{\fo}{\mathrm{FO}}
\newcommand{\FO}[1]{\fo[#1]}

\newcommand\ignore[1]{}

\newcommand{\ficp}{free variable}
\newcommand{\eqv}{\mathop{=_{\mathrm{syn}}}}
\newcommand{\neqv}{\mathop{\not =_{\mathrm{syn}}}}

\newcommand{\eqq}{\equiv}

\newcommand{\compEqGen}[4]{\cylr{#3}(#1) \cap \cyll{#4}(#2)}
\newcommand{\compEqSyn}[2]{\compEqGen{#1}{#2}{\synO(#2)}{\synO(#1)}}
\newcommand{\compEqSem}[2]{\compEqGen{#1}{#2}{\semO(#2)}{\semO(#1)}}
\newcommand{\disSyn}[1]{\synI(#1) \cap \synO(#1) = \emptyset}

\newcommand{\prob}[3]{
\noindent \fbox{\parbox{\dimexpr\linewidth-2\fboxsep-2\fboxrule\relax}
{Problem: \textbf{#1}
\\ \underline{Given:} #2
\\ \underline{Decide:} #3}}%
\smallskip}

\newcommand\fullproof[1]{}

\usepackage{enumitem}

\counterwithin*{equation}{subsubsection}

\begin{document}

\theoremstyle{acmdefinition}
\newtheorem{remark}[theorem]{Remark}
   
\title{Inputs, Outputs, and Composition in the Logic of Information Flows}


% \affiliation{%

\author{Heba Aamer}
\affiliation{%
  \institution{Hasselt University}
  \city{Hasselt}
  \country{Belgium}}
\email{heba.mohamed@uhasselt.be}

\author{Bart Bogaerts}
\affiliation{%
  \institution{Vrije Universiteit Brussel}
  \city{Brussels}
  \country{Belgium}}
\email{bart.bogaerts@vub.be}

\author{Dimitri Surinx}
\affiliation{%
  \institution{Hasselt University}
  \city{Hasselt}
  \country{Belgium}}
\email{surinxd@gmail.com}

\author{Eugenia Ternovska}
\affiliation{%
  \institution{Simon Fraser University}
  \city{Burnaby, BC}
  \country{Canada}}
\email{ter@sfu.ca}

\author{Jan Van den Bussche}
\affiliation{%
  \institution{Hasselt University}
  \city{Hasselt}
  \country{Belgium}}
\email{jan.vandenbussche@uhasselt.be}

\renewcommand{\shortauthors}{Aamer et al.}

\begin{abstract}
The logic of information flows (LIF) is
a general framework in which tasks of a procedural nature can be modeled
in a declarative, logic-based fashion.
The first contribution of this paper is to propose semantic and
syntactic definitions of inputs and outputs of LIF expressions.
We study how the two relate and show that our syntactic
definition is optimal in a sense that is made precise.
The second contribution  is a systematic study of
the expressive power of sequential composition in LIF\@.
Our results on composition tie in the results on inputs and
outputs, and relate LIF to first-order logic (FO) and bounded-variable
LIF to bounded-variable FO.

This paper is the extended version of a paper presented at KR
2020~\cite{AamerLIF}.
\end{abstract}




\maketitle

\section{Introduction}

The Logic of Information Flows (LIF)~\cite{lif_amw,lif_frocos}
is a knowledge representation framework designed to model and
understand how information propagates in complex  systems, and 
to find ways to navigate it efficiently.  The basic idea is that 
modules, that can be given procedurally or declaratively, are the 
atoms of a logic whose syntax resembles first-order logic, but whose 
semantics produces new modules.
In LIF, atomic modules are modeled as relations with designated 
input and output arguments.  Computation is modeled as propagation 
of information from inputs to outputs, similarly to propagation of 
tokens in Petri nets.   The specification of a complex system then 
amounts to \emph{connecting} atomic modules together.  For this purpose, 
LIF uses the classical logic connectives, i.e., the boolean operators, 
equality, and existential quantification. The goal is to start from 
constructs that are well understood, and to address the fundamental 
question of \emph{what logical means are necessary and sufficient to 
model computations declaratively}.  The eventual goal, which goes beyond 
the topic of this paper, is to come up with restrictions or extensions 
of LIF that make the computations efficient.

In its most general form, LIF is a rich family of logics with
recursion and higher-order variables.   Atomic modules are given
by formulae in various logics, and may be viewed as solving the
task of Model Expansion \cite{MT05}: the input structure is expanded 
to satisfy the specification of a module thus producing an output.
The semantics is given in terms of pairs of structures.  We can,
for example, give a graph (a relational structure) on the input of a 
module that returns a Hamiltonian cycle on the output, and compose it 
sequentially with a module that checks whether the produced cycle is of 
even length.  One can vary both the expressiveness of logics for specifying 
atomic modules and the operations for combining modules, to achieve desirable 
complexity of the computation for the tasks of interest.

Many issues surrounding LIF, however, are already interesting in
a first-order setting (see, e.g., \cite{flifexfo}); and in fact such a 
setting is more generic than the higher-order setting, which can be obtained 
by considering relations as atomary data values.  Thus, in this paper, 
we give a self-contained, first-order presentation of LIF.  Syntactically, 
atomic modules here are relation atoms with designated input and output 
positions.  Such atoms are combined using a set of algebraic operations into 
\emph{LIF expressions}.  The semantics is defined in terms of  pairs of 
valuations of first-order variables; the first valuation represents a 
situation right before applying the module, while the second represents 
a possible situation immediately afterwards.  The results in this paper 
are then also applicable to the case of higher-order variables.

Our contributions can be summarized as follows.
\begin{enumerate}
    \item
    While the input and output arguments of atomic modules
    are specified by the vocabulary, it is not clear how to
    designate the input and output variables of a
    complex LIF expression that represents a compound module.  Actually, 
    coming up with formal definitions of what it means for a variable to 
    be an input or output is a technically and philosophically interesting 
    undertaking.  We propose semantic definitions, based on natural 
    intuitions, which are, of course, open to further debate.  The 
    semantic notions of input and output turn out to be undecidable.  This 
    is not surprising, since LIF expressions subsume classical first-order
    logic formulas, for which most inference tasks in general
    are undecidable.
    \item
    We proceed to give an approximate, syntactic definition of
    the input and output variables of a formula, which is
    effectively computable. Indeed, our syntactic
    definition is \emph{compositional}, meaning that the set of
    syntactic input (or output) variables of a formula depends only on 
    the top-level operator of the formula, and the syntactic inputs and 
    outputs of the operands.  We prove our syntactic input--output notion 
    to be \emph{sound:} every semantic input or output is also a syntactic 
    input or output, and the syntactic inputs and outputs are connected by 
    a property that we call \emph{input--output determinacy}.
    Moreover, we prove an optimality result: our definition provides the 
    most precise approximation to semantic input and outputs among
    all compositional and sound definitions.
    \item
    We investigate the expressive power of sequential composition in the 
    context of LIF\@.  The sequential composition of two modules is 
    fundamental to building complex systems.  Hence, we are motivated to 
    understand in detail whether or not this operation is expressible in 
    terms of the basic LIF connectives.  This question turns out to be
    approachable through the notion of inputs and outputs.
    Indeed, there turns out to be a simple expression for the composition 
    of \emph{io-disjoint} modules.
    Here, io-disjointness means that inputs and outputs do not overlap.  
    For example, a module that computes a function of $x$ and returns the 
    result in $y$ is io-disjoint; a module that stores the result back in 
    $x$, thus overwriting the original input, is not.
    \item
    We then use the result on io-disjoint expressions to show that composition 
    is indeed an expressible operator in the classical setting of LIF, 
    where there is an infinite supply of fresh variables.  (In contrast, the 
    expression for io-disjoint modules does not need extra variables.)
    \item
    Finally, we complement the above findings with a result on LIF in a 
    bounded-variable setting: in this setting, composition is necessarily a 
    primitive operator.
\end{enumerate}
Many of our notions and results are stated generally in terms of transition 
systems (binary relations) on first-order valuations.  Consequently, we believe 
our work is also of value to settings other than LIF inasmuch as they involve 
dynamic semantics.  Several such settings,  where input--output specifications 
are important, are discussed in the related work section.

The rest of this paper is organized  as follows.  In Section~\ref{sec:prelims}, 
we formally introduce the Logic of Information Flows from a first-order perspective.  
Section~\ref{sec:inout} presents our study concerning the notion of 
inputs and outputs of complex expressions.  Section~\ref{sec:comp} then presents 
our study on the expressibility of sequential composition.  Section
\ref{sec:related} discusses related work.  We conclude in Section~\ref{sec:concl}.
In Sections~\ref{app:soundness}, \ref{app:precision}, and \ref{app:optimality}, 
we give extensive proofs of theorems we discuss in Section~\ref{sec:inout}.

\section{Preliminaries}\label{sec:prelims}
A \emph{(module) \schema} $\Sch{}$ is a triple $(\names,\ar{},\iar{})$ 
where:
\begin{itemize}
  \item $\names$ is a nonempty set, the elements of which are called 
  \emph{module names};
  \item $\ar{}$ assigns an arity to each module name in $\names{}$;
  \item $\iar{}$ assigns an input arity to each module name $M$ in 
  $\names{}$, where $\iar{(M)}\leq\ar{(M)}$.
 \end{itemize}

We fix a countably infinite universe $\dom{}$ of data elements.  An 
\emph{interpretation} $\inst{}$ of $\Sch{}$ assigns to each module name 
$M$ in $\names$ an $\ar{}(M)$-ary relation $\inst(M)$ over $\dom{}$.

Furthermore, we fix a universe of \emph{variables} $\V$. This set may be 
finite or infinite; the size of $\V$ will influence the expressive power of 
our logic.  A \emph{\val{}} is a function from $\V$ to $\dom{}$. The set of 
all \val{}s is denoted by $\sval$. We say that $\nu_1$ and $\nu_2$ 
\emph{agree on} $Y \subseteq \V$ if $\nu_1(y)=\nu_2(y)$ for all $y\in Y$ 
and that they \emph{agree outside} $Y$ if they agree on $\V - Y$.
A \emph{\pval{}} on $Y \subseteq \V$ is a function from $Y$ to $\V$; we will 
also call this a $Y$-valuation. If $\nu$ is a valuation, we use 
$\nu|_Y$ to denote its restriction to $Y$.  Let $\nu$ be a valuation and 
let $\nu_1$ be a partial valuation on $Y\subseteq \V$.  Then the substitution 
of $\nu_1$ into $\nu$, denoted by $\nu[\nu_1]$, is defined as 
$\nu_1\cup(\restr{\nu}{\V-Y})$.
In the special case where $\nu_1$ is defined on a single variable $x$
with $\nu_1(x) = d$, we also write $\nu[\nu_1]$ as $\nu[x:d]$.

We assume familiarity with the syntax and semantics of first-order logic 
(FO, relational calculus) over $\Sch$ \cite{Enderton72} and use $:=$ to 
mean ``is by definition''.

\subsection{Binary Relations on Valuations}
The semantics of LIF will be defined in terms of binary relations on 
$\sval{}$ (abbreviated $\bbar$: Binary Relations on Valuations).  Before 
formally introducing LIF, we define operations on $\bbar{}$s corresponding 
to the classical logical connectives, adapted to a dynamic semantics.  For
boolean connectives, we simply use the standard set operations.  For equality, 
we introduce selection operators.  For existential quantification, we introduce 
cylindrification operators.

Let $A$ and $B$ be $\bbar{}$s, let $Z$ be a finite set of variables, and let 
$x$ and $y$ be variables.
\begin{itemize}
\item \textbf{Set operations:} $A\cup B, A \cap B$, and $A\setminus B$ are 
well known.
\item \textbf{Composition}
\[A \comp B := \{(\nu_1,\nu_2)\mid \exists \nu_3: 
(\nu_1,\nu_3)\in A \text{ and } (\nu_3,\nu_2) \in B\}.\]
\withconverse{\item \textbf{Converse}
\[\conv{A} := \{(\nu_1,\nu_2) \mid (\nu_2,\nu_1)\in A\}.\]}
\item \textbf{Left and Right Cylindrifications}
\begin{align*}
\cyll{Z}(A) &:= \{(\nu_1,\nu_2) \mid \exists \nu_1': (\nu_1',\nu_2)\in A
\text{ and $\nu_1'$ and $\nu_1$ agree outside $Z$}\}.\\
\cylr{Z}(A) &:= \{(\nu_1,\nu_2) \mid \exists \nu_2': (\nu_1,\nu_2')\in A
\text{ and $\nu_2'$ and $\nu_2$ agree outside $Z$}\}.
\end{align*}
\item \textbf{Left and Right Selections}
\begin{align*}
\sell{x=y}(A) &:= \{(\nu_1,\nu_2)\in A \mid \nu_1(x) = \nu_1(y) \}.\\
\selr{x=y}(A) &:= \{(\nu_1,\nu_2)\in A \mid  \nu_2(x) = \nu_2(y)\}.
\end{align*}
\item \textbf{Left-To-Right Selection}
\[\sellr{x=y}(A) := \{(\nu_1,\nu_2)\in A \mid \nu_1(x) = \nu_2(y)\}.\]
\end{itemize}
If $\bar{x}$ and $\bar{y}$ are tuples of variables of length $n$, we write 
$\sellr{\bar{x}=\bar{y}}(A)$ for
\[\sellr{x_1=y_1} \sellr{x_2=y_2}\dots\sellr{x_n=y_n}(A)\]
and if $z$ is a variable we write $\cyll{z}$ for $\cyll{\{z\}}$.
Intuitively, a $\bbar$ is a dynamic system that manipulates the interpretation 
of variables.  A pair $(\nu_1,\nu_2)$ in a $\bbar$ represents that a transition 
from $\nu_1$ to $\nu_2$ is possible, i.e., that when given $\nu_1$ as input, 
the values of the variables can be updated to $\nu_2$.  The operations 
defined above correspond to manipulations/combinations of such dynamic systems.
Union, for instance, represents a non-deterministic choice, while composition 
corresponds to composing two such systems.  Left cylindrification corresponds, 
in the dynamic view, to performing search before following the underlying $\bbar$.
Indeed, when given an input $\nu_1$, alternative values for the cylindrified 
variables are searched for which transitions are possible.  The selection operations 
correspond to performing checks, on the input, the output, or a combination of both 
in addition to performing what the underlying $\bbar$ does.

Some of the above operators are redundant, in the sense that they can be 
expressed in terms of others, for instance, $A\cap B = A \setminus (A \setminus B)$.  
We also have:
\begin{lemma}\label{lem:redundant}
For any $\bbar$ $A$, and any variables $x$ and $y$, the following hold:
\begin{align*}
\withconverse{  
  \cylr{x}(A) & = \conv{(\cyll{x}(\conv{A}))}\\
  \cyll{x}(A) & = \conv{(\cylr{x}(\conv{A}))}\\}
  \selr{x=y}(A) &= A \cap \cyll{x}\sellr{(x,x)=(y,x)}\cyll{x}(A)\\
  \sell{x=y}(A) &= A \cap \cylr{x}\sellr{(y,x)=(x,x)}\cylr{x}(A)
  \withconverse{\\
  \sell{x=y}(A) &= \conv{\selr{x=y}(\conv{A})} }
\end{align*}
\end{lemma}
The expression for $\selr{x=y}$ can be explained as follows.  First, we copy $x$ 
from right to left by applying $\cyll{x}$ followed by $\sellr{x=x}$. Selection 
$\selr{x=y}$ can now be simulated by  $\sellr{x=y}$.  The original $x$ value on 
the left is restored by a final application of $\cyll{x}$ and intersecting with 
the original $A$.

\subsection{The Logic of Information Flows}
The language of LIF expressions $\alpha$ over a \schema $\Sch{}$ is defined by 
the following grammar:
\[
\alpha ::= \id \mid M(\oz) \mid \alpha \cup \alpha \mid \alpha \cap \alpha \mid 
\alpha \setminus \alpha \mid \alpha \comp{} \alpha \withconverse{ \mid\conv{\alpha}} 
\mid \cyll{Z}(\alpha) \mid \cylr{Z}(\alpha) \mid \sellr{x=y}(\alpha) \mid 
\sell{x=y}(\alpha)\mid \selr{x=y}(\alpha)
\]
Here, $M$ is any module name in $\Sch{}$; $Z$ is a finite set of variables; 
$\oz$ is a tuple of variables; and $x,y$ are variables.  For \emph{atomic module 
expressions}, i.e., expressions of the form $M(\oz{})$, the length of $\oz{}$ 
must equal $\ar(M)$.  In practice, we will often write $M(\ox;\oy)$ for atomic 
module expressions, where $\ox$ is a tuple of variables of length $\iar(M)$ and 
$\oy$ is a tuple of variables of length $\ar(M)- \iar(M)$.

We will define the semantics of a LIF expression $\alpha$, in the context of a 
given interpretation $D$, as a $\bbar$ which will be denoted by $\sem\alpha$.  
Thus, adapting Gurevich's terminology~\cite{gurevich_algfeas,gurevich_challenge}, 
every LIF expression $\alpha$ denotes a \emph{\gbrv} $\semNoI{\alpha}$: a function 
that maps interpretations $D$ of $\Sch$ to the \bbar $\alpha(D) := \sem{\alpha}$.

For atomic module expressions, we define
\[
\sem{M(\ox;\oy)} := \{(\nu_1,\nu_2)\in \sval{\times} \sval 
\mid \nu_1(\ox)\cdot\nu_2(\oy) \in \inst(M) 
\text{ and } \nu_1 \text{ and } \nu_2 \text{ agree outside } \oy{}\}.
\]
Here, $\nu_1(\ox)\cdot\nu_2(\oy)$ denotes the \emph{concatenation} of tuples.
Intuitively, the semantics of an expression $M(\ox;\oy)$ represents a transition 
from $\nu_1$ to $\nu_2$: the inputs of the module are ``read'' in $\nu_1$ and 
the outputs are updated in $\nu_2$.
The value of every variable that is not an output is preserved; this important 
semantic principle is a realization of the commonsense law of \emph{inertia} 
\cite{McCHay69,DBLP:conf/ijcai/Lifschitz87}.

We further define
\[
\sem{\id} := \{(\nu,\nu) \mid \nu \in \sval{}\}.
\]
The semantics of other operators is obtained directly by applying the 
corresponding operation on $\bbar$s, e.g.,
\begin{align*}
  \sem{\alpha \setminus \beta} & := \sem{\alpha}\setminus \sem{\beta}.\\
  \sem{\sellr{x=y}(\alpha)}& := \sellr{x=y}(\sem{\alpha}).
\end{align*}
We say that $\alpha$ and $\beta$ are \emph{equivalent} if
$\sem{\alpha}=\sem{\beta}$  for each \interpretation $\inst{}$,
i.e., if they denote the same \gbrv.

\subsection{Satisfiability of LIF Expressions}
In this section, we will show that the problem of deciding whether a given LIF 
expression is satisfiable is undecidable.  Thereto we begin by noting that 
first-order logic (FO) is naturally embedded in LIF in the following manner.
When evaluating FO formulas on interpretations, we agree that the domain of 
quantification is always $\dom$.
\begin{lemma}\label{lem:FoToLIF}
Let $\Sch$ be a vocabulary with $\iar(R)=0$ for every $R \in \Sch$. Then, for every 
FO formula $\varphi$ over $\Sch$, there exists a LIF expression $\alpha_\varphi$ 
such that for every \interpretation $\inst{}$ the following holds:
\[\sem{\alpha_\varphi}=\{(\nu,\nu) \mid \inst,\nu\models \varphi\}.\]
\end{lemma}
\begin{proof}
The proof is by structural induction on $\varphi$.
\begin{itemize}
	\item If $\varphi$ is $x=y$, 
	take $\alpha_\varphi=\selr{x=y}(\id)$.
	\item If $\varphi$ is $R(\bar x)$ for some $R \in \Sch$, 
	take $\alpha_\varphi = \id \cap R(;\bar x)$.
	\item If $\varphi$ is $\varphi_1\lor\varphi_2$, 
	take  $\alpha_\varphi = \alpha_{\varphi_1}\cup \alpha_{\varphi_2}$.
	\item If $\varphi$ is $\neg \varphi_1$, 
	take $\alpha_\varphi = \id \setminus \alpha_{\varphi_1}$.
	\item If $\varphi$ is $\exists x \, \varphi_1$, 
	take $\alpha_\varphi = $ $\sellr{x=x}(\cyll{x}(\cylr{x}(\alpha_{\varphi_1})))$.
	\qedhere
\end{itemize}
\end{proof}
It is well known that satisfiability of FO formulas over a fixed
countably infinite domain is undecidable.  This leads to the
following undecidability result.

\smallskip\smallskip

\prob{Satisfiability}
{a LIF expression $\alpha$.}
{Is there an \interpretation $\inst$ such that $\sem{\alpha} \neq \emptyset$?}

\begin{proposition}\label{prop:satund}
The satisfiability problem is undecidable.
\end{proposition}
\begin{proof}
The proof is by reduction from the satisfiability of FO formulas.
Let $\varphi$ be an FO formula and let $\alpha_\varphi$ be the LIF expression 
obtained from Lemma \ref{lem:FoToLIF}.  It is clear that $\alpha_\varphi$ is 
satisfiable iff $\varphi$ is.
\end{proof}%

\section{Inputs and Outputs}\label{sec:inout}
We are now ready to study inputs and outputs of LIF expressions, and, more 
generally, of global $\bbar$s. We first investigate what inputs and outputs 
mean on the semantic level before introducing a syntactic definition for LIF 
expressions.

\subsection{Semantic Inputs and Outputs for Global \texorpdfstring{$\bbar$}{BRV}s}
Intuitively, an output is a variable whose value can be changed by the 
expression, i.e., a variable that is not subject to inertia.
\begin{definition}
A variable $x$ is a \emph{semantic output} for a \gbrv $Q$ if there exists 
an \interpretation $\inst{}$ and $(\nu_1,\nu_2)\in Q(\inst{})$ such that 
$\nu_1(x) \neq \nu_2(x)$.  We use $\semO(Q)$ to denote the set of semantic
output variables of $Q$.  If $\alpha$ is a LIF expression, we call a variable 
a \emph{semantic output} of $\alpha$ if it is a semantic output of 
$\semNoI{\alpha}$.  We also write $\semO(\alpha)$ for the semantic outputs 
of $\alpha$.  A variable that is not a semantic output is also called an 
\emph{inertial variable}.
\end{definition}
Defining semantic inputs is a bit more subtle.  Intuitively, a variable is 
an input for a BRV if its value on the left-hand side matters for determining 
the right-hand side (i.e., that if the value of the input would have been 
different, so would have been the right-hand side; which is in fact a very 
coarse counterfactual definition of actual causality \cite{lewis}).  However, 
a naive formalization of this intuition would result in a situation in which 
all inertial variables (variables that are not outputs) are inputs since their 
value on the right-hand side always equals to the one on the left-hand side.  
A slight refinement of our intuition is that the inputs are those variables 
whose value matters for determining the possible values of the outputs.  
This is formalized in the following definitions.
\begin{definition}\label{def:determinacy}
Let $Q$ be a \gbrv and $X,Y$ be sets of variables.  We say that $X$ 
\emph{determines} $Q$ on $Y$ if for every \interpretation $\inst{}$, 
every $(\nu_1,\nu_2)\in Q(\inst{})$ and every $\nu_1'$ such that 
$\nu_1'=\nu_1$ on $X$, there exists a $\nu_2'$ such that $\nu_2'=\nu_2$ 
on $Y$ and $(\nu_1',\nu_2')\in Q(\inst{})$.
\end{definition}
\begin{definition}
A variable $x$ is a \emph{semantic input} for a \gbrv $Q$ if\/  $\V-\{x\}$ 
does not determine $Q$ on $ \semO(Q)$.  The set of input variables of $Q$
is denoted by $\semI(Q)$.  A variable is a \emph{semantic input} of a LIF 
expression $\alpha$ if it is a semantic input of $\semNoI{\alpha}$; the 
semantic inputs of $\alpha$ are denoted by $\semI(\alpha)$.
\end{definition}

From Definition~\ref{def:determinacy}, we can rephrase the definition for
semantic inputs to:
\begin{proposition}\label{prop:inputEqu}
A variable $x$ is a semantic input for a \gbrv $Q$ iff there is an 
\interpretation $\inst{}$, a value $d \in \dom$, and $(\nu_1,\nu_2)\in Q(\inst{})$ 
such that there is no valuation $\nu_2'$ that agrees with $\nu_2$ on $\semO(Q)$ 
and $(\nu_1[x:d],\nu_2')\in Q(\inst{})$.
\end{proposition}

The following proposition shows that the semantic inputs of $Q$ are indeed exactly 
the variables that determine $Q$.
\begin{proposition}\label{prop:determined1}
If a set of variables $X$ determines a \gbrv $Q$ on $\semO(Q)$, then 
$\semI(Q) \subseteq X$.
\end{proposition}
\begin{proof}
Let $v$ be any variable in $\semI(Q)$.  We know that $\V - \{v\}$ does not 
determine $Q$ on $\semO(Q)$.  If $v \not \in X$, then $X \subseteq \V - \{v\}$,
so $X$ would not determine $Q$ on $\semO(Q)$, which is impossible.  Hence, $v$
must be in $X$ as desired.
\end{proof}
Under a mild assumption, also the converse to Proposition~\ref{prop:determined1}
holds:\footnote{Proposition~\ref{prop:determined2} indeed provides a converse
to Proposition~\ref{prop:determined1}: given that $\semI(Q)$ determines $Q$
on $\semO(Q)$ and $\semI(Q) \subseteq X$ for some set $X$, clearly also $X$ 
determines $Q$ on $\semO(Q)$.}
\begin{proposition}\label{prop:determined2}
Assume there exists a finite set of variables that determines a \gbrv $Q$ 
on $\semO(Q)$.  Then, $\semI(Q)$ determines $Q$ on $\semO(Q)$.
\end{proposition}
\begin{proof}
Let $(\nu_1, \nu_2) \in Q(\inst)$ and $\nu'_1 = \nu_1$ on $\semI(Q)$ for some
\interpretation $\inst$ and valuations $\nu_1$, $\nu_2$, and $\nu'_1$.  To show 
that $\semI(Q)$ determines $Q$ on $\semO(Q)$, we need to find a valuation $\nu'_2$ 
such that $(\nu'_1, \nu'_2) \in Q(\inst)$ and $\nu'_2 = \nu_2$ on $\semO(Q)$.
By assumption, let $X$ be a finite set of variables that determines $Q$ on 
$\semO(Q)$.

Thereto, take $\nu''_1$ to be the valuation $\nu'_1[\nu_1|_{X}]$ which is the 
valuation $\nu'_1$ after changing the values for the variables in $X$ to be as 
in $\nu_1$.  Thus, $\nu''_1 = \nu_1$ on $X$, while $\nu''_1 = \nu'_1$ outside $X$.
Since $X$ determines $Q$ on $\semO(Q)$, we know that there is a valuation $\nu''_2$
such that $(\nu''_1, \nu''_2) \in Q(\inst)$ and $\nu''_2 = \nu_2$ on $\semO(Q)$.
To reach our goal, we would like to do incremental changes to $\nu''_1$ in order 
to be similar to $\nu'_1$ while showing that each of the intermediate valuations 
does satisfy the determinacy conditions.

From construction, we know that $\nu''_1 = \nu'_1$ on $X \cap \semI(Q)$.  Using
the finiteness assumption for $X$, let $X - \semI(Q)$ be the set of variables
$\{x_1, \ldots, x_n\}$.  Define the sequence of valuations 
$\mu_0, \mu_1, \ldots, \mu_n$ such that:
\begin{itemize}
    \item $\mu_0 := \nu''_1$; and
    \item $\mu_{i} := \mu_{i-1}[\{x_{i} \mapsto \nu'_1(x_{i})\}]$ so $\mu_{i}$
    is $\mu_{i-1}$ after changing the value of $x_{i}$ to be as in $\nu'_1$.
\end{itemize}
We \textbf{claim} that for $i \in \{0, \ldots, n\}$, there exists a valuation 
$\kappa_i$ such that $(\mu_i, \kappa_i) \in Q(\inst)$ and $\kappa_i = \nu_2$ on 
$\semO(Q)$.  Since $\mu_n$ is clearly the same valuation as $\nu'_1$, we can then 
take $\nu'_2$ to be $\kappa_n$ which is the required.

We verify our \textbf{claim} by induction.
\begin{description}
\item[Base Case:] for $i = 0$, we can see that $\kappa_0 = \nu''_2$.
\item[Inductive Step:] for $i > 0$, by assumption, we know that there is a 
valuation $\kappa_{i-1}$ such that $(\mu_{i-1}, \kappa_{i-1}) \in Q(\inst)$ and 
$\kappa_{i-1} = \nu_2$ on $\semO(Q)$.  It is clear that $\mu_i = \mu_{i-1}$ 
outside $\{x_i\}$ which is $\V - \{x_i\}$.  Since $x_i \not \in \semI(Q)$, we 
know that $\V - \{x_i\}$ determines $Q$ on $\semO(Q)$.  Hence, there is a valuation 
$\kappa_i$ such that $(\mu_i, \kappa_i) \in Q(\inst)$ and $\kappa_i = \kappa_{i-1}$
on $\semO(Q)$.  Since $\kappa_{i-1} = \nu_2$ on $\semO(Q)$, we can see that 
$\kappa_i = \nu_2$ on $\semO(Q)$. \qedhere
\end{description}
\end{proof}
In the next remark, we show that without our assumption, we can find an example 
of a \gbrv that is not determined on its semantic outputs by its semantic inputs. 
\begin{remark}\label{note:determinacyCE}
Let $Q$ be the \gbrv that maps every $\inst$ to the same $\bbar$, namely:
\[
Q(\inst) = \{(\nu_1,\nu_2) \in \sval \times \sval \mid \nu_1 \text{ and } \nu_2 
\text{ differ on finitely many variables}\}.
\]
Since the variables that can be changed by $Q$ are not restricted, we see that 
$\semO(Q) = \V$.  We now verify that $\semI(Q) = \emptyset$.  Let $v$ be any
variable.  We can see that $v \not \in \semI(Q)$.  Thereto, we check that 
$\V - \{v\}$ determines $Q$ on $\semO(Q)$.  Let $\inst$ be an \interpretation 
and $\nu_1$, $\nu_2$, and $\nu_1'$ valuations such that $(\nu_1, \nu_2) \in Q(\inst)$ 
and $\nu_1' = \nu_1$ outside $\{v\}$.  Since $\nu_1$ and $\nu_2$ differ on finitely 
many variables, we can see that $\nu_1'$ and $\nu_2$ also do.  Hence, 
$(\nu_1', \nu_2) \in Q(\inst)$.

Finally, we verify that $\semI(Q)$ does not determine $Q$ on $\semO(Q)$.  To see
a counterexample, let $\inst$ be an \interpretation and $\nu_1$ be the valuation 
that assigns $1$ to every variable.  We can see that $(\nu_1, \nu_1) \in Q(\inst)$.  
Let $\nu_2$ be the valuation that assigns $2$ to every variable.  It is clear that 
$\nu_2 = \nu_1$ on $\semI(Q) = \emptyset$, however, clearly 
$(\nu_2, \nu_1) \not \in Q(\inst)$.  By Proposition~\ref{prop:determined2}, we know 
that there is no finite set of variables that does determine $Q$ on $\semO(Q)$.\qed
\end{remark}
The reader should not be lulled into believing that $\semI(Q)$ determines a \gbrv 
$Q$ on the set $\V$ of \emph{all} variables since $\semI(Q)$ determines $Q$ on 
$\semO(Q)$ and no other variable outside $\semO(Q)$ can have its value changed.  
In the following remark, we give a simple counterexample.
\begin{remark}
We show that $\semI(Q)$ does not necessarily determine $Q$ on $\V$ for every 
\gbrv $Q$.  Take $Q$ to be defined by the LIF expression $\sell{x=y}(\id)$,
so, for every $\inst$, we have
\[
Q(\inst) = \{(\nu, \nu) \mid \nu \in \sval \text{ such that } \nu(x) = \nu(y)\}.
\]
It is clear that $\semO(Q) = \emptyset$ and $\semI(Q) = \{x,y\}$.  

Let $\nu_1$ be the valuation that assigns $1$ to every variable.  Clearly, 
$(\nu_1, \nu_1) \in Q(\inst)$ for any \interpretation $\inst$.  Now take $\nu'_1$ 
to be the valuation that assigns $1$ to $x$ and $y$, while it assigns $2$ to every
other variable.  It is clear that $\nu'_1 = \nu_1$ on $\semI(Q)$.

If $\semI(Q)$ were to determine $Q$ on $\V$, we should find a valuation $\nu_2$ 
that agrees with $\nu_1$ on $\V$ such that $(\nu'_1, \nu_2) \in Q(\inst)$. 
In other words, this means that $(\nu'_1, \nu_1) \in Q(\inst)$, which is clearly 
not possible.\qed
\end{remark}
Assuming determinacy holds for some \gbrv, we next show that it 
actually satisfies an even stricter version of determinacy.
\begin{proposition}\label{prop:semDetToAD}
Let $Q$ be a \gbrv.  If $\semI(Q)$ determines $Q$ on $\semO(Q)$, then
for every \interpretation $\inst$, every $(\nu_1, \nu_2) \in Q(\inst)$ and every 
$\nu_1'$ that agrees with $\nu_1$ on $\semI(Q)$ and outside $\semO(Q)$, we have 
$(\nu_1', \nu_2) \in Q(\inst)$.
\end{proposition}
\begin{proof}
Suppose that
\begin{equation}
(\nu_1, \nu_2) \in Q(\inst) \text{ and } \nu'_1 = \nu_1 \text{ on } \semI(Q) 
\text{ and outside } \semO(Q). \label{eq:semDetToAD1}    
\end{equation}
We know that all the variables outside $\semO(Q)$ are inertial, so 
\begin{equation}
\nu_1 = \nu_2 \text{ outside } \semO(Q).    \label{eq:semDetToAD2}
\end{equation}
Since $\semI(Q)$ determines $Q$ on $\semO(Q)$, we know that there exists $\nu_2'$ 
such that:
\begin{enumerate}[label=(\roman*)]
	\item\label{semDetToAD1} $\nu_2' = \nu_2$ on $\semO(Q)$;
	\item\label{semDetToAD2} $(\nu_1', \nu_2') \in \sem{Q}$.
\end{enumerate}
From~\ref{semDetToAD2}, we know that
\begin{equation}
\nu_1 = \nu_2 \text{ outside } \semO(Q). \label{eq:semDetToAD3}
\end{equation}
Together Equations~(\ref{eq:semDetToAD1}--\ref{eq:semDetToAD3}) imply that 
$\nu'_2 = \nu_2$ outside $\semO(Q)$.  Combining this with~\ref{semDetToAD1}, we 
know that $\nu_2=\nu_2'$, whence $(\nu_1',\nu_2)\in Q(\inst)$ as desired.
\end{proof}

Intuitively, the inputs and outputs are the only variables that
matter for a given \gbrv, similar to how in classical logic the free variables 
are the only ones that matter. All other variables can take arbitrary values, 
\emph{but}, their values are preserved by inertia, i.e., remain unchanged by the 
dynamic system. We now formalize this intuition.
\begin{definition}\label{def:inertialCyl}
Let $Q$ be a \gbrv and $X$ a set of variables.  We say that $Q$ is 
\emph{inertially cylindrified} on $X$ if:
\begin{enumerate}
\item all variables in $X$ are inertial; and
\item for every \interpretation $\inst$, every $(\nu_1,\nu_2)\in Q(\inst)$, and 
every $X$-valuation $\nu'$ also $(\nu_1[\nu'],\nu_2[\nu'])\in Q(\inst)$.
\end{enumerate}
\end{definition}
\begin{proposition}\label{prop:inertcyl}
Every \gbrv $Q$ is inertially cylindrified outside the semantic inputs and outputs 
of $Q$ assuming that $\semI(Q)$ determines $Q$ on $\semO(Q)$.
\end{proposition}
\begin{proof}
Let $Q$ be a \gbrv such that $\semI(Q)$ determines $Q$ on $\semO(Q)$.  Moreover, 
let $X$ be the set of variables that are neither semantic inputs nor semantic 
outputs of $Q$.  It is trivial to show that all the variables in $X$ are inertial 
since none of the variables in $X$ is a semantic output of $Q$.  What remains to 
show is that for every \interpretation $\inst$, every $(\nu_1,\nu_2) \in Q(\inst)$, 
and every $X$-valuation $\nu'$ also $(\nu_1[\nu'],\nu_2[\nu'])\in Q(\inst)$.

Let $(\nu_1, \nu_2) \in Q(\inst)$ for an arbitrary \interpretation $\inst$ and 
let $\nu'_1 = \nu_1[\nu']$ be a valuation for some $X$-valuation $\nu'$.  Since
$\nu'_1 = \nu_1$ on $\semI(Q)$, we know by determinacy that there is a valuation 
$\nu'_2$ such that $(\nu'_1, \nu'_2) \in Q(\inst)$ and $\nu'_2 = \nu_2$ on 
$\semO(Q)$.  We now argue that $\nu'_2 = \nu_2[\nu']$. On the variables of 
$\semO(Q)$, we know that $\nu_2 = \nu_2[\nu']$, whence, $\nu'_2 = \nu_2[\nu']$ 
on $\semO(Q)$.  Now we consider the variables that are not in $\semO(Q)$.  
It is clear that $\nu_1 = \nu_2$ outside $\semO(Q)$, whence, 
$\nu_1[\nu'] = \nu'_1 = \nu'_2 = \nu_2[\nu']$ outside $\semO(Q)$. 
\end{proof}
\begin{remark}
Without the assumption, we can give an example of a \gbrv that is not 
inertially cylindrified outside its semantic inputs and outputs.  Let $Q$ be 
the \gbrv that maps every $\inst$ to the same $\bbar$, namely:
\[
Q(\inst) = \{(\nu,\nu) \mid \nu \in \sval 
\text{ and no value in the domain occurs infinitely often in } \nu\}.
\]
It is clear that $\semO(Q) = \emptyset$ and $\semI(Q) = \emptyset$.  

We proceed to verify that $Q$ is not inertially cylindrified on $\V$.  Let 
$\inst$ be any \interpretation and $\nu$ be any valuation that maps every 
variable to a unique value from the domain.  We can see that 
$(\nu, \nu) \in \sem{Q}$ since every value in $\nu$ appears only once.  
Now fix some $a \in \dom$ arbitrarily and consider the valuation $\nu'$ 
that maps every variable to $a$.  We can see that 
$(\nu[\nu'], \nu[\nu']) \not \in \sem{Q}$ since $a$ appears infinitely often 
in $\nu' = \nu[\nu']$.\qed
\end{remark}
We remark that the converse of Proposition~\ref{prop:inertcyl} is not true:
\begin{remark}
Consider the same \gbrv Q discussed in Remark~\ref{note:determinacyCE} where
we showed that $\semI(Q)$ does not determine $Q$ on $\semO(Q)$.  Recall that 
$\semO(Q) = \V$, so the set of variables outside the semantic inputs and
outputs is empty.  Trivially, however, $Q$ is  inertially cylindrified on 
$\emptyset$.
\end{remark}

\subsection{Semantic Inputs and Outputs for LIF Expressions}

\subsubsection{Properties of LIF Expressions}
As we have discussed in the previous section, a \gbrv has the properties of
determinacy and inertial cylindrification under the assumption that there is
a finite set of variables that determines the \gbrv on its semantic outputs.
From the results of Section~\ref{sec:synIO}, this assumption indeed holds for 
LIF assumptions.  Indeed, we will define there, for every expression $\alpha$, 
a finite set of ``syntactic input variables'' and we will show that this set 
determines $\semNoI{\alpha}$ on a set of ``syntactic output variables''.  
Moreover, the latter set will contain $\semO(\alpha)$.

\subsubsection{Determining Semantic Inputs and Outputs}
For atomic LIF expressions, the semantic inputs and outputs are easy to 
determine, as we will show first.  Unfortunately, we show next that the problem 
is undecidable for general expressions.

We show that semantic inputs and outputs are exactly what one expects for atomic 
modules:
\begin{proposition}
If $\alpha$ is an atomic LIF expressions $M(\bar{x};\bar{y})$, then
\begin{itemize}
    \item $\semI(\alpha)= \{x_i \mid \bar{x} = x_1, \ldots, x_n 
    \text{ for } i \in \{1, \ldots, n\}\}$; and
    \item $\semO(\alpha) = \{y_i \mid \bar{y} = y_1, \ldots, y_m 
    \text{ for } i \in \{1, \ldots, m\}\}$.
\end{itemize}
\end{proposition}
\begin{example}\label{ex:inc:basic}
A variable can be both input and output of a given expression.  A very simple example 
is an atomic module $P_1(x;x)$.  To illustrate where this can be useful, assume 
$\dom = \mathbb{Z}$ and consider an \interpretation $\inst$ such that 
$\inst(P_1) = \{(n,n+1)\mid n\in \mathbb{Z}\}$.  In that case, the expression
$P_1(x;x)$ represents a dynamic system in which the value of $x$ is incremented by 
$1$; $x$ is an output of the system since its value is changed; it is an input 
since its original value matters for determining its value in the output.
\end{example}

We will now show that the problem of deciding whether a given variable is a semantic 
input or output of a LIF expression is undecidable.  Proposition~\ref{prop:satund}
showed that satisfiability of LIF expressions is undecidable.  This leads to the
following undecidability results.

\smallskip\smallskip

\prob{Semantic Output Membership}
{a variable $x$ and a LIF expression $\alpha$.}
{$x \in \semO(\alpha)$?}

\begin{proposition}\label{prop:semoutputund}
The semantic output membership problem is undecidable.
\end{proposition}
\begin{proof}
The proof is by reduction from the satisfiability of LIF expressions.  
Let $\alpha$ be a LIF expression.  Take $\beta$ to be $\cyll{x}(\alpha)$.  
What remains to show is that $x \in \semO(\beta)$ $\Leftrightarrow$ $\alpha$ 
is satisfiable.

$(\Rightarrow)$ Let $x \in \semO(\beta)$.  Then, there is certainly an 
interpretation $\inst$ and valuations $\nu_1$ and $\nu_2$ such that 
$(\nu_1, \nu_2) \in \sem{\cyll{x}(\alpha)}$.  Hence, there is also a valuation 
$\nu_1'$ such that $(\nu_1',\nu_2)\in \sem{\alpha}$.  Certainly, $\alpha$ is 
satisfiable.

$(\Leftarrow)$ Let $\alpha$ be satisfiable.  Then, there is an \interpretation 
$\inst$ and valuations $\nu_1$ and $\nu_2$ such that 
$(\nu_1,\nu_2) \in \sem{\alpha}$.  Also, let $\nu'$ be an $\{x\}$-valuation that 
maps $x$ to $a$ with $a \neq \nu_2(x)$.  It is clear then that 
$(\nu[\nu'], \nu) \in \sem{\cyll{x}(\alpha)}$.  We thus see that 
$x \in \semO(\beta)$.
\end{proof}%

\prob{Semantic Input Membership}
{a variable $x$ and a LIF expression $\alpha$.}
{$x \in \semI(\alpha)$?}

\begin{proposition}
The semantic input membership problem is undecidable.
\end{proposition}
\begin{proof}
Let $\alpha$ be a LIF expression.  Take $\beta$ to be $\sell{x=z}(\alpha)$, 
where $z$ is a variable that is not used in $\alpha$ and different from $x$.  
What remains to show is that $x \in \semI(\beta) $ $\Leftrightarrow$ $\alpha$ 
is satisfiable.

$(\Rightarrow)$ Let $x \in \semI(\beta)$. Then, certainly, there is an 
interpretation $\inst$ and valuations $\nu_1$ and $\nu_2$ such that 
$(\nu_1,\nu_2) \in \sem{\sell{x=z}(\alpha)} \subseteq \sem{\alpha}$. 
Certainly, $\alpha$ is satisfiable.

$(\Leftarrow)$ Let $\alpha$ be satisfiable.  Then, there is an interpretation 
$\inst$ and valuations $\nu_1$ and $\nu_2$ such that $(\nu_1,\nu_2) \in \sem{\alpha}$. 
Without loss of generality, we can assume that $\nu_1(z)=\nu_1(x)$ since $z$ is a 
fresh variable.  Hence, $(\nu_1, \nu_2) \in \sem{\sell{x=z}(\alpha)}$.  
Let $\nu'_1$ be a valuation that agrees with $\nu_1$ outside $x$ such that 
$\nu'_1(x) \neq \nu_1(x)$.  Since $x$ and $z$ are different variables, also 
$\nu'_1(x) \neq \nu'_1(z)$, so clearly there is no valuation $\nu'_2$ such that 
$(\nu'_1, \nu'_2) \in \sem{\sell{x=z}(\alpha)}$.  We then see that 
$x \in \semI(\beta)$.
\end{proof}

\subsection{Syntactic Inputs and Outputs}\label{sec:synIO}
Since both the membership problems for semantic inputs and outputs are 
undecidable, to determine inputs and outputs in practice, we will need 
decidable approximations of these concepts.  Before giving our syntactic 
definition, we define some properties of candidate definitions.
\begin{definition}\label{def:soundness}
Let $I$ and $O$ be functions from LIF expressions to sets of variables.
We say that $(I,O)$ is a \emph{sound input--output definition} if the following hold:
\begin{itemize}
\item If $\alpha = M(\ox;\oy)$, then $I(\alpha)=\ox$ and $O(\alpha) = \oy$,
\item $O(\alpha) \supseteq \semO(\alpha)$, and
\item $I(\alpha)$ determines $\semNoI{\alpha}$ on $O(\alpha)$.
\end{itemize}
\end{definition}
The first condition states that on atomic expressions (of which we know the inputs), 
$I$ and $O$ are defined correctly.  The next condition states $O$ approximates the 
semantic notion correctly.  We only allow for overapproximations; that is, false 
positives are allowed while false negatives are not. The reason for this is that 
falsely marking a variable as non-output while it is actually an output would mean 
incorrectly assuming the variable cannot change value.
The last condition establishes the relation between $I$ and $O$, and is called 
\emph{input--output determinacy}. It states that the inputs need to be large enough 
to determine the outputs, as such generalizing the defining condition of semantic 
inputs. Essentially, this means that whenever we overapproximate our outputs, we 
should also overapproximate our inputs to compensate for this; that correspondence 
is formalized in Lemma~\ref{lem:sound-input}.

We first remark that a proposition similar to Proposition~\ref{prop:determined1} 
can be made about sound output definitions.
\begin{proposition}\label{prop:sDetermined}
Let $(I,O)$ be a sound input-output definition, $\alpha$ be a LIF expression,
and $X$ a set of variables.  Then, $\semI(\alpha) \subseteq X$ if $X$ 
determines $\semNoI{\alpha}$ on $O(\alpha)$.
\end{proposition}
\begin{proof}
The proof follows directly from Proposition~\ref{prop:determined1} and knowing
that $O(\alpha) \supseteq \semO(\alpha)$.  Indeed, in general, if $X$ determines 
$\semNoI{\alpha}$ on some $Y$ and $Y \supseteq Z$, then clearly also $X$
determines $\semNoI{\alpha}$ on $Z$.
\end{proof}
This proposition along with the input-output determinacy condition imply 
a condition similar to the second condition about $I$:
\begin{proposition}\label{prop:inputApp}
Let $(I,O)$ be a sound input-output definition and $\alpha$ be a LIF expression. 
Then, $I(\alpha) \supseteq \semI(\alpha)$.
\end{proposition}
\begin{proof}
The proof follows from Proposition~\ref{prop:sDetermined} and 
knowing that $I(\alpha)$ determines $\semNoI{\alpha}$ on $O(\alpha)$.
\end{proof}
Besides requiring that our definitions to be sound, we will focus on definitions 
that are \emph{compositional}, in the sense that definitions of inputs and 
outputs of compound expressions can be given in terms of their direct 
subexpressions essentially treating subexpressions as black boxes. This means 
that the definition nicely follows the inductive definition of the syntax.
Formally:
\begin{definition}
Suppose $I$ and $O$ are functions from LIF expression to sets of variables.
We say that $(I,O)$ is \emph{compositional} if for all LIF expressions 
$\alpha_1$, $\alpha_2$, $\beta_1$, and $\beta_2$ with 
$I(\alpha_1) = I(\alpha_2)$, $O(\alpha_1) = O(\alpha_2)$, 
$I(\beta_1) = I(\beta_2)$, and $O(\beta_1) = O(\beta_2)$ the following hold:
\begin{itemize}
\item For every unary operator $\square$ \emph{:} 
$I(\square\alpha_1) = I(\square\alpha_2)$, and 
$O(\square\alpha_1) = O(\square\alpha_2)$; and
\item For every binary operator $\boxdot$\emph{:} 
$I(\alpha_1\boxdot \beta_1) = I(\alpha_2\boxdot \beta_2)$, and 
$O(\alpha_1\boxdot \beta_1) = O(\alpha_2\boxdot \beta_2)$.
\end{itemize}
\end{definition}
The previous definition essentially states that in order to be compositional, 
the inputs and outputs of $\alpha_1\boxdot \beta_1$ and $\square \alpha_1$ 
should only depend on the inputs and outputs of $\alpha_1$ and $\beta_1$, and 
not on their inner structure.

The following lemma rephrases input--output determinacy in terms of the inputs 
and outputs: in order to determine the  output-value of an inertial variable, we 
need to know its input-value.
\begin{lemma}\label{lem:sound-input}
Let $(I,O)$ be a sound input--output definition and let $\alpha$ be a
LIF expression.  If $\alpha$ is satisfiable, then
\[
O(\alpha)\setminus \semO(\alpha)\subseteq I(\alpha).
\] 
If $(I,O)$ is compositional, this holds for all $\alpha$.
\end{lemma}
\begin{proof}
Let $x \in O(\alpha) - \semO(\alpha)$.  For the sake of contradiction, assume that 
$x \not \in I(\alpha)$, so Proposition~\ref{prop:inertcyl} is applicable since 
$x \not \in \semI(\alpha)$ as we know by the soundness of $(I, O)$.  Hence, 
$\semNoI{\alpha}$ is inertially cylindrified on $\{x\}$.  We claim that this 
contradicts the fact that $(I,O)$ is a sound definition.  In particular, we can 
verify that $I(\alpha)$ can not determine $\semNoI{\alpha}$ on $O(\alpha)$ in
case $\alpha$ is satisfiable.  Let $\inst$ be an \interpretation and $\nu_1$ and 
$\nu_2$ be valuations such that $(\nu_1, \nu_2) \in \sem{\alpha}$.  We also know 
that $\nu_1(x) = \nu_2(x)$ since $\semNoI{\alpha}$ is inertially cylindrified on 
$\{x\}$.  Take $\nu'_1$ to be the valuation $\nu_1[\{x \mapsto a\}]$ where 
$a \neq \nu_1(x)$.  By determinacy, we know that there is a valuation $\nu'_2$ 
such that $(\nu'_1, \nu'_2) \in \sem{\alpha}$ and $\nu'_2 = \nu_2$ on $O(\alpha)$. 
Thus, $\nu'_2(x) = \nu_2(x) \neq a$ since $x \in O(\alpha)$.  On the other hand, 
$\nu'_2(x) = \nu'_1(x) = a$ since $\semNoI{\alpha}$ is inertially cylindrified on 
$\{x\}$.  Hence, a contradiction.

For the compositional case, we can always replace subexpressions by atomic 
expressions with the same inputs and outputs to ensure satisfiability.
It is clear that when $\alpha$ is an atomic module expression, it is always
satisfiable.  Now, consider any LIF expression $\alpha$ which is of the form 
$\square \alpha_1$ or $\alpha_1 \boxdot \alpha_2$, where $\square$ is any of 
the unary operators and $\boxdot$ is any of the binary operator.  Construct 
two atomic expressions $M_1$ and $M_2$ such that $I(M_i) = I(\alpha_i)$ 
and $O(M_i) = O(\alpha_i)$ for $i = 1,2$.
By compositionality, we know that $I(\square \alpha_1) = I(\square M_1)$ and 
$O(\square \alpha_1) = O(\square M_1)$ for any unary operator, while
$I(\alpha_1 \boxdot \alpha_2) = I(M_1 \boxdot M_2)$ and 
$O(\alpha_1 \boxdot \alpha_2) = O(M_1 \boxdot M_2)$ for any binary operator.
Next, we give examples for an \interpretation $\inst$ in which any 
$\sem{\square M_1}$ and $\sem{M_1 \boxdot M_2}$ can be shown not be empty, so 
$\square M_1$ and $M_1 \boxdot M_2$ are satisfiable expressions.

In what follows, let $\nu_1$ be the valuation that assigns $1$ to every variable.

\paragraph{Case $\boxdot$ is $\setminus$} Let $\inst$ be the \interpretation with
\begin{itemize}
    \item $\inst(M_1) = \{(1, \ldots, 1; 1, \ldots, 1)\}$; and
    \item $\inst(M_2) = \emptyset$.
\end{itemize}
It is clear that $(\nu_1, \nu_1) \in \sem{M_1}$, and 
$(\nu_1, \nu_1) \not \in \sem{M_2}$, whence, 
$(\nu_1, \nu_1) \in \sem{M_1 \setminus M_2}$.

\paragraph{All other cases} Let $\inst$ be the \interpretation with
\begin{itemize}
    \item $\inst(M_1) = \{(1, \ldots, 1; 1, \ldots, 1)\}$; and
    \item $\inst(M_2) = \{(1, \ldots, 1; 1, \ldots, 1)\}$.
\end{itemize}
We can see that $(\nu_1, \nu_1) \in \sem{M_1}$.  Consequently, 
$(\nu_1,\nu_1) \in \sem{\square M_1}$ for any unary operator $\square$.  
We can also see that $(\nu_1, \nu_1) \in \sem{M_2}$, whence, 
$(\nu_1, \nu_1) \in \sem{M_1 \boxdot M_2}$ for any binary operator 
$\boxdot \in \{\cup, \cap, ;\}$.
\end{proof}
We now provide a sound and compositional input--output definition. While 
the definition might seem complex, there is a good reason for the different 
cases.  Indeed, as we show below in Theorem~\ref{thm:optimality}, our definition 
is optimal among the sound and compositional definitions.
In the definition, the condition $x \eqv y$ simply means that $x$ and $y$ are 
the same variable and $\symdif{}$ denotes the symmetric difference of two sets.
\begin{definition}
The syntactic inputs and outputs of a LIF expression $\alpha$, denoted 
$\synI(\alpha)$ and $\synO(\alpha)$ respectively, are defined recursively as 
given in Table \ref{table:IO}.
\begin{table*}
\[
\begin{array}{l|l|l}
\toprule
\alpha & \synI(\alpha) & \synO(\alpha) \\ \toprule
\id & \emptyset & \emptyset \\ \midrule
M(\ox;\oy) 
& \{x_1,\ldots,x_n\} \text{ where }\ox=x_1,\ldots,x_n 
& \{y_1,\ldots,y_n\} \text{ where }\oy=y_1,\ldots,y_n \\ \midrule
\alpha_1 \cup \alpha_2 
& \synI(\alpha_1) \cup \synI(\alpha_2) 
  \cup (\synO(\alpha_1) \symdif{} \synO(\alpha_2)) 
& \synO(\alpha_1) \cup \synO(\alpha_2) \\ \midrule
\alpha_1 \cap \alpha_2 
& \synI(\alpha_1) \cup \synI(\alpha_2) 
  \cup (\synO(\alpha_1) \symdif{} \synO(\alpha_2)) 
& \synO(\alpha_1) \cap \synO(\alpha_2) \\ \midrule
\alpha_1 \setminus \alpha_2 
& \synI(\alpha_1) \cup \synI(\alpha_2) 
  \cup (\synO(\alpha_1) \symdif{} \synO(\alpha_2))
& \synO(\alpha_1) \\ \midrule
\alpha_1 \comp{} \alpha_2   
& \synI(\alpha_1) \cup (\synI(\alpha_2) - \synO(\alpha_1)) 
& \synO(\alpha_1) \cup \synO(\alpha_2) \\ \midrule
\withconverse{ \conv{\alpha_1} 
& \synO(\alpha_1) \cup \synI(\alpha_1) 
& \synO(\alpha_1) \\ } \midrule
\cyll{x}(\alpha_1)       
& \synI(\alpha_1)\setminus\{x\} 
& \synO(\alpha_1) \cup \{x\} \\ \midrule
\cylr{x}(\alpha_1)       
& \synI(\alpha_1) 
& \synO(\alpha_1) \cup \{x\} \\ \midrule
\sellr{x=y}(\alpha_1) 
& \begin{cases} 
    \synI(\alpha_1) & \text{if } x \eqv y \text{ and } y \not \in \synO(\alpha_1) \\ 
    \synI(\alpha_1) \cup \{x,y\} & \text{if } x \neqv y \text{ and } y \not \in \synO(\alpha_1) \\ \synI(\alpha_1) \cup \{x\} & \text{otherwise}
  \end{cases} 
& \begin{cases} 
    \synO(\alpha_1) \setminus \{x\} & \text{if } x \eqv y \\ 
    \synO(\alpha_1) & \text{otherwise}  
  \end{cases} \\ \midrule
\sell{x=y}(\alpha_1) 
& \begin{cases} 
    \synI(\alpha_1) & \text{if } x \eqv y  \\ 
    \synI(\alpha_1) \cup \{x,y\} & \text{otherwise}
  \end{cases} 
& \synO(\alpha_1) \\ \midrule
\selr{x=y}(\alpha_1) 
& \begin{cases} 
    \synI(\alpha_1) & \text{if } x \eqv y \\ 
    \synI(\alpha_1) \cup (\{x,y\} - \synO(\alpha_1)) & \text{otherwise}  
  \end{cases} 
& \synO(\alpha_1) \\
\bottomrule
\end{array}
\]
\caption{Syntactic inputs and outputs for LIF expressions.}
\label{table:IO}
\end{table*}
\label{def:IOSyn}
\end{definition}
While it would be tedious to discuss the motivation for all the cases of 
Definition~\ref{def:IOSyn} (their motivation will be clarified in Theorem 
\ref{thm:optimality}), we  discuss here one of the most difficult parts, namely 
the case where $\alpha = \sigma^{lr}_{x=y}(\alpha_1)$. For a given \interpretation 
$\inst$,
\[
\sem{\alpha} = \{(\nu_1,\nu_2) \in \sem{\alpha_1} \mid \nu_1(x)=\nu_2(y)\}.
\]
First, since $\sem{\alpha}\subseteq\sem{\alpha_1}$, it is clear that the outputs 
of $\alpha$ should be a subset of those of $\alpha_1$ (if $\alpha_1$ admits 
no pairs in its semantics that change the value of a variable, then neither does 
$\alpha$).  For the special case in which $x$ and $y$ are the same variable, this 
selection enforces $x$ to be inertial, i.e., it should not be an output of $\alpha$.

Secondly, all inputs of $\alpha_1$ remain inputs of $\alpha$.  Since we select those 
pairs whose $y$-value on the right equals the $x$-value on the left, clearly $x$ must 
be an input of $\alpha$ (the special case $x\eqv y$ and $y\not\in\synO(\alpha_1)$ only 
covers cases where $\alpha_1$ and $\alpha$ are actually equivalent).  Whether $y$ is an 
input depends on $\alpha_1$: if $y\not\in \synO(\alpha_1)$, $y$ is inertial. Since we 
compare the input-value of $x$ with the output-value of $y$, essentially this is the 
same as comparing the input-values of both variables, i.e., the value of $y$ on the 
input-side matters.  On the other hand, if $y\in \synO(\alpha_1)$, the value of $y$ 
can be changed by $\alpha_1$ and thus this selection does not force $y$ being an input.

Our syntactic definition is clearly compositional (since we only use the inputs and 
outputs of subexpressions). An important result is that our definition is also sound, 
i.e., our syntactic concepts are overapproximations of the semantic concepts.
\begin{theorem}[Soundness Theorem]\label{thm:soundComp}
$(\synI,\synO)$ is a sound input--output definition.
\end{theorem}
\begin{proof}
The proof is given in Section~\ref{app:soundness}.
\end{proof}
Of course, since the semantic notions of inputs and outputs are undecidable and our 
syntactic notions clearly are decidable, expressions exist in which the semantic and 
syntactic notions do not coincide. We give some examples.
\begin{example}
Consider the LIF expression
\[
\alpha := \sell{x=y}\selr{x=y}(R(x;y))
\]
In this case, $\semO(\alpha)=\emptyset$. However, it can be verified that 
$\synO(\alpha) = \{x,y\}$.
\end{example}
\begin{example}
Consider the LIF expression
\[
\alpha := \sellr{x=x}\cylr{x}\cyll{x} (P(x;)).
\]
Thus, we first cylindrify $x$ on both sides and afterwards only select those 
pairs that have inertia, therefore, we reach an expression $\alpha$ that is 
equivalent to $\id$. As such, $x$ is inertially cylindrified in $\alpha$ where 
$x\not\in\semO(\alpha)$ and $x\not\in\semI(\alpha)$.  
However, $\synI(\alpha) = \{x\}$. 
\end{example}
These examples suggest that our definition can be improved.
Indeed, one can probably keep coming up with ad-hoc but more
precise approximations of inputs and outputs for various specific
patterns of expressions.  Such improvements would not be
compositional, as they would be based on inspecting the structure of
subexpressions.  In the following results, we show that
$(\synI,\synO)$ is actually the most precise sound and
compositional input--output definition.
\begin{theorem}[Precision Theorem]\label{thm:precision}
Let $\alpha$ be a LIF expression that is either atomic, or a unary operator applied 
to an atomic module expression, or a binary operator applied to two atomic module 
expressions involving different module names. Then
\[
\semO(\alpha) = \synO(\alpha)\text{ and }\semI(\alpha) = \synI(\alpha).
\]
\end{theorem}
\begin{proof}
The proof is given in Section~\ref{app:precision}.
\end{proof}
Now, the precision theorem forms the basis for our main result on syntactic inputs and 
outputs, which states that Definition~\ref{def:IOSyn} yields the most precise sound and 
compositional input--output definition.
\begin{theorem}[Optimality Theorem]\label{thm:optimality}
Suppose $(I,O)$ is a sound and compositional input--output definition.  Then for each 
LIF expression $\alpha$:
\[ 
\synI(\alpha) \subseteq I(\alpha)\text{ and }\synO(\alpha) \subseteq O(\alpha).
\]
\end{theorem}
\begin{proof}
The proof is given in Section~\ref{app:optimality}.
\end{proof}

\section{Soundness Theorem Proof}\label{app:soundness}
In this section, we prove Theorem~\ref{thm:soundComp}.  Thereto, we need to 
verify its three conditions for every LIF expression $\alpha$ according to 
Definition~\ref{def:soundness}:
\begin{description}
\item[Atomic Module Case:]
If $\alpha = M(\ox;\oy)$, then $\synI(\alpha)=\ox$ and $\synO(\alpha) = \oy$.

This is clear from the definitions.
\item[Output Approximation:] 
$\synO(\alpha) \supseteq \semO(\alpha)$.

The output approximation property is shown in Proposition~\ref{prop:output}, 
which is given in Section~\ref{app:output}.
\item[Syntactic Input-Output Determinacy:] 
$\synI(\alpha)$ determines $\semNoI{\alpha}$ on $\synO(\alpha)$.

The syntactic input-output determinacy property is shown in 
Lemma~\ref{lem:inputDet}, which is given in Section~\ref{sub:inputDet}.  
First, however, in Section~\ref{sub:freeVars}, we need to prove a syntactic 
version of Property~\ref{prop:inertcyl}, which will be used in the proof of 
the syntactic input-output determinacy property.\qedhere
\end{description}

\subsection{Proof of Output Approximation}\label{app:output}
In this section, we prove:
\begin{proposition}\label{prop:output} Let $\alpha$ be a LIF expression. 
Then, $\semO(\alpha) \subseteq \synO(\alpha)$.
\end{proposition}
To prove this proposition, we introduce the following notion which is
related to Definition~\ref{def:inertialCyl}.
\begin{definition}
A $\bbar$ $B$ has inertia outside a set of variables $Z$ if for every 
$(\nu_1, \nu_2) \in B$, we have $\nu_1 = \nu_2$ outside $Z$.  A \gbrv $Q$ has 
inertia outside a set of variables $Z$ if $Q(\inst)$ has inertia outside $Z$ 
for every \interpretation $\inst$.
\end{definition}
Using this notion, Proposition~\ref{prop:output} can be equivalently 
formulated as follows.
\begin{proposition}[Inertia Property]\label{prop:inertiaProp}
Let $\alpha$ be a LIF expression. Then, $\semNoI{\alpha}$ has inertia outside 
$\synO(\alpha)$.
\end{proposition}
In the remainder of this section we prove the inertia property by structural 
induction on the shape of LIF expressions.  Also, we remove the superscript from 
$\synO$ and refer to it simply by $\Out$.

\subsubsection{Atomic Modules}
Let $\alpha$ be of the form $M(\bar{x}; \bar{y})$.  Recall that 
$\Out(\alpha) = Y$ where $Y$ is the set of variables in $\oy$.  The property 
directly follows from the definition of the semantics for atomic modules.

\subsubsection{Identity}
Let $\alpha$ be of the form $\id$. Recall that $\Out(\alpha) = \emptyset$. 
The property directly follows from the definition of $\id$.

\subsubsection{Union}
Let $\alpha$ be of the form $\alpha_1 \cup \alpha_2$.  Recall that 
$\Out(\alpha) = \outt1 \cup \outt2$.  If $(\nu_1, \nu_2) \in 
\sem{\alpha_1\cup\alpha_2}$, then at least one of the following holds:
\begin{enumerate}
	\item $(\nu_1, \nu_2) \in \sem{\alpha_1}$.  Then, by induction we know 
	that $\nu_1 = \nu_2$ outside $\Out{(\alpha_1)}$.
	Since $\Out(\alpha_1) \subseteq \Out(\alpha_1) \cup \Out(\alpha_2) = 
	\Out(\alpha)$, we know that $\nu_1 = \nu_2$ outside $\Out(\alpha)$.
	\item $(\nu_1, \nu_2) \in \sem{\alpha_2}$. Similar.
\end{enumerate}

\subsubsection{Intersection}
Let $\alpha$ be of the form $\alpha_1 \cap \alpha_2$.  Recall that 
$\Out(\alpha) = \outt1 \cap \outt2$.  If 
$(\nu_1, \nu_2) \in \sem{\alpha_1\cap \alpha_2}$, then 
$(\nu_1, \nu_2) \in \sem{\alpha_1}$ and $(\nu_1, \nu_2) \in \sem{\alpha_2}$.
By induction, $\nu_1 = \nu_2$ outside $\Out(\alpha_1)$ and also $\nu_1 = \nu_2$ 
outside $\Out(\alpha_2)$.
Hence, $\nu_1 = \nu_2$ outside $\Out(\alpha_1) \cap \Out(\alpha_2)$.

\subsubsection{Composition}
Let $\alpha$ be of the form $\alpha_1 \comp{} \alpha_2$.  Recall that 
$\Out(\alpha) = \Out(\alpha_1) \cup \Out(\alpha_2)$.  If 
$(\nu_1, \nu_2) \in \sem{\alpha_1\comp{}\alpha_2}$, then there exists a 
valuation $\nu$ such that $(\nu_1, \nu) \in \sem{\alpha_1}$ and 
$(\nu, \nu_2) \in \sem{\alpha_2}$.  By induction, $\nu_1 = \nu$ outside 
$\Out(\alpha_1)$ and also $\nu = \nu_2$ outside $\Out(\alpha_2)$.
Hence, $\nu_1 = \nu_2=\nu$ outside $\Out(\alpha_1) \cup \Out(\alpha_2)$.

\subsubsection{Difference}
Let $\alpha$ be of the form $\alpha_1 \setminus \alpha_2$.  Recall that 
$\Out(\alpha) = \outt1$.  The proof then follows immediately by induction.

\subsubsection{Converse}
Let $\alpha$ be of the form $ \conv{\alpha_1}$.  Recall that 
$\Out(\alpha) = \outt1$.  The proof is immediate by induction.

\subsubsection{Left and Right Selections}
Let $\alpha$ be of the form $\sell{x=y}(\alpha_1)$ or $\selr{x=y}(\alpha_1)$.  
Recall that $\Out(\alpha) = \Out(\alpha_1)$.  The proof is immediate by induction.

\subsubsection{Left-to-Right Selection}
Let $\alpha$ be of the form $\sellr{x=y}(\alpha_1)$.  Recall the definition:
\[
\Out(\alpha) = 
  \begin{cases} 
    \Out(\alpha_1) & \text{if } x \neqv y \\ 
    \Out(\alpha_1) \setminus \{x\} & \text{otherwise} 
  \end{cases}
\]
If $(\nu_1, \nu_2) \in \sem{\sellr{x=y}(\alpha_1)}$, then we know that:
\begin{enumerate}
	\item\label{innertia:sellr1} $(\nu_1, \nu_2) \in \sem{\alpha_1}$;
	\item\label{innertia:sellr2} $\nu_1(x) = \nu_2(y)$.
\end{enumerate}
By induction from (\ref{innertia:sellr1}), we know that $\nu_1 = \nu_2$ 
outside $\Out(\alpha_1)$. Hence, for $x \neqv y$ case we are done.  In the 
other case, we must additionally show that $\nu_1(x)=\nu_2(x)$, which follows 
from $(\ref{innertia:sellr2})$.

\subsubsection{Right and Left Cylindrifications}
Let $\alpha$ be of the form $\cylr{x}(\alpha_1)$.  The case for left 
cylindrification is analogous.  
Recall that $\Out(\alpha) = \outt1 \cup \{x\}$.  
If $(\nu_1, \nu_2) \in \sem{\cylr{x}(\alpha_1)}$, then there exists 
$\nu$ such that:
\begin{enumerate}
	\item \label{innertia:cylr1} $(\nu_1, \nu) \in \sem{\alpha_1}$;
	\item \label{innertia:cylr2} $\nu = \nu_2$ outside $\{x\}$.
\end{enumerate}
By induction from (\ref{innertia:cylr1}), we know that $\nu_1 = \nu$ 
outside $\Out(\alpha_1)$. Combining this with (\ref{innertia:cylr2}), 
we know that $\nu_1 = \nu_2$ outside $\Out(\alpha_1)\cup \{x\}$ as desired.

\subsection{Proof of Syntactic Free Variable Property}\label{sub:freeVars}
\begin{lemma}[Syntactic Free Variable Property]\label{lem:synFreeVarProp}
Let $\alpha$ be a LIF expression.  Denote $\synI(\alpha) \cup \synO(\alpha)$ by
$\fvars(\alpha)$.  Then, $\alpha$ is inertially cylindrified on 
$\V - \fvars(\alpha)$.
\end{lemma}
In the proof of this Lemma, we will often make use of the Lemma 
\ref{prop:inertcylsubset}.  In what follows, for a set of variables $X$, we define 
$\compl{X}$ to be $\V \setminus X$.  In the rest of the section, we remove the 
superscript from $\synO$ and refer to it simply by $\Out$.
\begin{lemma}\label{prop:inertcylsubset}
Let $B$ be a $\bbar{}$ that has inertia on $Y$.  Then, $B$ is inertially 
cylindrified on $Y$ if and only if $B$ is inertially cylindrified on every 
$X\subseteq Y$.
\end{lemma}
\begin{proof}
The `if'-direction is immediate. Let us now consider the `only if'.  To this end, 
suppose that $(\nu_1,\nu_2)\in B$ and that $\nu$ is a partial valuation on $X$. 
Extend $\nu$ to a valuation $\nu'$ by $\nu'=\nu_1$ on $Y\setminus X$.  Since $B$ 
has inertia on $Y$, we know that $\nu_1=\nu_2=\nu'$ on $Y\setminus X$.  Thus, 
$\nu_1[\nu']=\nu_1[\nu]$ and $\nu_2[\nu']=\nu_2[\nu]$.  The lemma now directly 
follows since $B$ is inertially cylindrified on $Y$.
\end{proof}
This Lemma is always applicable for any $\LIF$ expression $\alpha$ and 
$Y=\V\setminus \fvars(\alpha)$. Indeed, for every \interpretation $\inst{}$, 
we know by Proposition~\ref{prop:inertiaProp} that $\sem{\alpha}$ has inertia 
outside $\Out(\alpha)\subseteq\fvars(\alpha)$.

We are now ready to prove Lemma~\ref{lem:synFreeVarProp}.
\subsubsection{Atomic Modules}
Let $\alpha$ be of the form $M(\bar{x}; \bar{y})$.  Recall that 
$\fvars(\alpha) = X \cup Y$ where $X$ and $Y$ are the variables in 
$\ox$ and $\oy$, respectively. The property directly follows from the 
definition of the semantics for atomic modules.

\subsubsection{Identity}
Let $\alpha$ be of the form $\id$.  Recall that $\fvars(\alpha)= \emptyset$. 
The property directly follows from the definition of $\id$.

\subsubsection{Union}
Let $\alpha$ be of the form $\alpha_1\cup \alpha_2$.  Recall that 
$\fvars(\alpha) = \fvars(\alpha_1) \cup \fvars(\alpha_2)$.  If 
$(\nu_1, \nu_2) \in \sem{\alpha}$, then $(\nu_1, \nu_2) \in \sem{\alpha_1}$ 
or $(\nu_1, \nu_2) \in \sem{\alpha_2}$.  Let $Y = \V\setminus \fvars(\alpha)$ 
and let $\nu$ be a partial valuation on $Y$.  Assume without loss of generality 
that $(\nu_1,\nu_2)\in \sem{\alpha_1}$.  
Clearly, $Y\subseteq \V\setminus \fvars(\alpha_1)$ since 
$\fvars(\alpha_1)\subseteq \fvars(\alpha)$.  By induction and 
Lemma~\ref{prop:inertcylsubset}, we know that 
$(\nu_1[\nu],\nu_2[\nu]) \in \sem{\alpha_1}\subseteq \sem{\alpha}$ as desired.

\subsubsection{Intersection}
Let $\alpha$ be of the form $\alpha_1\cap \alpha_2$.  Recall that 
$\fvars(\alpha) = \fvars(\alpha_1) \cup \fvars(\alpha_2)$.  
If $(\nu_1, \nu_2) \in \sem{\alpha}$, then $(\nu_1, \nu_2) \in \sem{\alpha_1}$ 
and $(\nu_1, \nu_2) \in \sem{\alpha_2}$.
Let $Y = \V\setminus \fvars(\alpha)$ and let $\nu$ be a partial valuation on $Y$. 
Clearly, $Y\subseteq \V\setminus \fvars(\alpha_i)$ with $i=1,2$ since 
$\fvars(\alpha_i)\subseteq \fvars(\alpha)$. By induction and 
Lemma~\ref{prop:inertcylsubset}, we know that $(\nu_1[\nu],\nu_2[\nu]) \in \sem{\alpha_i}$ 
with $i=1,2$, whence $(\nu_1[\nu],\nu_2[\nu]) \in \sem{\alpha}$ as desired.

\subsubsection{Composition}
Let $\alpha$ be of the form $\alpha_1 \comp{} \alpha_2$.  Recall that 
$\fvars(\alpha) = \fvars(\alpha_1) \cup \fvars(\alpha_2)$.  If 
$(\nu_1, \nu_2) \in \sem{\alpha}$, then there exists a valuation $\nu_3$ 
such that $(\nu_1, \nu_3) \in \sem{\alpha_1}$ and $(\nu_3, \nu_2) \in \sem{\alpha_2}$.
Let $Y = \V\setminus \fvars(\alpha)$ and let $\nu$ be a partial valuation on $Y$.  
Clearly, $Y\subseteq \V\setminus \fvars(\alpha_i)$ with $i=1,2$ since 
$\fvars(\alpha_i)\subseteq \fvars(\alpha)$.  By induction and 
Lemma~\ref{prop:inertcylsubset}, we know that 
$(\nu_1[\nu],\nu_3[\nu]) \in \sem{\alpha_1}$ and 
$(\nu_3[\nu],\nu_2[\nu]) \in \sem{\alpha_2}$.  Therefore, we may conclude that 
$(\nu_1[\nu],\nu_2[\nu])\in \sem{\alpha_1\comp{}\alpha_2}$.

\subsubsection{Difference}
Let $\alpha$ be of the form $\alpha_1\setminus \alpha_2$.  Recall that 
$\fvars(\alpha) = \fvars(\alpha_1) \cup \fvars(\alpha_2)$.  If 
$(\nu_1, \nu_2) \in \sem{\alpha}$, then we know that:
\begin{enumerate}
\item \label{IC:diff1} $(\nu_1, \nu_2) \in \sem{\alpha_1}$.  By inertia, we know that 
$\nu_1=\nu_2$ outside $\Out(\alpha_1)\subseteq \fvars(\alpha_1)\subseteq\fvars(\alpha)$.
\item \label{IC:diff2} $(\nu_1, \nu_2) \not \in \sem{\alpha_2}$.
\end{enumerate}
Let $Y = \V\setminus \fvars(\alpha)$ and let $\nu$ be a partial valuation on $Y$. 
Clearly, $Y\subseteq \V\setminus \fvars(\alpha_i)$ with $i=1,2$ since 
$\fvars(\alpha_i)\subseteq \fvars(\alpha)$.  By induction from (\ref{IC:diff1}) 
and Lemma~\ref{prop:inertcylsubset}, we know that 
$(\nu_1[\nu], \nu_2[\nu]) \in \sem{\alpha_1}$.  All that remains is to show that 
$(\nu_1[\nu], \nu_2[\nu]) \not\in \sem{\alpha_2}$
Now, suppose for the sake of contradiction that 
$(\nu_1[\nu], \nu_2[\nu]) \in \sem{\alpha_2}$.  By induction and 
Lemma~\ref{prop:inertcylsubset}, we know that
$((\nu_1[\nu])[\restr{\nu_1}{Y}], (\nu_2[\nu])[\restr{\nu_1}{Y}])\in \sem{\alpha_2}$. 
Clearly, $(\nu_1[\nu])[\restr{\nu_1}{Y}]=\nu_1$.  Moreover, 
$(\nu_2[\nu])[\restr{\nu_1}{Y}] = \nu_2$ since $\nu_1=\nu_2$ outside $\fvars(\alpha)$ 
by $(\ref{IC:diff1})$.  We have thus obtained that $(\nu_1,\nu_2)\in \sem{\alpha_2}$, 
which contradicts $(\ref{IC:diff2})$.

\subsubsection{Converse}
Let $\alpha$ be of the form $\conv{\alpha_1}$.  Recall that 
$\fvars(\alpha) = \fvars(\alpha_1)$.  The property follows directly by induction 
since $\fvars(\alpha_1) = \fvars(\conv{\alpha_1})$.

\subsubsection{Left Selection}
Let $\alpha$ be of the form $\sell{x=y}(\alpha_1)$.  Recall the definition:
\[
\fvars(\alpha) = 
  \begin{cases} 
    \fvars(\alpha_1) & \text{if } x \eqv y \\ 
    \fvars(\alpha_1) \cup \{x,y\} & \text{otherwise} 
  \end{cases}
\]
In case of $x \eqv y$, clearly $\sem{\sell{x = y}(\alpha_1)} = \sem{\alpha_1}$. 
The property holds trivially by induction. In the other case when $x \neqv y$, 
if $(\nu_1, \nu_2) \in \sem{\alpha}$, then $(\nu_1, \nu_2) \in \sem{\alpha_1}$. 
Let $Y = \V\setminus \fvars(\alpha)$ and let $\nu$ be a partial valuation on $Y$.
Clearly, $Y\subseteq \V\setminus \fvars(\alpha_1)$ since 
$\fvars(\alpha_1)\subseteq \fvars(\alpha)$.  By induction and 
Lemma~\ref{prop:inertcylsubset}, we know that 
$(\nu_1[\nu], \nu_2[\nu]) \in \sem{\alpha_1}$.  Moreover, since 
$\{x,y\}\cap Y=\emptyset$, we know that the selection does not look at $\nu$, 
whence $(\nu_1[\nu], \nu_2[\nu]) \in  \sem{\sell{x=y}(\alpha_1)}$ as desired.

\subsubsection{Right Selection}
Let $\alpha$ be of the form $\selr{x=y}(\alpha_1)$.  Recall the definition:
\[
\fvars(\alpha) = 
  \begin{cases} 
    \fvars(\alpha_1) & \text{if } x \eqv y \\ 
    \fvars(\alpha_1) \cup \{x,y\} & \text{otherwise} 
  \end{cases}
\]
In case of $x \eqv y$, clearly $\sem{\selr{x = y}(\alpha_1)} = \sem{\alpha_1}$. 
Hence, the property holds trivially by induction.  In the other case, it is 
analogous to $\sell{x=y}(\alpha_1)$ since here also 
$\{x,y\} \cap (\V\setminus \fvars(\alpha))=\emptyset$.

\subsubsection{Left-to-Right Selection}
Let $\alpha$ be of the form $\sellr{x=y}(\alpha_1)$.  Recall the definition:
\[
\fvars(\alpha) = 
  \begin{cases} 
    \fvars(\alpha_1) & \text{if } x \eqv y \text{ and } y \not \in \Out(\alpha_1) \\ 
    \fvars(\alpha_1) \cup \{x,y\} & \text{otherwise} 
  \end{cases}
\]
In case of $x \eqv y$ and $y \not \in \Out(\alpha_1) $, clearly 
$\sem{\sellr{x=y}(\alpha_1)} = \sem{\alpha_1}$.  Hence, the property holds trivially 
by induction.  In the other case, it is analogous to $\sell{x=y}(\alpha_1)$ since 
here also $\{x,y\}\cap (\V\setminus \fvars(\alpha))=\emptyset$.

\subsubsection{Right and Left Cylindrifications}
Let $\alpha$ be of the form $\cylr{x}(\alpha_1)$.  The case for left cylindrification 
is analogous.  Recall that $\fvars(\alpha) = \fvars(\alpha_1) \cup \{y\}$.  
If $(\nu_1, \nu_2) \in \sem{\alpha}$, then there exists a valuation $\nu_3$ 
such that:
\begin{enumerate}
	\item \label{IC:cylr1} $(\nu_1, \nu_3) \in \sem{\alpha_1}$;
	\item \label{IC:cylr2} $\nu_3 = \nu_2$ outside $\{x\}$.
\end{enumerate}
Let $Y = \V\setminus \fvars(\alpha)$ and let $\nu$ be a partial valuation on $Y$. 
Clearly, $Y\subseteq \V\setminus \fvars(\alpha_1)$ since 
$\fvars(\alpha_1)\subseteq \fvars(\alpha)$.  By induction from (\ref{IC:cylr1}) and 
Lemma~\ref{prop:inertcylsubset} we know that 
$(\nu_1[\nu], \nu_3[\nu])\in\sem{\alpha_1}$.  Since $x \not \in Y$, we know from 
$(\ref{innertia:cylr2})$ that $\nu_3[\nu]= \nu_2[\nu]$ outside $\{x\}$.  Hence, 
we can conclude that $(\nu_1[\nu], \nu_2[\nu])\in\sem{\alpha}$.

\subsection{Proof of Syntactic Input-Output Determinacy}\label{sub:inputDet}
Syntactic Input-Output determinacy is certainly proved if we can prove the 
following Lemma.
\begin{lemma}[Syntactic Input-Output Determinacy]\label{lem:inputDet}
Let $\alpha$ be a LIF expression. Then, for every \interpretation $\inst$, every 
$(\nu_1, \nu_2) \in \sem{\alpha}$ and every $\nu_1'$ that agrees with $\nu_1$ on 
$\synI(\alpha)$, there exists a valuation $\nu_2'$ that agrees with $\nu_2$ on 
$\synO(\alpha)$ and $(\nu_1', \nu_2') \in \sem{\alpha}$.
\end{lemma}
In the proof, we will make use of a useful alternative formulation of syntactic 
input-output determinacy which is defined next.
\begin{definition}[Alternative Input-Output Determinacy]\label{def:altInDet}
A LIF expression $\alpha$ is said to satisfy \emph{alternative input-output determinacy} 
if for every \interpretation $\inst$, every $(\nu_1, \nu_2) \in \sem{\alpha}$ and every 
$\nu_1'$ that agrees with $\nu_1$ on $\synI(\alpha)$ and outside $\synO(\alpha)$, we have 
$(\nu_1', \nu_2) \in \sem{\alpha}$.
\end{definition}
The following Lemma shows that the two definitions are equivalent.  In what 
follows, will remove the superscript from $\synI$ and $\synO$ and refer to 
them as $\In$ and $\Out$, respectively. 
\begin{lemma}\label{lem:eqvInDetDef}
Let $\alpha$ be a LIF expression.  $\alpha$ satisfies 
\emph{alternative input-output determinacy} iff it satisfies 
\emph{syntactic input-output determinacy}.
\end{lemma}
\begin{proof}
The proof of the ``if''-direction is similar to the proof of 
Proposition~\ref{prop:semDetToAD}.
% We proceed to verify the "if"-direction.  Suppose that 
% \[
% (\nu_1, \nu_2) \in \sem{\alpha} \text{ with } \nu'_1 = \nu_1 \text{ on } \In(\alpha) 
% \text{ and outside } \Out(\alpha) \tag{$*$}
% \] 
% Hence, by inertia, we know that 
% \[ 
% \nu_1 = \nu_2 \text{ outside } \Out(\alpha). \tag{$**$}
% \]
% By syntactic input-output determinacy, there exists $\nu_2'$ such that:
% \begin{enumerate}
% 	\item\label{apidtopid1} $\nu_2' = \nu_2$ on $\Out(\alpha)$;
% 	\item\label{apidtopid2} $(\nu_1', \nu_2') \in \sem{\alpha}$.
% \end{enumerate}
% Again by inertia, we have 
% \[
% \nu'_1 = \nu'_2 \text{ outside } \Out(\alpha) \tag{$***$}
% \]
% Together ($*$), ($**$), and ($***$) imply that $\nu'_2 = \nu_2$ outside 
% $\Out(\alpha)$.  Combining this with (\ref{apidtopid1}), we know that 
% $\nu_2=\nu_2'$, whence $(\nu_1',\nu_2)\in \sem{\alpha}$.
Now, we proceed to verify the other direction.  Suppose that 
$(\nu_1, \nu_2) \in \sem{\alpha}$ and $\nu'_1  = \nu_1$ on $\In(\alpha)$. 
We now construct a new valuation $\nu''_1$ such that it agrees with $\nu_1'$ 
on $\fvars(\alpha)$ and it agrees with $\nu_1$ elsewhere.  We thus have the 
following properties for $\nu_1''$:
\begin{enumerate}
	\item\label{pidtoapid1} $\nu_1'' = \nu_1'$ on $\fvars(\alpha)$, and
	\item\label{pidtoapid2} $\nu_1'' = \nu_1$ on $\V\setminus \fvars(\alpha)$.
\end{enumerate}
We know that $\nu'_1 = \nu_1$ on $\In(\alpha)$ by assumption, whence 
(\ref{pidtoapid1}) implies that $\nu''_1 = \nu_1$ on $\In(\alpha)$ since 
$\In(\alpha) \subseteq \fvars(\alpha)$.  Combining this with (\ref{pidtoapid2}), 
we know that $\nu_1''= \nu_1$ on $\In(\alpha)$ and outside $\fvars(\alpha)$. 
Thus, alternative input-output determinacy implies that 
$(\nu''_1, \nu_2) \in \sem{\alpha}$.  Since $\nu_1''=\nu_1'$ on $\fvars(\alpha)$, 
we know that there is a partial valuation $\nu$ on $\V\setminus \fvars(\alpha)$ 
such that $\nu_1''[\nu]=\nu_1'$.  By syntactic \ficp{}, we know that 
$(\nu_1''[\nu],\nu_2[\nu]) \in \sem{\alpha}$. 
Thus, $(\nu_1',\nu_2[\nu]) \in \sem{\alpha}$ as desired.
\end{proof}
We are now ready for the proof of Lemma~\ref{lem:inputDet}.  In the following
proof, we will use the notation $\compl{X}$ to mean $\V - X$.  Moreover, since 
we established by Lemma \ref{lem:eqvInDetDef} that the two definitions for 
input-output determinacy are equivalent, we will verify any of them for each 
LIF expression.

% \paragraph{Atomic Modules}
\subsubsection{Atomic Modules}
Let $\alpha$ be of the form $M(\bar{x}; \bar{y})$. Recall the definitions:
\begin{itemize}
	\item $\In(\alpha) = X$ where $X$ are the variables in $\ox$.
	\item $\Out(\alpha) = Y$ where $Y$ are the variables in $\oy$.
\end{itemize}
Syntactic input-output determinacy directly follows from the definition of the 
semantics for atomic modules.

\subsubsection{Identity}
Let $\alpha$ be of the form $\id$.  Recall that the definition for 
$\In(\alpha) = \Out(\alpha) = \emptyset$.  We proceed to verify that $\alpha$ 
satisfies alternative input-output determinacy.  Indeed, this is true since 
$\compl{\Out(\alpha)} \cup \In(\alpha)=\V$.

\subsubsection{Union}
Let $\alpha$ be of the form $\alpha_1\cup \alpha_2$.  Recall the definitions:
\begin{itemize}
	\item $\In(\alpha) = \In(\alpha_1) \cup \In(\alpha_2) 
	\cup (\Out(\alpha_1) \symdif \Out(\alpha_2))$.
	\item $\Out(\alpha) = \Out(\alpha_1) \cup \Out(\alpha_2)$.
\end{itemize}
We proceed to verify that $\alpha$ satisfies alternative input-output 
determinacy.  If $(\nu_1, \nu_2) \in \sem{\alpha}$, then 
$(\nu_1, \nu_2) \in \sem{\alpha_1}$ or $(\nu_1, \nu_2) \in \sem{\alpha_2}$. 
Assume without loss of generality that $(\nu_1, \nu_2) \in \sem{\alpha_1}$. 
Now, let $\nu'_1$ be a valuation such that $\nu_1' = \nu_1$ on 
$\compl{\Out(\alpha)} \cup \In(\alpha)$.  Moreover, we have the following:
\begin{align*}
\compl{\Out(\alpha)} \cup \In(\alpha) 
&= \compl{\Out(\alpha_1) \cup \Out(\alpha_2)} \cup \In(\alpha_1) \cup \In(\alpha_2) 
\cup (\Out(\alpha_1) \symdif \Out(\alpha_2)) \\
&= \compl{\Out(\alpha_1) \cap \Out(\alpha_2)} \cup \In(\alpha_1) \cup \In(\alpha_2)\\
&= \compl{\Out(\alpha_1)} \cup \compl{\Out(\alpha_2)} \cup \In(\alpha_1) \cup \In(\alpha_2)
\end{align*}
Therefore, certainly $\nu_1' = \nu_1$ on $\compl{\Out(\alpha_1)}\cup \In(\alpha_1)$. 
Thus, $(\nu_1', \nu_2) \in \sem{\alpha_1}$ by induction and Lemma~\ref{lem:eqvInDetDef}, 
whence $(\nu_1', \nu_2) \in \sem{\alpha}$ as desired.

\subsubsection{Intersection}
Let $\alpha$ be of the form $\alpha_1\cap \alpha_2$. Recall the definitions:
\begin{itemize}
	\item $\In(\alpha) = \In(\alpha_1) \cup \In(\alpha_2) \cup 
	(\Out(\alpha_1) \symdif \Out(\alpha_2))$.
	\item $\Out(\alpha) = \Out(\alpha_1) \cap \Out(\alpha_2)$.
\end{itemize}
We proceed to verify that $\alpha$ satisfies alternative input-output determinacy. 
If $(\nu_1, \nu_2) \in \sem{\alpha}$, then $(\nu_1, \nu_2) \in \sem{\alpha_1}$ and 
$(\nu_1, \nu_2) \in \sem{\alpha_2}$.  Now, let $\nu'$ be a valuation such that 
$\nu_1' = \nu_1$ on $\compl{\Out(\alpha)} \cup \In(\alpha)$.  Just as in the case for 
$\cup$, we have that $\compl{\Out(\alpha)} \cup \In(\alpha) =
\compl{\Out(\alpha_1)} \cup \compl{\Out(\alpha_2)} \cup \In(\alpha_1) \cup \In(\alpha_2)$.  
Therefore, certainly $\nu_1' = \nu_1$ on $\compl{\Out(\alpha_1)}\cup \In(\alpha_1)$.  
Thus, $(\nu_1', \nu_2) \in \sem{\alpha_1}$ and $(\nu_1', \nu_2) \in \sem{\alpha_2}$ 
by induction and Lemma~\ref{lem:eqvInDetDef}, whence $(\nu_1', \nu_2) \in \sem{\alpha}$ 
as desired.

\subsubsection{Composition}
Let $\alpha$ be of the form $\alpha_1 \comp{} \alpha_2$.  Recall the definitions:
\begin{itemize}
	\item $\In(\alpha) = \In(\alpha_1) \cup (\In(\alpha_2) \setminus \Out(\alpha_1))$.
	\item $\Out(\alpha) = \Out(\alpha_1) \cup \Out(\alpha_2)$.
\end{itemize}
We proceed to verify that $\alpha$ satisfies syntactic input-output determinacy. 
If $(\nu_1,\nu_2)\in\sem{\alpha}$, then there exists a valuation $\nu$ such that
\begin{enumerate}[label=(\roman*)]
	\item\label{pid:compi} $(\nu_1, \nu) \in \sem{\alpha_1}$;
	\item\label{pid:compii} $(\nu, \nu_2) \in \sem{\alpha_2}$.
\end{enumerate}
Now, let $\nu'_1$ be a valuation such that
\begin{equation}
\nu'_1 = \nu_1 \text{ on } \In(\alpha) = \In(\alpha_1) \cup 
(\In(\alpha_2) \setminus \Out(\alpha_1)) \label{pid:comp1}    
\end{equation}
Since $\In(\alpha_1) \subseteq \In(\alpha)$, then by induction there exists a 
valuation $\nu'$ such that:
\begin{enumerate}[label=(\roman*)]
	\setcounter{enumi}{2}
	\item\label{pid:compa} $(\nu'_1, \nu') \in \sem{\alpha_1}$;
	\item\label{pid:compb} $\nu' = \nu$ on $\Out(\alpha_1)$.
\end{enumerate}
By applying inertia to \ref{pid:compi} and \ref{pid:compa} we get that 
$\nu_1 = \nu$ and $\nu_1' = \nu'$ outside $\Out(\alpha_1)$.  Combining this with
(\ref{pid:comp1}) we have that $\nu' = \nu_1'= \nu_1 = \nu$ on 
$\In(\alpha) \cap \compl{\Out(\alpha_1)} = (\In(\alpha_1) \cup 
\In(\alpha_2)) \setminus \Out(\alpha_1)$.  Together with \ref{pid:compb}, this 
implies that
\begin{equation}
\nu'=\nu \text{ on } \In(\alpha_1) \cup \In(\alpha_2) \cup \Out(\alpha_1)  
\label{pid:comp2} 
\end{equation}
By induction from \ref{pid:compii}, there exists $\nu'_2$ such that:
\begin{enumerate}[label=(\roman*)]
	\setcounter{enumi}{4}
	\item\label{pid:compc} $(\nu', \nu'_2) \in \sem{\alpha_2}$;
	\item\label{pid:compd} $\nu'_2 = \nu_2$ on $\Out(\alpha_2)$.
\end{enumerate}
From \ref{pid:compa} and \ref{pid:compd} we get that 
$(\nu'_1, \nu'_2) \in \sem{\alpha}$.  All that remains to be shown is that 
$\nu'_2 = \nu_2$ on $\Out(\alpha)$.
By applying inertia to \ref{pid:compii} and \ref{pid:compc} we get that:
\begin{align*}
\nu &= \nu_2 \text{ outside } {\Out(\alpha_2)} \\
\nu' &= \nu_2' \text{ outside } {\Out(\alpha_2)}
\end{align*}
Combining this with \ref{pid:comp2} we have that $\nu'_2 = \nu_2$ on 
$(\In(\alpha_1) \cup \In(\alpha_2) \cup \Out(\alpha_1)) \cap \compl{\Out(\alpha_2)} = 
(\In(\alpha_1) \cup \In(\alpha_2) \cup \Out(\alpha_1)) \setminus \Out(\alpha_2)$. 
Together with \ref{pid:compd} this implies that $\nu'_2 = \nu_2$ on 
$\In(\alpha_1) \cup \In(\alpha_2) \cup \Out(\alpha_1) \cup \Out(\alpha_2)$, whence 
$\nu'_2 = \nu_2$ on $\Out(\alpha)$ as desired.

\subsubsection{Difference}
Let $\alpha$ be of the form $\alpha_1\setminus \alpha_2$.  Recall that:
\begin{itemize}
	\item $\In(\alpha) = \In(\alpha_1) \cup \In(\alpha_2) \cup 
	(\Out(\alpha_1) \symdif \Out(\alpha_2))$.
	\item $\Out(\alpha) = \Out(\alpha_1)$.
\end{itemize}
We proceed to verify that $\alpha$ satisfies alternative input-output determinacy. 
If $(\nu_1, \nu_2) \in \sem{\alpha}$, then we know that:
\begin{enumerate}
	\item\label{pid:diff1} $(\nu_1, \nu_2) \in \sem{\alpha_1}$;
	\item\label{pid:diff2} $(\nu_1, \nu_2) \not \in \sem{\alpha_2}$.
\end{enumerate}
Now, let $\nu'_1$ be a valuation such that $\nu_1' = \nu_1$ on 
$\compl{\Out(\alpha)} \cup \In(\alpha)$.  Since 
$\compl{\Out(\alpha_1)} \subseteq \compl{\Out(\alpha)}$ and 
$\In(\alpha_1) \subseteq \In(\alpha)$, then $\nu_1' = \nu_1$ on 
$\compl{\Out(\alpha_1)} \cup \In(\alpha_1)$.  Thus, by induction from (\ref{pid:diff1}) 
and Lemma \ref{lem:eqvInDetDef}, we know that $(\nu_1', \nu_2) \in \sem{\alpha_1}$.

To prove that $(\nu_1', \nu_2)\in \sem{\alpha}$, all that remains is to show that 
$(\nu_1',\nu_2) \not\in \sem{\alpha_2}$. Assume for the sake of contradiction that 
$(\nu_1',\nu_2) \in \sem{\alpha_2}$. Since $\nu_1' = \nu_1$ on 
$\compl{\Out(\alpha)} \cup \In(\alpha)$ and 
$\compl{\Out(\alpha)} \cup \In(\alpha) = 
\compl{\Out(\alpha_1)} \cup \compl{\Out(\alpha_2)} \cup \In(\alpha_1) \cup \In(\alpha_2)$, 
we know that $\nu_1' = \nu_1$ on $\compl{\Out(\alpha_2)} \cup \In(\alpha_2)$.
Hence, $(\nu_1,\nu_2) \in \sem{\alpha_2}$ by induction and Lemma \ref{lem:eqvInDetDef}, 
which contradicts (\ref{pid:diff2}).

\subsubsection{Converse}
Let $\alpha$ be of the form $\conv{\alpha_1}$.  Recall the definitions:
\begin{itemize}
	\item $\In(\alpha) = \In(\alpha_1) \cup \Out(\alpha_1)$.
	\item $\Out(\alpha) = \Out(\alpha_1)$.
\end{itemize}
Alternative input-output determinacy holds since 
$\compl{\Out(\alpha)} \cup \In(\alpha) = \V$.

\subsubsection{Left Selection}
Let $\alpha$ be of the form $\sell{x=y}(\alpha_1)$.  Recall the definitions:
\begin{itemize}
	\item $\In(\alpha) = \begin{cases} \In(\alpha_1) & x \eqv y \\ 
	\In(\alpha_1) \cup \{x,y\} & \text{otherwise}  \end{cases}$
	\item $\Out(\alpha) = \Out(\alpha_1)$
\end{itemize}
We proceed to verify that $\alpha$ satisfies alternative input-output determinacy. 
Clearly, the property holds trivially by induction in case of $x \eqv y$.  Indeed, 
in this case, $\sem{\alpha} = \sem{\alpha_1}$. In the other case when $x \neqv y$, 
if $(\nu_1, \nu_2) \in \sem{\alpha}$, then we know that:
\begin{enumerate}
	\item\label{pid:sell1} $(\nu_1, \nu_2) \in \sem{\alpha_1}$;
	\item\label{pid:sell2} $\nu_1(x) = \nu_1(y)$;
\end{enumerate}
Let $\nu_1'$ be a valuation such that $\nu_1' = \nu_1$ on 
$\compl{\Out(\alpha)} \cup \In(\alpha)$.  In all cases, 
$\Out(\alpha) \subseteq \Out(\alpha_1)$.  Hence, 
$\compl{\Out(\alpha_1)} \subseteq \compl{\Out(\alpha)}$.  Moreover, in all cases 
$\In(\alpha_1) \subseteq \In(\alpha)$.  Thus, $\nu_1' = \nu_1$ on 
$\compl{\Out(\alpha_1)} \cup \In(\alpha_1)$.  By induction from (\ref{pid:sell1}) 
and Lemma \ref{lem:eqvInDetDef}, we know that $(\nu_1', \nu_2) \in \sem{\alpha_1}$.

All that remains to be shown is that $\nu_1'(x) = \nu_1'(y)$.  Since 
$\{x,y\}\subseteq \In(\alpha)$, we know that $\nu_1'=\nu_1$ on $\{x,y\}$ by assumption. 
Hence, $\nu_1'(x) = \nu_1'(y)$ by (\ref{pid:sell2}).

\subsubsection{Right Selection}
Let $\alpha$ be of the form $\selr{x=y}(\alpha_1)$.  Recall the definitions:
\begin{itemize}
	\item $\In(\alpha) = \begin{cases} \In(\alpha_1) & x \eqv y \\ 
	\In(\alpha_1) \cup (\{x,y\} - \Out(\alpha_1)) & \text{otherwise}  \end{cases}$
	\item $\Out(\alpha) = \Out(\alpha_1)$
\end{itemize}
We proceed to verify that $\alpha$ satisfies alternative input-output determinacy. 
Clearly, the property holds trivially by induction in case of $x \eqv y$.  Indeed, 
in this case, $\sem{\alpha} = \sem{\alpha_1}$. In the other case when $x \neqv y$, 
if $(\nu_1, \nu_2) \in \sem{\alpha}$, then we know that:
\begin{enumerate}
	\item\label{pid:selr1} $(\nu_1, \nu_2) \in \sem{\alpha_1}$;
	\item\label{pid:selr2} $\nu_2(x) = \nu_2(y)$;
\end{enumerate}
Let $\nu_1'$ be a valuation such that $\nu_1' = \nu_1$ on 
$\compl{\Out(\alpha)} \cup \In(\alpha)$.  In all cases, 
$\Out(\alpha) \subseteq \Out(\alpha_1)$. Hence, 
$\compl{\Out(\alpha_1)} \subseteq \compl{\Out(\alpha)}$.  Moreover, in all cases 
$\In(\alpha_1) \subseteq \In(\alpha)$.  Thus, $\nu_1' = \nu_1$ on 
$\compl{\Out(\alpha_1)} \cup \In(\alpha_1)$.  By induction from (\ref{pid:selr1}) 
and Lemma~\ref{lem:eqvInDetDef}, we know that $(\nu_1', \nu_2) \in \sem{\alpha_1}$. 
Together with (\ref{pid:selr2}), we know that $(\nu_1', \nu_2) \in \sem{\alpha}$.

\subsubsection{Left-to-Right Selection}
Let $\alpha$ be of the form $\sellr{x=y}(\alpha_1)$.  Recall the definitions:
\begin{itemize}
	\item $\In(\alpha) = 
	\begin{cases} \In(\alpha_1) & x \eqv y \text{ and } y \not \in \Out(\alpha_1) \\ 
	\In(\alpha_1) \cup \{x,y\} & x \neqv y \text{ and } y \not \in \Out(\alpha_1) \\ 
	\In(\alpha_1) \cup \{x\} & \text{otherwise}  \end{cases}$
	\item $\Out(\alpha) = \begin{cases} \Out(\alpha_1) \setminus \{x\} & x \eqv y \\ 
	\Out(\alpha_1) & \text{otherwise}  \end{cases}$
\end{itemize}
We proceed to verify that $\alpha$ satisfies alternative input-output determinacy. 
Clearly, $\sem{\alpha} = \sem{\alpha_1}$ in case of $x \eqv y$ and $y \not \in \Out(\alpha_1)$. 
Hence, the property holds trivially by induction.  In the other cases, if 
$(\nu_1, \nu_2) \in \sem{\alpha}$, then we know that:
\begin{enumerate}
	\item\label{pid:sellr1} $(\nu_1, \nu_2) \in \sem{\alpha_1}$;
	\item\label{pid:sellr2} $\nu_1(x) = \nu_2(y)$.
\end{enumerate}
Let $\nu_1'$ be a valuation such that $\nu_1' = \nu_1$ on 
$\compl{\Out(\alpha)} \cup \In(\alpha)$. In all cases, $\Out(\alpha) \subseteq \Out(\alpha_1)$. 
Hence, $\compl{\Out(\alpha_1)} \subseteq \compl{\Out(\alpha)}$.  Moreover, in all cases 
$\In(\alpha_1) \subseteq \In(\alpha)$.  Thus, $\nu_1' = \nu_1$ on 
$\compl{\Out(\alpha_1)} \cup \In(\alpha_1)$.  By induction from (\ref{pid:sellr1}) and 
Lemma~\ref{lem:eqvInDetDef}, we know that $(\nu_1', \nu_2) \in \sem{\alpha_1}$.
All that remains to be shown is that $\nu_1'(x)=\nu_2(y)$.  Since $x \in \In(\alpha)$ and 
$\nu_1' = \nu_1$ on $\compl{\Out(\alpha)} \cup \In(\alpha)$, we have $\nu_1'(x) = \nu_1(x)$. 
Together with (\ref{pid:sellr2}), we get that $\nu_1'(x) = \nu_2(y)$ as desired, 
whence $(\nu_1', \nu_2) \in \sem{\alpha}$.

\subsubsection{Right Cylindrification}
Let $\alpha$ be of the form $\cylr{x}(\alpha_1)$.  Recall the definitions:
\begin{itemize}
	\item $\In(\alpha) = \In(\alpha_1)$.
	\item $\Out(\alpha) = \Out(\alpha_1) \cup \{x\}$.
	\item $\fvars(\alpha) = \fvars(\alpha_1) \cup \{x\}$.
\end{itemize}
We proceed to verify that $\alpha$ satisfies alternative input-output determinacy. 
If $(\nu_1, \nu_2) \in \sem{\alpha}$, then there exists a valuation $\nu'_2$ such that:
\begin{enumerate}
	\item\label{pid:cylr1} $(\nu_1, \nu'_2) \in \sem{\alpha_1}$;
	\item\label{pid:cylr2} $\nu'_2 = \nu_2$ outside $\{x\}$.
\end{enumerate}
Now, let $\nu_1'$ be a valuation such that $\nu_1' = \nu_1$ on 
$\compl{\Out(\alpha)} \cup \In(\alpha)$.  We now split the proof in two cases:
\begin{itemize}
	\item Suppose that $x \in \fvars(\alpha_1)$.  Then, $\fvars(\alpha)=\fvars(\alpha_1)$. 
	Thus, we know that $\nu_1' = \nu_1$ on $\overline{\fvars(\alpha_1)} \cup \In(\alpha_1)$,
	whence $\nu_1' = \nu_1$ on $\overline{\Out(\alpha_1)} \cup \In(\alpha_1)$.  Thus, 
	by induction from (\ref{pid:cylr1}) and Lemma~\ref{lem:eqvInDetDef}, we know that 
	$(\nu_1', \nu_2') \in \sem{\alpha_1}$.  Hence, $(\nu_1', \nu_2) \in \sem{\alpha}$.

	\item Conversely, suppose that $x \not \in \fvars(\alpha_1)$. We have 
	$\overline{\Out(\alpha)} \cup I(\alpha) = 
	(\overline{\Out(\alpha_1)} \cup \In(\alpha_1)) \setminus \{x\}$.  Thus, 
	$\nu_1'[\restr{\nu_1}{\{x\}}] = \nu_1$ on $\overline{\Out(\alpha_1)} \cup \In(\alpha_1)$ 
	since $\nu_1' = \nu_1$ on $\compl{\Out(\alpha)} \cup \In(\alpha)$.  By induction and 
	Lemma~\ref{lem:eqvInDetDef}, then $(\nu_1'[\restr{\nu_1}{\{x\}}], \nu'_2) \in \sem{\alpha_1}$. 
	By syntactic \ficp{}, we know that 
	$(\nu_1'[\restr{\nu_1}{\{x\}}][\restr{\nu_1'}{\{x\}}], \nu'_2[\restr{\nu_1'}{\{x\}}]) 
	\in \sem{\alpha_1}$ since $x \not \in \fvars(\alpha_1)$. 
	Clearly, $\nu_1'[\restr{\nu_1}{\{x\}}][\restr{\nu_1'}{\{x\}}]=\nu_1'$, whence 
	$(\nu_1', \nu'_2[\restr{\nu_1'}{\{x\}}]) \in \sem{\alpha_1}$.  Consequently, 
	$(\nu_1',\nu_2) \in \sem{\alpha}$ as desired.
\end{itemize}

\subsubsection{Left Cylindrification}
Let $\alpha$ be of the form $\cyll{x}(\alpha_1)$.  Recall the definitions:
\begin{itemize}
	\item $\In(\alpha) = \In(\alpha_1) - \{x\}$.
	\item $\Out(\alpha) = \Out(\alpha_1) \cup \{x\}$.
\end{itemize}
We proceed to verify that $\alpha$ satisfies alternative input-output determinacy. 
If $(\nu_1, \nu_2) \in \sem{\alpha}$, then there exists a valuation $\nu'_1$ such that:
\begin{enumerate}[label=(\roman*)]
	\item\label{pid:cyll1} $\nu'_1 = \nu_1$ outside $\{x\}$;
	\item\label{pid:cyll2} $(\nu'_1, \nu_2) \in \sem{\alpha_1}$.
\end{enumerate}
Now, let $\nu$ be a valuation such that
\begin{equation}
\nu = \nu_1 \text{ on } \compl{\Out(\alpha)} \cup \In(\alpha). \label{pid:lc1}
\end{equation}
Clearly, $\nu[\restr{\nu_1'}{\{x\}}]=\nu$ outside $\{x\}$.  Since 
$x \not \in \compl{\Out(\alpha)} \cup \In(\alpha)$, we also know that 
$\nu[\restr{\nu_1'}{\{x\}}]=\nu$ on $\compl{\Out(\alpha)} \cup \In(\alpha)$. 
Combining this with (\ref{pid:lc1}), we get that $\nu[\restr{\nu_1'}{\{x\}}]=\nu_1'$ on 
$\compl{\Out(\alpha)} \cup \In(\alpha)\cup \{x\}$.  Clearly, 
$\compl{\Out(\alpha)} \cup \In(\alpha)\cup \{x\} \supseteq 
\compl{\Out(\alpha_1)} \cup \In(\alpha_1)$, whence $\nu[\restr{\nu_1'}{\{x\}}]=\nu_1'$ 
on $\compl{\Out(\alpha_1)} \cup \In(\alpha_1)$.  By induction from \ref{pid:cyll2} and 
Lemma~\ref{lem:eqvInDetDef}, we get that $(\nu[\restr{\nu_1'}{\{x\}}], \nu_2) \in \sem{\alpha_1}$, 
whence also $(\nu,\nu_2) \in \sem{\alpha_1}$.

\ignore{

\section{Precision Theorem Proof}\label{app:precision}
In this section, we prove Theorem~\ref{thm:precision}.  
By soundness and Proposition~\ref{prop:inputApp}, it suffices to prove 
$\synO(\alpha) \subseteq \semO(\alpha)$ and 
$\synI(\alpha) \subseteq \semI(\alpha)$ for every LIF expression $\alpha$.
For the latter inequality, it will be convenient to use the following 
equivalent definition of semantic input variables:
\begin{definition}[Semantic Inputs]
$x$ is a semantic input of $\alpha$, if there is an \interpretation $\inst$, a 
value $d \in \dom$, and $(\nu_1, \nu_2) \in \sem{\alpha}$ such that there is no 
valuation $\nu_2'$ that agrees with $\nu_2$ on the outputs of $\alpha$ such that 
$(\nu_1[x:d], \nu_2') \in \sem{\alpha}$.
\end{definition}
In the proof of the Precision Theorem, we will often make use of the following 
two technical lemmas.
\begin{lemma}\label{lem:null_id}
Let $M$ be a nullary relation name and let $\inst$ be an \interpretation where 
$\inst(M)$ is nonempty. Then $\sem{M()}$ consists of all identical pairs of 
valuations.
\end{lemma}
\begin{proof}
The proof follows directly from the semantics of atomic modules.
\end{proof}
\begin{lemma}\label{lem:reduction}
Suppose $\alpha_1 = M_1(\bar{x_1}; \bar{y_1})$ and 
$\alpha_2 = M_2(\bar{x_2}; \bar{y_2})$ where $M_1 \neq M_2$.  Let $\alpha$ be 
either $\alpha_1 \cup \alpha_2$ or $\alpha_1 - \alpha_2$.  Assume that 
$\synO(\alpha_i) \subseteq \semO(\alpha_i)$ and 
$\synI(\alpha_i) \subseteq \semI(\alpha_i)$ for $i = 1, 2$. 
Let $j \neq k \in \{1, 2\}$.  If $\semm{\alpha}{\inst} = \semm{\alpha_k}{\inst}$ 
for any \interpretation $\inst$ where $\inst(M_j)$ is empty, then 
$\synO(\alpha_k) \subseteq \semO(\alpha)$ and 
$\synI(\alpha_k) \subseteq \semI(\alpha)$.
\end{lemma}
\begin{proof}
First, we verify that $\synO(\alpha_k) \subseteq \semO(\alpha)$.  Let 
$v \in \synO(\alpha_k)$.  Since $\synO(\alpha_k) \subseteq \semO(\alpha_k)$, 
then $v \in \semO(\alpha_k)$.  By definition, we know that there is an 
\interpretation $\inst'$ and $(\nu_1, \nu_2) \in \semm{\alpha_k}{\inst'}$ such 
that $\nu_1(v) \neq \nu_2(v)$.  Take $\inst''$ to be the interpretation where 
$\inst''(M) = \inst'(M)$ for any $M \neq M_j$ while $\inst''(M_j)$ is empty.  
Clearly, $(\nu_1, \nu_2) \in \semm{\alpha}{\inst''}$, whence, $M_j \neq M_k$ 
and $\semm{\alpha}{\inst''} = \semm{\alpha_k}{\inst'}$.  It follows then that 
$v \in \semO(\alpha)$.

Similarly, we proceed to verify $\synI(\alpha_k) \subseteq \semI(\alpha)$.  Let 
$v \in \synI(\alpha_k)$. Since $\synI(\alpha_k) \subseteq \semI(\alpha_k)$, then 
$v \in \semI(\alpha_k)$. By definition, we know that there is an \interpretation 
$\inst'$, $(\nu_1, \nu_2) \in \semm{\alpha_k}{\inst'}$, and 
$\nu_1'(v) \neq \nu_1(v)$ such that 
$(\nu_1', \nu_2') \not \in \semm{\alpha_k}{\inst'}$ for every valuation $\nu_2'$ 
that agrees with $\nu_2$ on $\semO(\alpha_k)$.

Take $\inst''$ to be the interpretation where $\inst''(M) = \inst'(M)$ for any 
$M \neq M_j$ while $\inst''(M_j)$ is empty. Clearly, 
$\semm{\alpha}{\inst''} = \semm{\alpha_k}{\inst'}$, whence, $M_j \neq M_k$. 
Therefore, $\semO(\alpha_k) \subseteq \semO(\alpha)$.  Hence, $v \in \semI(\alpha)$. 
Indeed, $(\nu_1, \nu_2) \in \semm{\alpha}{\inst''}$ and for any valuation $\nu_2'$ 
if $\nu_2'$ agrees with $\nu_2$ on $\semO(\alpha)$, then $\nu_2'$ agrees with $\nu_2$ 
on $\semO(\alpha_k)$.
\end{proof}
The proof of Theorem~\ref{thm:precision} is done by extensive case analysis. 
Intuitively, for each of the different operations, and every variable 
$z\in \synO(\alpha)$, we construct an \interpretation $\inst$ such that $z$ 
is not inertial in $\sem{\alpha}$ and thus $z \in \semO(\alpha)$.  Similarly, 
for every variable $z \in \synI(\alpha)$, we construct an \interpretation 
$\inst$ as a witness of the fact that $\V \setminus \{z\}$ does not determine 
$\semNoI{\alpha}$ on $\semO(\alpha)$ and thus that $z\in \semI(\alpha)$.  
In the proof, we often remove the superscript from $\synI$ and $\synO$ and refer 
to them by $\In$ and $\Out$, respectively.

\subsection{Atomic Modules}
Let $\alpha$ be of the form $\alpha_1$, where $\alpha_1$ is $M(\bar{x}; \bar{y})$. 
Recall the definition:
\begin{itemize}
	\item $\synO(\alpha) = Y$, where $Y$ are the variables in $\bar{y}$.
	\item $\synI(\alpha) = X$, where $Y$ are the variables in $\bar{x}$.
\end{itemize}
We first proceed to verify $\synO(\alpha) \subseteq \semO(\alpha)$.  Let $v \in Y$. 
Consider an \interpretation $\inst$ where 
\[\inst(M) = \{(1, \ldots, 1; 2, \ldots, 2)\}.\]
Let $\nu_1$ be the valuation that is $1$ everywhere. Also let $\nu_2$ be the valuation 
that is $2$ on $Y$ and $1$ everywhere else. Clearly, $(\nu_1, \nu_2) \in \sem{\alpha}$ 
since $\nu_1(\bar{x}) \cdot \nu_2(\bar{y}) \in \inst(M)$ and $\nu_1$ agrees with $\nu_2$ 
outside $Y$. Hence, $v \in \semO(\alpha)$ since $\nu_1(v) \neq \nu_2(v)$.

Now we proceed to verify $\synI(\alpha) \subseteq \semI(\alpha)$.  Let $v \in X$.  
Consider the same \interpretation $\inst$ and the same valuations $\nu_1$ and $\nu_2$ 
as discussed above. We already established that $(\nu_1, \nu_2) \in \sem{\alpha}$. 
Take $\nu_1' := \nu_1[v:2]$. We establish that $v \in \semI(\alpha)$ by arguing that 
there is no $\nu_2'$ for which $(\nu_1', \nu_2') \in \sem{\alpha}$.  Indeed, this is 
true since $v \in X$. Consequently, $v \in \semI(\alpha)$.

\subsection{Identity}
Let $\alpha$ be of the form $\id$. We recall that $\synI(\alpha)$ and $\synO(\alpha)$ 
are both empty. Hence, $\synO(\alpha) \subseteq \semO(\alpha)$ and 
$\synI(\alpha) \subseteq \semI(\alpha)$ is trivial.

\subsection{Union}
Let $\alpha$ be of the form $\alpha_1 \cup \alpha_2$, where $\alpha_1$ is 
$M_1(\bar{x_1}; \bar{y_1})$ and $\alpha_2$ is $M_2(\bar{x_2}; \bar{y_2})$.  We 
distinguish different cases based on whether $M_1$ or $M_2$ is nullary.  If $M_1$ 
and $M_2$ are both nullary there is nothing to prove.

\paragraph{$M_1$ is nullary, $M_2$ is not} Clearly, $\sem{\alpha} = \sem{\alpha_2}$ 
for any \interpretation $\inst$ where $\inst(M_1)$ is empty. By induction and 
Lemma~\ref{lem:reduction}, we establish that $\synO(\alpha_2) \subseteq \semO(\alpha)$ 
and $\synI(\alpha_2) \subseteq \semI(\alpha)$. Since $\inn1$ and $\outt1$ are both empty, 
then we observe that:
\begin{itemize}
	\item $\synO(\alpha) = \outt2$.
	\item $\synI(\alpha) = \inn2 \cup \outt2$.
\end{itemize}
Thus, $\outt2 \subseteq \semO(\alpha)$ and $\inn2 \subseteq \semI(\alpha)$ is trivial.

We proceed to verify $\outt2 - \inn2 \subseteq \semI(\alpha)$. Let $v \in \outt2 - \inn2$. 
Consider the \interpretation $\inst$ where $\inst(M_1)$ is not empty and 
\[
\inst(M_2) = \{(1, \ldots, 1; 2, \ldots, 2) \}
\]
Let $\nu_1$ be the valuation that is $1$ everywhere.  Clearly, 
$(\nu_1, \nu_1) \in \sem{\alpha}$ since $(\nu_1, \nu_1) \in \sem{\alpha_1}$ by 
Lemma~\ref{lem:null_id}. Take $\nu_1' := \nu_1[v:2]$.  We establish that 
$v \in \semI(\alpha)$ by arguing that there for every valuation $\nu_2'$ for 
which $(\nu_1', \nu_2') \in \sem{\alpha}$, we show that $\nu_2'$ and $\nu_1$ 
disagrees on $\semO(\alpha)$. Thereto, suppose $(\nu_1', \nu_2') \in \sem{\alpha}$. 
In particular, $(\nu_1', \nu_2') \in \sem{\alpha_1}$, whence $\nu_2' = \nu_1'$ by 
Lemma~\ref{lem:null_id}.  Indeed, $v \in \outt2$, $\outt2 \subseteq \semO(\alpha)$, 
and $\nu_2'(v) = 2$ but $\nu_1(v) = 1$. Otherwise, $(\nu_1', \nu_2') \in \sem{\alpha_2}$. 
However, since $v \in \outt2$, then $\nu_2'(v) = 2$ as well.  Therefore, there is no 
$\nu_2'$ that agrees with $\nu_1$ on $\semO(\alpha)$ and $(\nu_1', \nu_2') \in \sem{\alpha}$ 
at the same time. We conclude that $v \in \semI(\alpha)$.

\paragraph{$M_2$ is nullary, $M_1$ is not} This case is symmetric to the previous one.

\paragraph{Neither $M_1$ nor $M_2$ is nullary}
Recall the definitions:
\begin{itemize}
	\item $\synO(\alpha) = \outt 1 \cup \outt 2$.
	\item $\synI(\alpha) = \inn 1 \cup \inn 2 \cup (\outt 1 \symdif \outt 2)$.
\end{itemize}
We first proceed to verify that $\synO(\alpha) \subseteq \semO(\alpha)$.  By induction, 
$\synO(\alpha_i) = \semO(\alpha_i)$ for $i = 1$ or $2$.  Consequently, if $v \in \outt i$, 
then there is an \interpretation $\inst_i$, $(\nu_1, \nu_2) \in \sem{\alpha_i}$ such that 
$\nu_1(v) \neq \nu_2(v)$.  Indeed, $v \in \semO(\alpha)$, whence, 
$(\sem{\alpha_1} \cup \sem{\alpha_2}) \subseteq \sem{\alpha}$ for any \interpretation $\inst$.

We then proceed to verify that $\synI(\alpha) \subseteq \semO(\alpha)$.  The proof has 
four possibilities. Each case is discussed in one of the following subsections.

\subsection*{$v \in \inn1$}
If $v \in \inn1$ and $M_1 \neq M_2$, it is clear that $\sem{\alpha} = \sem{\alpha_1}$
for any \interpretation $\inst$ where $\inst(M_2)$ is empty.
By Lemma~\ref{lem:reduction} and by induction, we easily establish that 
$v \in \semI(\alpha)$.

\subsection*{$v \in \inn2$} This case is symmetric to the previous one.

\subsection*{$v \in \outt 1 - (\outt 2 \cup \inn 1 \cup \inn 2)$}
If $v \in \outt 1 - (\outt 2 \cup \inn 1 \cup \inn 2)$, then consider the 
\interpretation $\inst$ such that $\inst(M_1)  = \{(1, \ldots, 1; 1, \ldots, 1)\}$ 
and $\inst(M_2)  = \{(1, \ldots, 1; 1, \ldots, 1)\}$. Let $\nu_1$ be the valuation that 
is $2$ on $v$ and $1$ elsewhere. Clearly, $(\nu_1, \nu_1) \in \sem{\alpha}$, whence 
$(\nu_1, \nu_1) \in \sem{\alpha_2}$.	Now take $\nu_1' := \nu_1[v:1]$. If we can show 
that $\nu_2'$ does not agree with $\nu_1$ on $\semO(\alpha)$ for any valuation $\nu_2'$ 
such that $(\nu_1', \nu_2') \in \sem{\alpha}$, we are done. Thereto, suppose that there 
exists a valuation $\nu_2'$ such that $(\nu_1', \nu_2') \in \sem{\alpha}$.
\begin{itemize}
	\item In particular, if $(\nu_1', \nu_2') \in \sem{\alpha_1}$, then $\nu_2'(v) = 1$ 
	since $v \in \outt1$.
	\item Otherwise, if $(\nu_1', \nu_2') \in \sem{\alpha_2}$, then 
	$\nu_2'(v) = \nu_1'(v) = 1$ since $v \not \in (\inn2 \cup \outt2)$.
\end{itemize}
In both cases, $\nu_2'$ have to be $1$ on $v$ which disagrees with $\nu_1$ on $v$.  Since 
$v \in \outt1$ and $\outt1 \subseteq \semO(\alpha)$, then $\nu_2'$ does not agree with 
$\nu_1$ on $\semO(\alpha)$ as desired. We conclude that $v \in \semI(\alpha)$.

\subsection*{$v \in \outt 2 - (\outt 1 \cup \inn 1 \cup \inn 2)$} This case is symmetric 
to the previous one.

\subsection{Intersection}
Let $\alpha$ be of the form $\alpha_1 \cap \alpha_2$, where $\alpha_1$ is 
$M_1(\bar{x_1}; \bar{y_1})$ and $\alpha_2$ is $M_2(\bar{x_2}; \bar{y_2})$.  We distinguish 
different cases based on whether $M_1$ or $M_2$ is nullary. If $M_1$ and $M_2$ are both 
nullary there is nothing to prove.

\paragraph{$M_1$ is nullary, $M_2$ is not}
In this case, $\inn1$ and $\outt1$ are both empty, then we observe that:
\begin{itemize}
	\item $\synO(\alpha) = \emptyset$.
	\item $\synI(\alpha) = \inn2 \cup \outt2$.
\end{itemize}
It is trivial to verify $\synO(\alpha) \subseteq \semO(\alpha)$ since $\synO(\alpha)$ 
is empty.

We proceed to verify $\synI(\alpha) \subseteq \semI(\alpha)$. Let $v \in \inn2 \cup \outt2$. 
Consider an \interpretation $\inst$ where $\inst(M_1)$ is not empty and 
$\inst(M_2) = \{(1, \ldots, 1; 1, \ldots, 1)\}$.  Let $\nu_1$ be the valuation that is 
$1$ everywhere.  Clearly, $(\nu_1, \nu_1) \in \sem{\alpha}$ since 
$(\nu_1, \nu_1) \in \sem{\alpha_1}$ and $(\nu_1, \nu_1) \in \sem{\alpha_2}$.  Take 
$\nu_1' := \nu_1[v:2]$. We establish that $v \in \semI(\alpha)$ by arguing that there 
is no valuation $\nu_2'$ for which $(\nu_1', \nu_2') \in \sem{\alpha}$.  Thereto, suppose 
$(\nu_1', \nu_2') \in \sem{\alpha}$.  In particular, when $v \in \inn2$, it is clear that 
$(\nu_1', \nu_2') \not \in \sem{\alpha_1}$.  In the other case when $v \in \outt2 - \inn2$,
there is no $\nu_2'$ such that $(\nu_1', \nu_2')$ belongs to both $\sem{\alpha_1}$ and 
$\sem{\alpha_2}$. Indeed, the value for $\nu_2'(v)$ will never be agreed upon by $\alpha_1$ 
and $\alpha_2$.  Hence, $(\nu_1', \nu_2') \not \in \sem{\alpha}$ as desired.  We conclude 
that $v \in \semI(\alpha)$.

\paragraph{$M_2$ is nullary, $M_1$ is not} This case is symmetric to the previous one.

\paragraph{Neither $M_1$ nor $M_2$ is nullary}
Recall the definitions:
\begin{itemize}
	\item $\synO(\alpha) = \outt 1 \cap \outt 2$.
	\item $\synI(\alpha) = \inn 1 \cup \inn 2 \cup (\outt 1 \symdif \outt 2)$.
\end{itemize}

We first proceed to verify $\synO(\alpha) \subseteq \semO(\alpha)$.  Let 
$v \in \outt1 \cap \outt2$.  Consider an \interpretation $\inst$ such that 
\[
\inst(M_1) = \{(1, \ldots, 1; o_1, \ldots, o_m)\}
\]
, where $o_1, \ldots, o_m$ are all the combinations of $\{1, 2\}$.  Similarly, 
$\inst(M_2) = \{(1, \ldots, 1; o_1, \ldots, o_n)\}$, where $o_1, \ldots, o_n$ are 
all the combinations of $\{1, 2\}$.

Let $\nu_1$ be the valuation that is $1$ everywhere. Also let $\nu_2$ be the valuation 
that is $2$ on $\outt1 \cap \outt2$ and $1$ elsewhere.  Clearly, 
$(\nu_1, \nu_2) \in \sem{\alpha}$, whence $(\nu_1, \nu_2) \in \sem{\alpha_1}$ and 
$(\nu_1, \nu_2) \in \sem{\alpha_2}$.  Hence, $v \in \semO(\alpha)$.  Indeed, 
$\nu_2(v) \neq \nu_1(v)$ for $v \in \synO(\alpha)$.

We then proceed to verify $\synI(\alpha) \subseteq \semO(\alpha)$.  Let 
$v \in \inn1 \cup \inn2 \cup (\outt1 \symdif \outt2)$.  Consider an \interpretation 
$\inst$ where 
\[
\inst(M_1) = \{(1, \ldots, 1; 1, \ldots, 1)\}
\] 
Similarly, $\inst(M_2) = \{(1, \ldots, 1; 1, \ldots, 1)\}$.  Let $\nu_1$ be the valuation 
that is $1$ everywhere. Clearly, $(\nu_1, \nu_1) \in \sem{\alpha}$, whence 
$(\nu_1, \nu_1) \in \sem{\alpha_1}$ and $(\nu_1, \nu_1) \in \sem{\alpha_2}$.  Take 
$\nu_1' := \nu_1[v:2]$. We establish that $v \in \semI(\alpha)$ by arguing that there 
is no valuation $\nu_2'$ such that $(\nu_1', \nu_2') \in \sem{\alpha}$.  Indeed, this is 
clear when $v \in \inn1$ or $v \in \inn2$.  On the other hand, when 
$v \in (\outt1 \symdif \outt2) - (\inn1 \cup \inn2)$, we have $\nu_2'$ for which 
$(\nu_1', \nu_2') \in \sem{\alpha}$ whence $(\nu_1', \nu_2') \in \sem{\alpha_1}$ and 
$(\nu_1', \nu_2') \in \sem{\alpha_2}$. This is not possible since $v$ belongs to either 
$\outt1$ or $\outt2$, but not both. Hence, the value for $\nu_2'(v)$ will never be 
agreed upon by $\alpha_1$ and $\alpha_2$. We conclude that $v \in \semI(\alpha)$.

\subsection{Difference}
Let $\alpha$ be of the form $\alpha_1 \setminus \alpha_2$, where $\alpha_1$ is 
$M_1(\bar{x_1}; \bar{y_1})$ and $\alpha_2$ is $M_2(\bar{x_2}; \bar{y_2})$.  We 
distinguish different cases based on whether $M_1$ or $M_2$ is nullary.  If $M_1$ 
and $M_2$ are both nullary there is nothing to prove.

\paragraph{$M_1$ is nullary, $M_2$ is not} In this case, $\inn1$ and $\outt 1$ are 
empty. In particular, $\synO(\alpha)$ is empty, so $\synO(\alpha) \subseteq \semO(\alpha)$ 
is trivial.

We proceed to verify $\synI(\alpha) \subseteq \semI(\alpha)$.  Observe that
\[
\synI(\alpha) = \inn2 \cup \outt2.
\]
Let $v \in \synI(\alpha)$. Consider the \interpretation $\inst$ where $\inst(M_1)$ is 
not empty and $\inst(M_2) = \{(1, \ldots, 1; 1, \ldots, 1) \}$.
Let $\nu_1$ be the valuation that is $2$ on $v$ and $1$ elsewhere.  Clearly, 
$(\nu_1, \nu_1) \in \sem{\alpha}$. Take $\nu_1' := \nu_1[v:1]$.  We establish that 
$v \in \semI(\alpha)$ by arguing that there is no $\nu_2'$ for which 
$(\nu_1', \nu_2') \in \sem{\alpha}$.  Thereto, suppose $(\nu_1', \nu_2') \in \sem{\alpha}$. 
In particular, $(\nu_1', \nu_2') \in \sem{\alpha_1}$, whence $\nu_2' = \nu_1'$ by Lemma 
\ref{lem:null_id}. However, $(\nu_1', \nu_1') \in \sem{\alpha_2}$, so 
$(\nu_1', \nu_2') \not \in \sem{\alpha}$ as desired.

\paragraph{$M_2$ is nullary, $M_1$ is not} Clearly, $\sem{\alpha} = \sem{\alpha_1}$ for any
\interpretation $\inst$ where $\inst(M_2)$ is empty. 
By induction and Lemma \ref{lem:reduction}, we establish that 
$\synO(\alpha_1) \subseteq \semO(\alpha)$ and $\synI(\alpha_1) \subseteq \semI(\alpha)$. 
Since $\inn2$ and $\outt2$ are both empty, then we observe that:
\begin{itemize}
	\item $\synO(\alpha) = \outt1$.
	\item $\synI(\alpha) = \inn1 \cup \outt1$.
\end{itemize}
Thus, $\outt1 \subseteq \semO(\alpha)$ and $\inn1 \subseteq \semI(\alpha)$ is trivial.

We proceed to verify $\outt1 - \inn1 \subseteq \semI(\alpha)$.  Let
$v \in \outt1 - \inn1$.  Consider the \interpretation $\inst$ where $\inst(M_2)$ is 
not empty and 
\[
\inst(M_1) = \{(1, \ldots, 1; 1, \ldots, 1) \}
\]
Let $\nu_1$ be the valuation that is $2$ on $v$ and $1$ elsewhere and let $\nu_2$ be 
the valuation that is $1$ everywhere. Clearly, $(\nu_1, \nu_2) \in \sem{\alpha}$.  Take 
$\nu_1' := \nu_1[v:1]$.  We establish that $v \in \semI(\alpha)$ by arguing that there 
is no $\nu_2'$ for which $(\nu_1', \nu_2') \in \sem{\alpha}$.  Thereto, suppose 
$(\nu_1', \nu_2') \in \sem{\alpha}$.  In particular, 
$(\nu_1', \nu_2') \in \sem{\alpha_1}$, whence $\nu_2' = \nu_1'$ from the structure of 
$\inst$. However, $(\nu_1', \nu_1') \in \sem{\alpha_2}$ by Lemma \ref{lem:null_id}, 
so $(\nu_1', \nu_2') \not \in \sem{\alpha}$ as desired.

\paragraph{Neither $M_1$ nor $M_2$ is nullary} Recall the definitions:
\begin{itemize}
	\item $\synO(\alpha) = \outt 1$.
	\item $\synI(\alpha) = \inn 1 \cup \inn 2 \cup (\outt 1 \symdif \outt 2)$.
\end{itemize}
The proof of $\synO(\alpha) \subseteq \semO(\alpha)$ is done together with the proof 
that $v \in \semI(\alpha)$ for every $v \in \inn1$.  Subsections for the other cases 
for $v \in \synI(\alpha)$ follow afterwards. Since $M_1 \neq M_2$, it is clear that 
$\sem{\alpha} = \sem{\alpha_1}$ for any \interpretation $\inst$ where $\inst(M_2)$ is 
empty.  By induction and Lemma \ref{lem:reduction}, we establish that 
$\synO(\alpha_1) \subseteq \semO(\alpha)$ and $\synI(\alpha_1) \subseteq \semI(\alpha)$. 
Thus, $\outt1 \subseteq \semO(\alpha)$ and $\inn1 \subseteq \semI(\alpha)$ is trivial.

\subsection*{$v \in \inn2 - \inn1$}
Let $v \in \inn2 - \inn1$. Consider an \interpretation $\inst$ where 
\[\inst(M_1) = \{(1, \ldots, 1; 1, \ldots, 1)\}\] 
and similarly $\inst(M_2)$ $= \{(1, \ldots, 1; 1, \ldots, 1)\}$. Let $\nu_1$ be the 
valuation that $2$ on $v$ and $1$ elsewhere. Also, let $\nu_2$ be the valuation that is 
$1$ on $\outt1$ and agrees with $\nu_1$ everywhere else. Clearly,
$(\nu_1, \nu_2) \in \sem{\alpha_1}$. Further, $(\nu_1, \nu_2) \not \in \sem{\alpha_2}$.
Indeed, since $v \in \inn2$ then $\nu_1$ should have the value of $1$ on $v$ for
$(\nu_1, \nu_2)$ to be in $\sem{\alpha_2}$. Take $\nu_1' := \nu_1[v:1]$. We establish 
that $v \in \semI(\alpha)$ by arguing that there is no $\nu_2'$ for which
$(\nu_1', \nu_2') \in \sem{\alpha}$. Thereto, suppose that 
$(\nu_1', \nu_2') \in \sem{\alpha}$. Hence, $(\nu_1', \nu_2') \in \sem{\alpha_1}$ and 
$(\nu_1', \nu_2') \in \sem{\alpha_2}$. Indeed, $(\nu_1', \nu_2') \in \sem{\alpha_1}$ 
whence $\nu_1' = \nu_2'$. Clearly, $(\nu_1', \nu_1') \in \sem{\alpha_2}$ showing that
$(\nu_1', \nu_1') \not \in \sem{\alpha}$ as desired. Therefore, $v \in \semI(\alpha)$.

\subsection*{$v \in (\outt1 \symdif \outt2) - (\inn1 \cup \inn2)$}
Let $v \in (\outt1 \symdif \outt2) - (\inn1 \cup \inn2)$. Consider an \interpretation 
$\inst$ where $\inst(M_1) = \{(1, \ldots, 1; 1, \ldots, 1)\}$ and 
$\inst(M_2) = \{(1, \ldots, 1; 1, \ldots, 1)\}$. Let $\nu_1$ be the valuation that is $2$ 
on $v$ and $1$ elsewhere. Also let $\nu_2$ be the valuation that is $1$ on $\outt1$ and
agrees with $\nu_1$ everywhere else. Clearly, $(\nu_1, \nu_2) \in \sem{\alpha_1}$.
Furthermore, $(\nu_1, \nu_2) \not \in \sem{\alpha_2}$. In particular, when 
$v \in \outt1 - (\inn1 \cup \inn2 \cup \outt2)$, we know that $\nu_1(v) = 2$ and 
$\nu_2(v) = 1$. Since $v \not \in \outt2$, then $(\nu_1, \nu_2) \not \in \sem{\alpha_2}$.
In the other case, when $v \in \outt2 - (\inn1 \cup \inn2 \cup \outt1)$, we know that
$\nu_1(v) = \nu_2(v) = 2$ since $(\nu_1, \nu_2) \in \sem{\alpha_1}$. Consequently, 
$(\nu_1, \nu_2) \not \in \sem{\alpha_2}$ since $v \in \outt2$ but $\nu_2(v) = 2$. We 
verify that $(\nu_1, \nu_2) \in \sem{\alpha}$. Take $\nu_1' := \nu_1[v:1]$. We establish 
that $v \in \semI(\alpha)$ by arguing that there is no $\nu_2'$ for which 
$(\nu_1', \nu_2') \in \sem{\alpha}$. Thereto, suppose that 
$(\nu_1', \nu_2') \in \sem{\alpha}$. Hence, $(\nu_1', \nu_2') \in \sem{\alpha_1}$ and
$(\nu_1', \nu_2') \in \sem{\alpha_2}$. Indeed, $(\nu_1', \nu_2') \in \sem{\alpha_1}$
whence $\nu_1' = \nu_2'$. Clearly, $(\nu_1', \nu_1') \in \sem{\alpha_2}$ showing that
$(\nu_1', \nu_1') \not \in \sem{\alpha}$ as desired. Therefore, $v \in \semI(\alpha)$.

\subsection{Composition}
Let $\alpha$ be of the form $\alpha_1 \comp \alpha_2$, where $\alpha_1$ is 
$M_1(\bar{x_1}; \bar{y_1})$ and $\alpha_2$ is $M_2(\bar{x_2}; \bar{y_2})$. We
distinguish different cases based on whether $M_1$ or $M_2$ is nullary. If $M_1$ and 
$M_2$ are both nullary there is nothing to prove.

\paragraph{$M_1$ is nullary, $M_2$ is not} Clearly, $\sem{\alpha} = \sem{\alpha_2}$ for 
any \interpretation $\inst$ where $\inst(M_1)$ is not empty. In this case, $\inn1$ and 
$\outt1$ are both empty, then we observe that:
\begin{itemize}
	\item $\synO(\alpha) = \outt2$.
	\item $\synI(\alpha) = \inn2$.
\end{itemize}
First, we verify $\synO(\alpha) \subseteq \semO(\alpha)$. Let $v \in \outt2$. We know 
that $\synO(\alpha_2) \subseteq \semO(\alpha_2)$ by induction, then $v \in \semO(\alpha_2)$.
By definition, we know that there is an \interpretation $\inst'$ and
$(\nu_1, \nu_2) \in \semm{\alpha_2}{\inst'}$ such that $\nu_1(v) \neq \nu_2(v)$. Take
$\inst''$ to be the interpretation where $\inst''(M) = \inst'(M)$ for any $M \neq M_1$
while $\inst''(M_1)$ is not empty. Clearly, $(\nu_1, \nu_2) \in \semm{\alpha}{\inst''}$, 
whence, $M_1 \neq M_2$, $(\nu_1, \nu_1) \in \semm{\alpha_1}{\inst''}$ by 
Lemma \ref{lem:null_id}, and $\semm{\alpha}{\inst''} = \semm{\alpha_2}{\inst'}$. It 
follows then that $v \in \semO(\alpha)$.

Similarly, we proceed to verify $\synI(\alpha) \subseteq \semI(\alpha)$. Let $v \in \inn2$.
We know that $\synI(\alpha_2) \subseteq \semI(\alpha_2)$ by induction, then 
$v \in \semI(\alpha_2)$. By definition, we know that there is an \interpretation
$\inst'$, $(\nu_1, \nu_2) \in \semm{\alpha_2}{\inst'}$, and $\nu_1'(v) \neq \nu_1(v)$
such that $(\nu_1', \nu_2') \not \in \semm{\alpha_2}{\inst'}$ for every valuation 
$\nu_2'$ that agrees with $\nu_2$ on $\semO(\alpha_2)$.

Take $\inst''$ to be the interpretation where $\inst''(M) = \inst'(M)$ for any 
$M \neq M_1$ while $\inst''(M_1)$ is not empty. Clearly, 
$\semm{\alpha}{\inst''} = \semm{\alpha_2}{\inst'}$, whence, $M_1 \neq M_2$. Therefore, 
$\semO(\alpha_2) \subseteq \semO(\alpha)$. Hence, $v \in \semI(\alpha)$. Indeed, 
$(\nu_1, \nu_2) \in \semm{\alpha}{\inst''}$ and for any valuation $\nu_2'$ if $\nu_2'$ 
agrees with $\nu_2$ on $\semO(\alpha)$, then $\nu_2'$ agrees with $\nu_2$ on 
$\semO(\alpha_2)$.

\paragraph{$M_2$ is nullary, $M_1$ is not} Clearly, $\sem{\alpha} = \sem{\alpha_1}$ for 
any \interpretation $\inst$ where $\inst(M_2)$ is not empty. In this case, $\inn2$ and
$\outt2$ are both empty, then we observe that:
\begin{itemize}
	\item $\synO(\alpha) = \outt1$.
	\item $\synI(\alpha) = \inn1$.
\end{itemize}
First, we verify $\synO(\alpha) \subseteq \semO(\alpha)$. Let $v \in \outt1$. We know that
$\synO(\alpha_1) \subseteq \semO(\alpha_1)$ by induction, then $v \in \semO(\alpha_1)$. 
By definition, we know that there is an \interpretation $\inst'$ and 
$(\nu_1, \nu_2) \in \semm{\alpha_1}{\inst'}$ such that $\nu_1(v) \neq \nu_2(v)$. Take 
$\inst''$ to be the interpretation where $\inst''(M) = \inst'(M)$ for any $M \neq M_2$
while $\inst''(M_2)$ is not empty. Clearly, $(\nu_1, \nu_2) \in \semm{\alpha}{\inst''}$,
whence, $M_1 \neq M_2$, $(\nu_2, \nu_2) \in \semm{\alpha_2}{\inst''}$ by 
Lemma \ref{lem:null_id}, and $\semm{\alpha}{\inst''} = \semm{\alpha_1}{\inst'}$. It
follows then that $v \in \semO(\alpha)$.

Similarly, we proceed to verify $\synI(\alpha) \subseteq \semI(\alpha)$. Let $v \in \inn1$.
We know that $\synI(\alpha_1) \subseteq \semI(\alpha_1)$ by induction, then 
$v \in \semI(\alpha_1)$. By definition, we know that there is an \interpretation
$\inst'$, $(\nu_1, \nu_2) \in \semm{\alpha_1}{\inst'}$, and $\nu_1'(v) \neq \nu_1(v)$ 
such that $(\nu_1', \nu_2') \not \in \semm{\alpha_1}{\inst'}$ for every valuation
$\nu_2'$ that agrees with $\nu_2$ on $\semO(\alpha_1)$.

Take $\inst''$ to be the interpretation where $\inst''(M) = \inst'(M)$ for any 
$M \neq M_2$ while $\inst''(M_2)$ is not empty. Clearly,
$\semm{\alpha}{\inst''} = \semm{\alpha_1}{\inst'}$, whence, $M_1 \neq M_2$. Therefore,
$\semO(\alpha_1) \subseteq \semO(\alpha)$. Hence, $v \in \semI(\alpha)$. Indeed, 
$(\nu_1, \nu_2) \in \semm{\alpha}{\inst''}$ and for any valuation $\nu_2'$ if $\nu_2'$ 
agrees with $\nu_2$ on $\semO(\alpha)$, then $\nu_2'$ agrees with $\nu_2$ on 
$\semO(\alpha_1)$.

\paragraph{Neither $M_1$ nor $M_2$ is nullary} Recall the definitions:
\begin{itemize}
	\item $\synO(\alpha) = \outt1 \cup \outt2$.
	\item $\synI(\alpha) = \inn1 \cup (\inn2 - \outt1)$.
\end{itemize}
We first proceed to verify $\synO(\alpha) \subseteq \semO(\alpha)$.  Let 
$v \in \outt1 \cup \outt2$. Consider an \interpretation $\inst$ such that 
\[
\inst(M_1) = \{(1, \ldots, 1; 2, \ldots, 2), (i_1, \ldots, i_m; 3, \ldots, 3)\}
\]
, where $i_1, \ldots, i_m$ are all the combinations of $\{1, 2\}$. 
Similarly, 
\[
\inst(M_2) = \{(1, \ldots, 1; 2, \ldots, 2), (i_1, \ldots, i_n; 3, \ldots, 3)\}
\]
, where $i_1, \ldots, i_n$ are all the combinations of $\{1, 2\}$.

Let $\nu_1$ be the valuation that is $1$ everywhere. Also, let $\nu$ be the valuation 
that is $2$ on $\outt1$ and $1$ elsewhere. Clearly, $(\nu_1, \nu) \in \sem{\alpha_1}$. 
Let $\nu_2$ be the valuation that is $3$ on $\outt2$, $2$ on $\outt1 - \outt2$, and 
$1$ elsewhere. Clearly, $(\nu_1, \nu_2) \in \sem{\alpha}$, whence 
$(\nu, \nu_2) \in \sem{\alpha_2}$. Hence, $v \in \semO(\alpha)$. Indeed, 
$\nu_2(v) \neq \nu_1(v)$ for $v \in \synO(\alpha)$.

Now we proceed to verify $\synI(\alpha) \subseteq \semI(\alpha)$. Let 
$v \in \inn1 \cup (\inn2 - \outt1)$. Consider an \interpretation $\inst$ 
where $\inst(M_1) = \{(1, \ldots, 1; 1, \ldots, 1)\}$ and similarly 
$\inst(M_2) = \{(1, \ldots, 1; 1, \ldots, 1)\}$. Let $\nu_1$ be the valuation that 
is $1$ everywhere. Clearly, $(\nu_1, \nu_1) \in \sem{\alpha}$, whence 
$(\nu_1, \nu_1) \in \sem{\alpha_1}$ and $(\nu_1, \nu_1) \in \sem{\alpha_2}$.

Take $\nu_1' := \nu_1[v:2]$. We establish that $v \in \semI(\alpha)$ by arguing that 
there is no valuation $\nu_2'$ for which $(\nu_1', \nu_2') \in \sem{\alpha}$.
In particular, when $v \in \inn1$. Clearly, there is no $\nu_2'$ such that 
$(\nu_1', \nu_2') \in \sem{\alpha_1}$. On the other hand, when $v \in \inn2 - \outt1$. 
Clearly, $(\nu_1', \nu) \in \sem{\alpha_1}$, whence $\nu = \nu_1'$. However, there is 
no $\nu_2'$ such that $(\nu_1', \nu_2') \in \sem{\alpha_2}$. Thus, there is no $\nu_2'$ 
such that $(\nu_1', \nu_2') \not \in \sem{\alpha}$ as desired. We conclude that 
$v \in \semI(\alpha)$.

\subsection{Converse}
Let $\alpha$ be of the form $\conv{\alpha_1}$, where $\alpha_1 := M(\bar{x}; \bar{y})$. 
Recall the definitions:
\begin{itemize}
	\item $\synO(\alpha) = \outt 1$.
	\item $\synI(\alpha) = \inn 1 \cup \outt 1$.
\end{itemize}
We first proceed to verify $\synO(\alpha) \subseteq \semO(\alpha)$.  Let $v \in \outt1$.
Consider an \interpretation $\inst$ where 
\[
\inst(M) = \{(1, \ldots, 1; 2, \ldots, 2)\}
\] 
Let $\nu_1$ be the valuation that is $2$ on $\outt1$ and $1$ elsewhere.  Also 
let $\nu_2$ be the valuation that is $1$ everywhere. Clearly, 
$(\nu_1, \nu_2) \in \sem{\alpha}$ since $(\nu_2, \nu_1) \in \sem{\alpha_1}$. 
Therefore, $v \in \semO(\alpha)$ since $\nu_1(v) \neq \nu_2(v)$.

Now we proceed to verify $\synI(\alpha) \subseteq \semI(\alpha)$. Let 
$v \in \inn1 \cup \outt1$. Consider the same \interpretation $\inst$ and the same 
valuations $\nu_1$ and $\nu_2$. We established that $(\nu_1, \nu_2) \in \sem{\alpha}$.
Take $\nu_1'  := \nu_1[v:3]$. We establish that $v \in \semI(\alpha)$ by arguing that 
there is no $\nu_2'$ for which $(\nu_1', \nu_2') \in \sem{\alpha}$. Indeed, when 
$v \in \outt1$, then $\nu_1$ have to be $2$ on $v$. In the other case, when 
$v \in \inn1 - \outt1$, then $\nu_1$ have to be $1$ on $v$. Thus, there is no 
$\nu_2'$ for which $(\nu_1', \nu_2') \in \sem{\alpha}$ as desired. Consequently, 
$v \in \semI(\alpha)$.

\subsection{Left Cylindrification}
Let $\alpha$ be of the form $\cyll{x}({\alpha_1})$, where $\alpha_1 := M(\bar{x}; \bar{y})$. 
Recall the definitions:
\begin{itemize}
	\item $\synO(\alpha) = \outt 1 \cup \{x\}$.
	\item $\synI(\alpha) = \inn 1 - \{x\}$.
\end{itemize}
We first proceed to verify $\synO(\alpha) \subseteq \semO(\alpha)$. Let 
$v \in \outt1 \cup \{x\}$. Consider an \interpretation $\inst$ where 
$\inst(M) = \{(1, \ldots, 1; 2, \ldots, 2)\}$. Let $\nu_1$ be the valuation that is 
$3$ on $x$ and $1$ elsewhere. Also let $\nu_2$ be the valuation that is $2$ on $\outt1$
and $1$ everywhere else. Clearly, $(\nu_1, \nu_2) \in \sem{\alpha}$ since 
$(\nu_1[x:1], \nu_2) \in \sem{\alpha_1}$. Therefore, $v \in \semO(\alpha)$ since 
$\nu_1(v) \neq \nu_2(v)$.

Now we proceed to verify $\synI(\alpha) \subseteq \semI(\alpha)$. Let $v \in \inn1 - \{x\}$. 
Consider the same \interpretation $\inst$ and the same valuations $\nu_1$ and $\nu_2$. We 
established that $(\nu_1, \nu_2) \in \sem{\alpha}$.
Take $\nu_1'  := \nu_1[v:2]$. We establish that $v \in \semI(\alpha)$ by arguing that there is 
no $\nu_2'$ for which $(\nu_1', \nu_2') \in \sem{\alpha}$. Indeed, this is true since
$v \in \inn1 - \{x\}$. Consequently, $v \in \semI(\alpha)$.

\subsection{Right Cylindrification}
Let $\alpha$ be of the form $\cylr{x}({\alpha_1})$, where $\alpha_1 := M(\bar{x}; \bar{y})$. 
Recall the definitions:
\begin{itemize}
	\item $\synO(\alpha) = \outt 1 \cup \{x\}$.
	\item $\synI(\alpha) = \inn 1$.
\end{itemize}
We first proceed to verify $\synO(\alpha) \subseteq \semO(\alpha)$. Let 
$v \in \outt1 \cup \{x\}$. Consider an \interpretation $\inst$ where
$\inst(M) = \{(1, \ldots, 1; 2, \ldots, 2)\}$. Let $\nu_1$ be the valuation that 
is $1$ everywhere. Also let $\nu_2$ be the valuation that is $2$ on $\outt1$ and on $x$ 
and $1$ everywhere else. Clearly, $(\nu_1, \nu_2) \in \sem{\alpha}$ since either 
$(\nu_1, \nu_2[x:1]) \in \sem{\alpha_1}$ or $(\nu_1, \nu_2) \in \sem{\alpha_1}$. 
Therefore, $v \in \semO(\alpha)$ since $\nu_1(v) \neq \nu_2(v)$.

Now we proceed to verify $\synI(\alpha) \subseteq \semI(\alpha)$. Let $v \in \inn1$. 
Consider the same \interpretation $\inst$ and the same valuations $\nu_1$ and $\nu_2$.
We established that $(\nu_1, \nu_2) \in \sem{\alpha}$.
Take $\nu_1'  := \nu_1[v:2]$. We establish that $v \in \semI(\alpha)$ by arguing that 
there is no $\nu_2'$ for which $(\nu_1', \nu_2') \in \sem{\alpha}$. Indeed, this is true
since $v \in \inn1$. Consequently, $v \in \semI(\alpha)$.

\subsection{Left Selection}
Let $\alpha$ be of the form $\sell{x = y}({\alpha_1})$, where 
$\alpha_1 := M(\bar{u}; \bar{w})$, $\bar{u} = u_1, \ldots, u_n$, and 
$\bar{w} = w_1, \ldots, w_m$. We distinguish different cases based on whether
$x \eqv y$.

\paragraph{$x$ and $y$ are the same variable ($x \eqv y$)} Recall the definitions in this
case:
\begin{itemize}
	\item $\synO(\alpha) = \outt1$.
	\item $\synI(\alpha) = \inn1$.
\end{itemize}
We proceed to verify that $\synO(\alpha) \subseteq \semO(\alpha)$ and 
$\synI(\alpha) \subseteq \semI(\alpha)$. Indeed, this is true since 
$\sem{\alpha} = \sem{\alpha_1}$ for any \interpretation $\inst$ because of 
$x \eqv y$.

\paragraph{$x$ and $y$ are different variables ($x \neqv y$)} Recall the definitions 
in this case:
\begin{itemize}
	\item $\synO(\alpha) := \outt1$.
	\item $\synI(\alpha) := \inn1 \cup \{x,y\}$.
\end{itemize}
We first proceed to verify $\synO(\alpha) \subseteq \semO(\alpha)$. Let $v \in \outt1$.
Consider an \interpretation $\inst$ where
\[\inst(M) = \{(1, \ldots, 1; 2, \ldots, 2)\}.\] 
Let $\nu_1$ be the valuation that is $1$ everywhere. Also let $\nu_2$ be the valuation that 
is $2$ on $\outt1$ and $1$ everywhere else. Clearly, $(\nu_1, \nu_2) \in \sem{\alpha}$ 
since $(\nu_1, \nu_2) \in \sem{\alpha_1}$ and $\nu_1(x) = \nu_1(y)$. Therefore, 
$v \in \semO(\alpha)$ since $\nu_1(v) \neq \nu_2(v)$.

Now we proceed to verify $\synI(\alpha) \subseteq \semI(\alpha)$. Let 
$v \in \inn1 \cup \{x, y\}$. Consider an \interpretation $\inst$ where 
$\inst(M_1) = \{(1, \ldots, 1; 1, \ldots, 1)\}$. Let $\nu_1$ be the valuation that
is $1$ everywhere. Clearly, $(\nu_1, \nu_1) \in \sem{\alpha}$ since 
$(\nu_1, \nu_1) \in \sem{\alpha_1}$ and $\nu_1(x) = \nu_1(y)$. Take 
$\nu_1' := \nu_1[v:2]$. We establish that $v \in \semI(\alpha)$ by arguing that there 
is no $\nu_2'$ for which $(\nu_1', \nu_2') \in \sem{\alpha}$. In particular, when 
$v \in \inn1$, it is clear that there is no $\nu_2'$ such that
$(\nu_1', \nu_2') \in \sem{\alpha_1}$. In the other case, when $v$ is either $x$ or $y$,
there is no $\nu_2'$ such that $(\nu_1', \nu_2') \in \sem{\alpha}$. Indeed, this is true
since $x \neqv y$ and $\nu_1'(x) \neq \nu_1'(y)$. Consequently, $v \in \semI(\alpha)$.

\subsection{Right Selection}
Let $\alpha$ be of the form $\selr{x = y}({\alpha_1})$, where
$\alpha_1 := M(\bar{u}; \bar{w})$, $\bar{u} = u_1, \ldots, u_n$, and 
$\bar{w} = w_1, \ldots, w_m$. We distinguish different cases based on whether
$x \eqv y$.

\paragraph{$x$ and $y$ are the same variable ($x \eqv y$)} Recall the definitions in 
this case:
\begin{itemize}
	\item $\synO(\alpha) = \outt1$.
	\item $\synI(\alpha) = \inn1$.
\end{itemize}
We proceed to verify that $\synO(\alpha) \subseteq \semO(\alpha)$ and 
$\synI(\alpha) \subseteq \semI(\alpha)$. Indeed, this is true since 
$\sem{\alpha} = \sem{\alpha_1}$ for any \interpretation $\inst$ because of
$x \eqv y$.

\paragraph{$x$ and $y$ are different variables ($x \neqv y$)} Recall the definitions 
in this case:
\begin{itemize}
	\item $\synO(\alpha) := \outt1$.
	\item $\synI(\alpha) := \inn1 \cup (\{x,y\} - \outt1)$.
\end{itemize}
We first proceed to verify $\synO(\alpha) \subseteq \semO(\alpha)$. Let $v \in \outt1$.
Consider an \interpretation $\inst$ where 
\[\inst(M) = \{(i_1, \ldots, i_n; 2, \ldots, 2)\}\]
such that $i_j = 2$ if $u_j$ is either $x$ or $y$ and $u_j \not \in \outt1$, otherwise, 
$u_j = 1$. Let $\nu_1$ be the valuation that is $2$ on $x$ if $x \not \in \outt1$, $2$ on 
$y$ if $y \not \in \outt1$, and $1$ everywhere. Also let $\nu_2$ be the valuation that is 
$2$ on $\outt1$ and agrees with $\nu_1$ everywhere else. Clearly, 
$(\nu_1, \nu_2) \in \sem{\alpha}$ since $(\nu_1, \nu_2) \in \sem{\alpha_1}$ and
$\nu_2(x) = \nu_2(y)$. Therefore, $v \in \semO(\alpha)$ since $\nu_1(v) \neq \nu_2(v)$.

Now we proceed to verify $\synI(\alpha) \subseteq \semI(\alpha)$. Let
$v \in \inn1 \cup \{x, y\}$. Consider an \interpretation $\inst$ where 
$\inst(M_1) = \{(1, \ldots, 1; 1, \ldots, 1)\}$. Let $\nu_1$ be the valuation that
is $1$ everywhere. Clearly, $(\nu_1, \nu_1) \in \sem{\alpha}$ since 
$(\nu_1, \nu_1) \in \sem{\alpha_1}$ and $\nu_1(x) = \nu_1(y)$. Take $\nu_1'  := \nu_1[v:2]$.
We establish that $v \in \semI(\alpha)$ by arguing that there is no $\nu_2'$ for which
$(\nu_1', \nu_2') \in \sem{\alpha}$. In particular, when $v \in \inn1$, it is clear that
there is no $\nu_2'$ such that $(\nu_1', \nu_2') \in \sem{\alpha_1}$. Now we need to verify
the same when $v$ is $x$ or $y$ and $v \not \in \inn1$. Thereto, suppose 
$(\nu_1', \nu_2') \in \sem{\alpha}$. In the case of $v$ is $x$ and $x \not \in \inn1$, this
is only possible when $x \not \in \outt1$. Therefore, $\nu_2'(x) = \nu_1'(x) = 2$ but
$\nu_2'(y) = 1$ whether $y \in \outt1$ or not. Hence, 
$(\nu_1', \nu_2') \not \in \sem{\alpha}$ since $x \neqv y$ and $\nu_2'(x) \neq \nu_2'(y)$.
The case when $v$ is $y$ and $y \not \in \inn1$ is symmetric. Consequently,
$v \in \semI(\alpha)$.

\subsection{Left-to-Right Selection}
Let $\alpha$ be of the form $\sellr{x = y}({\alpha_1})$, where
$\alpha_1 := M(\bar{u}; \bar{w})$, $\bar{u} = u_1, \ldots, u_n$, and 
$\bar{w} = w_1, \ldots, w_m$. We distinguish different cases based on whether
$x \eqv y$ and $y \in \outt1$.

\paragraph{$x \eqv y$ and $y \in \outt1$} Recall the definitions in this case:
\begin{itemize}
	\item $\synO(\alpha) = \outt1 - \{x\}$.
	\item $\synI(\alpha) = \inn1 \cup \{x\}$.
\end{itemize}
In what follows, since $x \eqv y$ we will refer to both of them with $x$. We first 
proceed to verify $\synO(\alpha) \subseteq \semO(\alpha)$. Let $v \in \outt1 - \{x\}$. 
Consider an \interpretation $\inst$ such that 
$\inst(M) = \{(1, \ldots, 1; o_1, \ldots, o_m)\}$ where $o_j = 1$ if $w_j = y$, otherwise
$o_j = 2$ . Let $\nu_1$ be the valuation that is $1$ everywhere. Also let $\nu_2$ be
the valuation that is $2$ on $\outt1 - \{x\}$ and $1$ everywhere else. Clearly, 
$(\nu_1, \nu_2) \in \sem{\alpha}$ since $(\nu_1, \nu_2) \in \sem{\alpha_1}$ and 
$\nu_1(x) = \nu_2(x)$. Therefore, $v \in \semO(\alpha)$ since $\nu_1(v) \neq \nu_2(v)$.

Now we proceed to verify $\synI(\alpha) \subseteq \semI(\alpha)$. Let 
$v \in \inn1 \cup \{x\}$. Consider an \interpretation $\inst$ where 
$\inst(M) = \{(1, \ldots, 1; 1, \ldots, 1)\}$. Let $\nu_1$ be the valuation that is
$1$ everywhere. Clearly, $(\nu_1, \nu_1) \in \sem{\alpha}$ since
$(\nu_1, \nu_1) \in \sem{\alpha_1}$. Take $\nu_1'  := \nu_1[v:2]$. We establish that 
$v \in \semI(\alpha)$ by arguing that there is no $\nu_2'$ for which 
$(\nu_1', \nu_2') \in \sem{\alpha}$. Thereto, suppose that 
$(\nu_1', \nu_2') \in \sem{\alpha}$. In particular, when $v \in \inn1$ it is clear 
that $(\nu_1', \nu_2') \not \in \sem{\alpha_1}$. On the other hand, when $v = x$ and
$x \in \outt1 - \inn1$, clearly $\nu_1'(x) = 2 \neq 1 = \nu_2'(x)$. Consequently,
$v \in \semI(\alpha)$.

\paragraph{$x \eqv y$ and $y \not \in \outt1$} Recall the definitions in this case:
\begin{itemize}
	\item $\synO(\alpha) = \outt1$.
	\item $\synI(\alpha) = \inn1$.
\end{itemize}
In what follows, since $x \eqv y$ we will refer to both of them with $x$. We proceed to 
verify $\synO(\alpha) \subseteq \semO(\alpha)$ and $\synI(\alpha) \subseteq \semI(\alpha)$.
Indeed, this is true since $\sem{\alpha} = \sem{\alpha_1}$ for any \interpretation 
$\inst$ because of $x \eqv y$ and $x \not \in \outt1$.

\paragraph{$x \neqv y$ and $y \in \outt1$} Recall the definitions in this case:
\begin{itemize}
	\item $\synO(\alpha) = \outt1$.
	\item $\synI(\alpha) = \inn1 \cup \{x\}$.
\end{itemize}
We first proceed to verify $\synO(\alpha) \subseteq \semO(\alpha)$. Let $v \in \outt1$.
Consider an \interpretation $\inst$ such that 
$\inst(M) = \{(i_1, \ldots, i_n; o_1, \ldots, o_m)\}$ where $i_j = 2$ if $u_j = x$, 
otherwise $i_j = 1$. Also, $o_j = 3$ if $w_j = x$, otherwise $o_j = 2$ . Let $\nu_1$ be 
the valuation that is $2$ on $x$ and $1$ everywhere else. Also let $\nu_2$ be the 
valuation that is $2$ on $\outt1 - \{x\}$, $3$ on $x$ if $x \in \outt1$ and agrees with 
$\nu_1$ everywhere else. Clearly, $(\nu_1, \nu_2) \in \sem{\alpha}$ since 
$(\nu_1, \nu_2) \in \sem{\alpha_1}$ and $\nu_1(x) = \nu_2(y)$. Therefore,
$v \in \semO(\alpha)$. Indeed, in both cases whether $x \in \outt1$ or not, 
$\nu_1(v) \neq \nu_2(v)$.

Now we proceed to verify $\synI(\alpha) \subseteq \semI(\alpha)$. Let 
$v \in \inn1 \cup \{x\}$. Consider an \interpretation $\inst$ where 
$\inst(M) = \{(1, \ldots, 1; 1, \ldots, 1)\}$. Let $\nu_1$ be the valuation that is 
$1$ everywhere. Clearly, $(\nu_1, \nu_1) \in \sem{\alpha}$ since 
$(\nu_1, \nu_1) \in \sem{\alpha_1}$. Take $\nu_1'  := \nu_1[v:2]$. We establish that
$v \in \semI(\alpha)$ by arguing that there is no $\nu_2'$ for which 
$(\nu_1', \nu_2') \in \sem{\alpha}$. Thereto, suppose that 
$(\nu_1', \nu_2') \in \sem{\alpha}$. In particular, when $v \in \inn1$ it is clear
that $(\nu_1', \nu_2') \not \in \sem{\alpha_1}$. On the other hand, when $v = x$ and
$y \in \outt1$, clearly $\nu_1'(x) = 2 \neq 1 = \nu_2'(y)$. Consequently, 
$v \in \semI(\alpha)$.

\paragraph{$x \neqv y$ and $y \not \in \outt1$} Recall the definitions in this case:
\begin{itemize}
	\item $\synO(\alpha) = \outt1$.
	\item $\synI(\alpha) = \inn1 \cup \{x, y\}$.
\end{itemize}
We first proceed to verify $\synO(\alpha) \subseteq \semO(\alpha)$. Let $v \in \outt1$.
Consider an \interpretation $\inst$ such that $\inst(M) = \{(1, \ldots, 1; 2, \ldots, 2)\}$.
Let $\nu_1$ be the valuation that is $1$ everywhere. Also let $\nu_2$ be the valuation that 
is $2$ on $\outt1$ and $1$ everywhere else. Clearly, $(\nu_1, \nu_2) \in \sem{\alpha}$ since
$(\nu_1, \nu_2) \in \sem{\alpha_1}$ and $\nu_1(x) = \nu_2(y)$. Indeed, this is true since 
$y \not \in \outt1$, then $\nu_1(y) = \nu_2(y)$. Therefore, $v \in \semO(\alpha)$ since
$\nu_1(v) \neq \nu_2(v)$.

Now we proceed to verify $\synI(\alpha) \subseteq \semI(\alpha)$. Let
$v \in \inn1 \cup \{x, y\}$. Consider an \interpretation $\inst$ where
$\inst(M) = \{(1, \ldots, 1; 1, \ldots, 1)\}$. Let $\nu_1$ be the valuation that is
$1$ everywhere. Clearly, $(\nu_1, \nu_1) \in \sem{\alpha}$ since 
$(\nu_1, \nu_1) \in \sem{\alpha_1}$. Take $\nu_1'  := \nu_1[v:2]$. We establish that
$v \in \semI(\alpha)$ by arguing that there is no $\nu_2'$ for which 
$(\nu_1', \nu_2') \in \sem{\alpha}$. Thereto, suppose that 
$(\nu_1', \nu_2') \in \sem{\alpha}$. In particular, when $v \in \inn1$ it is clear that
$(\nu_1', \nu_2') \not \in \sem{\alpha_1}$. On the other hand, when $v = x$ or $v = y$,
clearly $\nu_1'(x) \neq (\nu_1'(y) = \nu_2'(y))$ since $y \not \in \outt1$ and 
$x \neqv y$. Consequently, $v \in \semI(\alpha)$.
}

\section{Precision Theorem Proof}\label{app:precision}
In this section, we prove Theorem~\ref{thm:precision}.  
By soundness and Proposition~\ref{prop:inputApp}, it suffices to prove 
$\synO(\alpha) \subseteq \semO(\alpha)$ and 
$\synI(\alpha) \subseteq \semI(\alpha)$ for every LIF expression $\alpha$.
For the latter inequality, it will be convenient to use the equivalent definition 
of semantic input variables introduced in Proposition~\ref{prop:inputEqu}.
% \begin{definition}[Semantic Inputs]
% $x$ is a semantic input of $\alpha$, if there is an \interpretation $\inst$, a 
% value $d \in \dom$, and $(\nu_1, \nu_2) \in \sem{\alpha}$ such that there is no 
% valuation $\nu_2'$ that agrees with $\nu_2$ on the outputs of $\alpha$ such that 
% $(\nu_1[x:d], \nu_2') \in \sem{\alpha}$.
% \end{definition}
Moreover, in the proof of the Precision Theorem, we will often make use of 
the following two technical lemmas.
\begin{lemma}\label{lem:null_id}
Let $M$ be a nullary relation name and let $\inst$ be an \interpretation where 
$\inst(M)$ is nonempty. Then $\sem{M()}$ consists of all identical pairs of 
valuations.
\end{lemma}
\begin{proof}
The proof follows directly from the semantics of atomic modules.
\end{proof}
\begin{lemma}\label{lem:reduction}
Suppose $\alpha_1 = M_1(\bar{x_1}; \bar{y_1})$ and 
$\alpha_2 = M_2(\bar{x_2}; \bar{y_2})$ where $M_1 \neq M_2$.  Let $\alpha$ be 
either $\alpha_1 \cup \alpha_2$ or $\alpha_1 - \alpha_2$.  Assume that 
$\synO(\alpha_i) \subseteq \semO(\alpha_i)$ and 
$\synI(\alpha_i) \subseteq \semI(\alpha_i)$ for $i = 1, 2$. 
Let $j \neq k \in \{1, 2\}$.  If $\semm{\alpha}{\inst} = \semm{\alpha_k}{\inst}$ 
for any \interpretation $\inst$ where $\inst(M_j)$ is empty, then 
$\synO(\alpha_k) \subseteq \semO(\alpha)$ and 
$\synI(\alpha_k) \subseteq \semI(\alpha)$.
\end{lemma}
\begin{proof}
First, we verify that $\synO(\alpha_k) \subseteq \semO(\alpha)$.  Let 
$v \in \synO(\alpha_k)$.  Since $\synO(\alpha_k) \subseteq \semO(\alpha_k)$, 
then $v \in \semO(\alpha_k)$.  By definition, we know that there is an 
\interpretation $\inst'$ and $(\nu_1, \nu_2) \in \semm{\alpha_k}{\inst'}$ such 
that $\nu_1(v) \neq \nu_2(v)$.  Take $\inst''$ to be the interpretation where 
$\inst''(M) = \inst'(M)$ for any $M \neq M_j$ while $\inst''(M_j)$ is empty.  
Clearly, $(\nu_1, \nu_2) \in \semm{\alpha}{\inst''}$, whence, $M_j \neq M_k$ 
and $\semm{\alpha}{\inst''} = \semm{\alpha_k}{\inst'}$.  It follows then that 
$v \in \semO(\alpha)$.

Similarly, we proceed to verify $\synI(\alpha_k) \subseteq \semI(\alpha)$.  Let 
$v \in \synI(\alpha_k)$. Since $\synI(\alpha_k) \subseteq \semI(\alpha_k)$, then 
$v \in \semI(\alpha_k)$. By definition, we know that there is an \interpretation 
$\inst'$, $(\nu_1, \nu_2) \in \semm{\alpha_k}{\inst'}$, and 
$\nu_1'(v) \neq \nu_1(v)$ such that 
$(\nu_1', \nu_2') \not \in \semm{\alpha_k}{\inst'}$ for every valuation $\nu_2'$ 
that agrees with $\nu_2$ on $\semO(\alpha_k)$.

Take $\inst''$ to be the interpretation where $\inst''(M) = \inst'(M)$ for any 
$M \neq M_j$ while $\inst''(M_j)$ is empty. Clearly, 
$\semm{\alpha}{\inst''} = \semm{\alpha_k}{\inst'}$, whence, $M_j \neq M_k$. 
Therefore, $\semO(\alpha_k) \subseteq \semO(\alpha)$.  Hence, $v \in \semI(\alpha)$. 
Indeed, $(\nu_1, \nu_2) \in \semm{\alpha}{\inst''}$ and for any valuation $\nu_2'$ 
if $\nu_2'$ agrees with $\nu_2$ on $\semO(\alpha)$, then $\nu_2'$ agrees with $\nu_2$ 
on $\semO(\alpha_k)$.
\end{proof}
The proof of Theorem~\ref{thm:precision} is done by extensive case analysis. 
Intuitively, for each of the different operations, and every variable 
$z\in \synO(\alpha)$, we construct an \interpretation $\inst$ such that $z$ 
is not inertial in $\sem{\alpha}$ and thus $z \in \semO(\alpha)$.  Similarly, 
for every variable $z \in \synI(\alpha)$, we construct an \interpretation 
$\inst$ as a witness of the fact that $\V \setminus \{z\}$ does not determine 
$\semNoI{\alpha}$ on $\semO(\alpha)$ and thus that $z\in \semI(\alpha)$.  
In the proof, we often remove the superscript from $\synI$ and $\synO$ and refer 
to them by $\In$ and $\Out$, respectively.

\subsection{Atomic Modules}
Let $\alpha$ be of the form $\alpha_1$, where $\alpha_1$ is $M(\bar{x}; \bar{y})$. 
Recall the definition:
\begin{itemize}
	\item $\synO(\alpha) = Y$, where $Y$ are the variables in $\bar{y}$.
	\item $\synI(\alpha) = X$, where $Y$ are the variables in $\bar{x}$.
\end{itemize}
We first proceed to verify $\synO(\alpha) \subseteq \semO(\alpha)$.  Let $v \in Y$. 
Consider an \interpretation $\inst$ where 
\[\inst(M) = \{(1, \ldots, 1; 2, \ldots, 2)\}.\]
Let $\nu_1$ be the valuation that is $1$ everywhere. Also let $\nu_2$ be the valuation 
that is $2$ on $Y$ and $1$ everywhere else. Clearly, $(\nu_1, \nu_2) \in \sem{\alpha}$ 
since $\nu_1(\bar{x}) \cdot \nu_2(\bar{y}) \in \inst(M)$ and $\nu_1$ agrees with $\nu_2$ 
outside $Y$. Hence, $v \in \semO(\alpha)$ since $\nu_1(v) \neq \nu_2(v)$.

Now we proceed to verify $\synI(\alpha) \subseteq \semI(\alpha)$.  Let $v \in X$.  
Consider the same \interpretation $\inst$ and the same valuations $\nu_1$ and $\nu_2$ 
as discussed above. We already established that $(\nu_1, \nu_2) \in \sem{\alpha}$. 
Take $\nu_1' := \nu_1[v:2]$. We establish that $v \in \semI(\alpha)$ by arguing that 
there is no $\nu_2'$ for which $(\nu_1', \nu_2') \in \sem{\alpha}$.  Indeed, this is 
true since $v \in X$. Consequently, $v \in \semI(\alpha)$.

\subsection{Identity}
Let $\alpha$ be of the form $\id$. We recall that $\synI(\alpha)$ and $\synO(\alpha)$ 
are both empty. Hence, $\synO(\alpha) \subseteq \semO(\alpha)$ and 
$\synI(\alpha) \subseteq \semI(\alpha)$ is trivial.

\subsection{Union}
Let $\alpha$ be of the form $\alpha_1 \cup \alpha_2$, where $\alpha_1$ is 
$M_1(\bar{x_1}; \bar{y_1})$ and $\alpha_2$ is $M_2(\bar{x_2}; \bar{y_2})$.  We 
distinguish different cases based on whether $M_1$ or $M_2$ is nullary.  If $M_1$ 
and $M_2$ are both nullary there is nothing to prove.

\subsubsection{$M_1$ is nullary, $M_2$ is not} Clearly, $\sem{\alpha} = \sem{\alpha_2}$ 
for any \interpretation $\inst$ where $\inst(M_1)$ is empty. By induction and 
Lemma~\ref{lem:reduction}, we establish that $\synO(\alpha_2) \subseteq \semO(\alpha)$ 
and $\synI(\alpha_2) \subseteq \semI(\alpha)$. Since $\inn1$ and $\outt1$ are both empty, 
then we observe that:
\begin{itemize}
	\item $\synO(\alpha) = \outt2$.
	\item $\synI(\alpha) = \inn2 \cup \outt2$.
\end{itemize}
Thus, $\outt2 \subseteq \semO(\alpha)$ and $\inn2 \subseteq \semI(\alpha)$ is trivial.

We proceed to verify $\outt2 - \inn2 \subseteq \semI(\alpha)$. Let $v \in \outt2 - \inn2$. 
Consider the \interpretation $\inst$ where $\inst(M_1)$ is not empty and 
\[
\inst(M_2) = \{(1, \ldots, 1; 2, \ldots, 2) \}
\]
Let $\nu_1$ be the valuation that is $1$ everywhere.  Clearly, 
$(\nu_1, \nu_1) \in \sem{\alpha}$ since $(\nu_1, \nu_1) \in \sem{\alpha_1}$ by 
Lemma~\ref{lem:null_id}. Take $\nu_1' := \nu_1[v:2]$.  We establish that 
$v \in \semI(\alpha)$ by arguing that there for every valuation $\nu_2'$ for 
which $(\nu_1', \nu_2') \in \sem{\alpha}$, we show that $\nu_2'$ and $\nu_1$ 
disagrees on $\semO(\alpha)$. Thereto, suppose $(\nu_1', \nu_2') \in \sem{\alpha}$. 
In particular, $(\nu_1', \nu_2') \in \sem{\alpha_1}$, whence $\nu_2' = \nu_1'$ by 
Lemma~\ref{lem:null_id}.  Indeed, $v \in \outt2$, $\outt2 \subseteq \semO(\alpha)$, 
and $\nu_2'(v) = 2$ but $\nu_1(v) = 1$. Otherwise, $(\nu_1', \nu_2') \in \sem{\alpha_2}$. 
However, since $v \in \outt2$, then $\nu_2'(v) = 2$ as well.  Therefore, there is no 
$\nu_2'$ that agrees with $\nu_1$ on $\semO(\alpha)$ and $(\nu_1', \nu_2') \in \sem{\alpha}$ 
at the same time. We conclude that $v \in \semI(\alpha)$.

\subsubsection{$M_2$ is nullary, $M_1$ is not} This case is symmetric to the previous one.

\subsubsection{Neither $M_1$ nor $M_2$ is nullary}
Recall the definitions:
\begin{itemize}
	\item $\synO(\alpha) = \outt 1 \cup \outt 2$.
	\item $\synI(\alpha) = \inn 1 \cup \inn 2 \cup (\outt 1 \symdif \outt 2)$.
\end{itemize}
We first proceed to verify that $\synO(\alpha) \subseteq \semO(\alpha)$.  By induction, 
$\synO(\alpha_i) = \semO(\alpha_i)$ for $i = 1$ or $2$.  Consequently, if $v \in \outt i$, 
then there is an \interpretation $\inst_i$, $(\nu_1, \nu_2) \in \sem{\alpha_i}$ such that 
$\nu_1(v) \neq \nu_2(v)$.  Indeed, $v \in \semO(\alpha)$, whence, 
$(\sem{\alpha_1} \cup \sem{\alpha_2}) \subseteq \sem{\alpha}$ for any \interpretation $\inst$.

We then proceed to verify that $\synI(\alpha) \subseteq \semO(\alpha)$.  The proof has 
four possibilities. Each case is discussed separately below.

\paragraph{When $v \in \inn1$}
If $v \in \inn1$ and $M_1 \neq M_2$, it is clear that $\sem{\alpha} = \sem{\alpha_1}$
for any \interpretation $\inst$ where $\inst(M_2)$ is empty.
By Lemma~\ref{lem:reduction} and by induction, we easily establish that 
$v \in \semI(\alpha)$.

\paragraph{When $v \in \inn2$} This case is symmetric to the previous one.

\paragraph{When $v \in \outt 1 - (\outt 2 \cup \inn 1 \cup \inn 2)$}
If $v \in \outt 1 - (\outt 2 \cup \inn 1 \cup \inn 2)$, then consider the 
\interpretation $\inst$ such that $\inst(M_1)  = \{(1, \ldots, 1; 1, \ldots, 1)\}$ 
and $\inst(M_2)  = \{(1, \ldots, 1; 1, \ldots, 1)\}$. Let $\nu_1$ be the valuation that 
is $2$ on $v$ and $1$ elsewhere. Clearly, $(\nu_1, \nu_1) \in \sem{\alpha}$, whence 
$(\nu_1, \nu_1) \in \sem{\alpha_2}$.	Now take $\nu_1' := \nu_1[v:1]$. If we can show 
that $\nu_2'$ does not agree with $\nu_1$ on $\semO(\alpha)$ for any valuation $\nu_2'$ 
such that $(\nu_1', \nu_2') \in \sem{\alpha}$, we are done. Thereto, suppose that there 
exists a valuation $\nu_2'$ such that $(\nu_1', \nu_2') \in \sem{\alpha}$.
\begin{itemize}
	\item In particular, if $(\nu_1', \nu_2') \in \sem{\alpha_1}$, then $\nu_2'(v) = 1$ 
	since $v \in \outt1$.
	\item Otherwise, if $(\nu_1', \nu_2') \in \sem{\alpha_2}$, then 
	$\nu_2'(v) = \nu_1'(v) = 1$ since $v \not \in (\inn2 \cup \outt2)$.
\end{itemize}
In both cases, $\nu_2'$ have to be $1$ on $v$ which disagrees with $\nu_1$ on $v$.  Since 
$v \in \outt1$ and $\outt1 \subseteq \semO(\alpha)$, then $\nu_2'$ does not agree with 
$\nu_1$ on $\semO(\alpha)$ as desired. We conclude that $v \in \semI(\alpha)$.

\paragraph{When $v \in \outt 2 - (\outt 1 \cup \inn 1 \cup \inn 2)$} This case is symmetric 
to the previous one.

\subsection{Intersection}
Let $\alpha$ be of the form $\alpha_1 \cap \alpha_2$, where $\alpha_1$ is 
$M_1(\bar{x_1}; \bar{y_1})$ and $\alpha_2$ is $M_2(\bar{x_2}; \bar{y_2})$.  We distinguish 
different cases based on whether $M_1$ or $M_2$ is nullary. If $M_1$ and $M_2$ are both 
nullary there is nothing to prove.

\subsubsection{$M_1$ is nullary, $M_2$ is not}
In this case, $\inn1$ and $\outt1$ are both empty, then we observe that:
\begin{itemize}
	\item $\synO(\alpha) = \emptyset$.
	\item $\synI(\alpha) = \inn2 \cup \outt2$.
\end{itemize}
It is trivial to verify $\synO(\alpha) \subseteq \semO(\alpha)$ since $\synO(\alpha)$ 
is empty.

We proceed to verify $\synI(\alpha) \subseteq \semI(\alpha)$. Let $v \in \inn2 \cup \outt2$. 
Consider an \interpretation $\inst$ where $\inst(M_1)$ is not empty and 
$\inst(M_2) = \{(1, \ldots, 1; 1, \ldots, 1)\}$.  Let $\nu_1$ be the valuation that is 
$1$ everywhere.  Clearly, $(\nu_1, \nu_1) \in \sem{\alpha}$ since 
$(\nu_1, \nu_1) \in \sem{\alpha_1}$ and $(\nu_1, \nu_1) \in \sem{\alpha_2}$.  Take 
$\nu_1' := \nu_1[v:2]$. We establish that $v \in \semI(\alpha)$ by arguing that there 
is no valuation $\nu_2'$ for which $(\nu_1', \nu_2') \in \sem{\alpha}$.  Thereto, suppose 
$(\nu_1', \nu_2') \in \sem{\alpha}$.  In particular, when $v \in \inn2$, it is clear that 
$(\nu_1', \nu_2') \not \in \sem{\alpha_1}$.  In the other case when $v \in \outt2 - \inn2$,
there is no $\nu_2'$ such that $(\nu_1', \nu_2')$ belongs to both $\sem{\alpha_1}$ and 
$\sem{\alpha_2}$. Indeed, the value for $\nu_2'(v)$ will never be agreed upon by $\alpha_1$ 
and $\alpha_2$.  Hence, $(\nu_1', \nu_2') \not \in \sem{\alpha}$ as desired.  We conclude 
that $v \in \semI(\alpha)$.

\subsubsection{$M_2$ is nullary, $M_1$ is not} This case is symmetric to the previous one.

\subsubsection{Neither $M_1$ nor $M_2$ is nullary}
Recall the definitions:
\begin{itemize}
	\item $\synO(\alpha) = \outt 1 \cap \outt 2$.
	\item $\synI(\alpha) = \inn 1 \cup \inn 2 \cup (\outt 1 \symdif \outt 2)$.
\end{itemize}

We first proceed to verify $\synO(\alpha) \subseteq \semO(\alpha)$.  Let 
$v \in \outt1 \cap \outt2$.  Consider an \interpretation $\inst$ such that 
\[
\inst(M_1) = \{(1, \ldots, 1; o_1, \ldots, o_m)\}
\]
, where $o_1, \ldots, o_m$ are all the combinations of $\{1, 2\}$.  Similarly, 
$\inst(M_2) = \{(1, \ldots, 1; o_1, \ldots, o_n)\}$, where $o_1, \ldots, o_n$ are 
all the combinations of $\{1, 2\}$.

Let $\nu_1$ be the valuation that is $1$ everywhere. Also let $\nu_2$ be the valuation 
that is $2$ on $\outt1 \cap \outt2$ and $1$ elsewhere.  Clearly, 
$(\nu_1, \nu_2) \in \sem{\alpha}$, whence $(\nu_1, \nu_2) \in \sem{\alpha_1}$ and 
$(\nu_1, \nu_2) \in \sem{\alpha_2}$.  Hence, $v \in \semO(\alpha)$.  Indeed, 
$\nu_2(v) \neq \nu_1(v)$ for $v \in \synO(\alpha)$.

We then proceed to verify $\synI(\alpha) \subseteq \semO(\alpha)$.  Let 
$v \in \inn1 \cup \inn2 \cup (\outt1 \symdif \outt2)$.  Consider an \interpretation 
$\inst$ where 
\[
\inst(M_1) = \{(1, \ldots, 1; 1, \ldots, 1)\}
\] 
Similarly, $\inst(M_2) = \{(1, \ldots, 1; 1, \ldots, 1)\}$.  Let $\nu_1$ be the valuation 
that is $1$ everywhere. Clearly, $(\nu_1, \nu_1) \in \sem{\alpha}$, whence 
$(\nu_1, \nu_1) \in \sem{\alpha_1}$ and $(\nu_1, \nu_1) \in \sem{\alpha_2}$.  Take 
$\nu_1' := \nu_1[v:2]$. We establish that $v \in \semI(\alpha)$ by arguing that there 
is no valuation $\nu_2'$ such that $(\nu_1', \nu_2') \in \sem{\alpha}$.  Indeed, this is 
clear when $v \in \inn1$ or $v \in \inn2$.  On the other hand, when 
$v \in (\outt1 \symdif \outt2) - (\inn1 \cup \inn2)$, we have $\nu_2'$ for which 
$(\nu_1', \nu_2') \in \sem{\alpha}$ whence $(\nu_1', \nu_2') \in \sem{\alpha_1}$ and 
$(\nu_1', \nu_2') \in \sem{\alpha_2}$. This is not possible since $v$ belongs to either 
$\outt1$ or $\outt2$, but not both. Hence, the value for $\nu_2'(v)$ will never be 
agreed upon by $\alpha_1$ and $\alpha_2$. We conclude that $v \in \semI(\alpha)$.

\subsection{Difference}
Let $\alpha$ be of the form $\alpha_1 \setminus \alpha_2$, where $\alpha_1$ is 
$M_1(\bar{x_1}; \bar{y_1})$ and $\alpha_2$ is $M_2(\bar{x_2}; \bar{y_2})$.  We 
distinguish different cases based on whether $M_1$ or $M_2$ is nullary.  If $M_1$ 
and $M_2$ are both nullary there is nothing to prove.

\subsubsection{$M_1$ is nullary, $M_2$ is not} In this case, $\inn1$ and $\outt 1$ are 
empty. In particular, $\synO(\alpha)$ is empty, so $\synO(\alpha) \subseteq \semO(\alpha)$ 
is trivial.

We proceed to verify $\synI(\alpha) \subseteq \semI(\alpha)$.  Observe that
\[
\synI(\alpha) = \inn2 \cup \outt2.
\]
Let $v \in \synI(\alpha)$. Consider the \interpretation $\inst$ where $\inst(M_1)$ is 
not empty and $\inst(M_2) = \{(1, \ldots, 1; 1, \ldots, 1) \}$.
Let $\nu_1$ be the valuation that is $2$ on $v$ and $1$ elsewhere.  Clearly, 
$(\nu_1, \nu_1) \in \sem{\alpha}$. Take $\nu_1' := \nu_1[v:1]$.  We establish that 
$v \in \semI(\alpha)$ by arguing that there is no $\nu_2'$ for which 
$(\nu_1', \nu_2') \in \sem{\alpha}$.  Thereto, suppose $(\nu_1', \nu_2') \in \sem{\alpha}$. 
In particular, $(\nu_1', \nu_2') \in \sem{\alpha_1}$, whence $\nu_2' = \nu_1'$ by Lemma 
\ref{lem:null_id}. However, $(\nu_1', \nu_1') \in \sem{\alpha_2}$, so 
$(\nu_1', \nu_2') \not \in \sem{\alpha}$ as desired.

\subsubsection{$M_2$ is nullary, $M_1$ is not} Clearly, $\sem{\alpha} = \sem{\alpha_1}$ for any
\interpretation $\inst$ where $\inst(M_2)$ is empty. 
By induction and Lemma \ref{lem:reduction}, we establish that 
$\synO(\alpha_1) \subseteq \semO(\alpha)$ and $\synI(\alpha_1) \subseteq \semI(\alpha)$. 
Since $\inn2$ and $\outt2$ are both empty, then we observe that:
\begin{itemize}
	\item $\synO(\alpha) = \outt1$.
	\item $\synI(\alpha) = \inn1 \cup \outt1$.
\end{itemize}
Thus, $\outt1 \subseteq \semO(\alpha)$ and $\inn1 \subseteq \semI(\alpha)$ is trivial.

We proceed to verify $\outt1 - \inn1 \subseteq \semI(\alpha)$.  Let
$v \in \outt1 - \inn1$.  Consider the \interpretation $\inst$ where $\inst(M_2)$ is 
not empty and 
\[
\inst(M_1) = \{(1, \ldots, 1; 1, \ldots, 1) \}
\]
Let $\nu_1$ be the valuation that is $2$ on $v$ and $1$ elsewhere and let $\nu_2$ be 
the valuation that is $1$ everywhere. Clearly, $(\nu_1, \nu_2) \in \sem{\alpha}$.  Take 
$\nu_1' := \nu_1[v:1]$.  We establish that $v \in \semI(\alpha)$ by arguing that there 
is no $\nu_2'$ for which $(\nu_1', \nu_2') \in \sem{\alpha}$.  Thereto, suppose 
$(\nu_1', \nu_2') \in \sem{\alpha}$.  In particular, 
$(\nu_1', \nu_2') \in \sem{\alpha_1}$, whence $\nu_2' = \nu_1'$ from the structure of 
$\inst$. However, $(\nu_1', \nu_1') \in \sem{\alpha_2}$ by Lemma \ref{lem:null_id}, 
so $(\nu_1', \nu_2') \not \in \sem{\alpha}$ as desired.

\subsubsection{Neither $M_1$ nor $M_2$ is nullary} Recall the definitions:
\begin{itemize}
	\item $\synO(\alpha) = \outt 1$.
	\item $\synI(\alpha) = \inn 1 \cup \inn 2 \cup (\outt 1 \symdif \outt 2)$.
\end{itemize}
The proof of $\synO(\alpha) \subseteq \semO(\alpha)$ is done together with the proof 
that $v \in \semI(\alpha)$ for every $v \in \inn1$.  Discussions for the other cases 
for $v \in \synI(\alpha)$ follow afterwards. Since $M_1 \neq M_2$, it is clear that 
$\sem{\alpha} = \sem{\alpha_1}$ for any \interpretation $\inst$ where $\inst(M_2)$ is 
empty.  By induction and Lemma \ref{lem:reduction}, we establish that 
$\synO(\alpha_1) \subseteq \semO(\alpha)$ and $\synI(\alpha_1) \subseteq \semI(\alpha)$. 
Thus, $\outt1 \subseteq \semO(\alpha)$ and $\inn1 \subseteq \semI(\alpha)$ is trivial.

\paragraph{When $v \in \inn2 - \inn1$}
Let $v \in \inn2 - \inn1$. Consider an \interpretation $\inst$ where 
\[\inst(M_1) = \{(1, \ldots, 1; 1, \ldots, 1)\}\] 
and similarly $\inst(M_2)$ $= \{(1, \ldots, 1; 1, \ldots, 1)\}$. Let $\nu_1$ be the 
valuation that $2$ on $v$ and $1$ elsewhere. Also, let $\nu_2$ be the valuation that is 
$1$ on $\outt1$ and agrees with $\nu_1$ everywhere else. Clearly,
$(\nu_1, \nu_2) \in \sem{\alpha_1}$. Further, $(\nu_1, \nu_2) \not \in \sem{\alpha_2}$.
Indeed, since $v \in \inn2$ then $\nu_1$ should have the value of $1$ on $v$ for
$(\nu_1, \nu_2)$ to be in $\sem{\alpha_2}$. Take $\nu_1' := \nu_1[v:1]$. We establish 
that $v \in \semI(\alpha)$ by arguing that there is no $\nu_2'$ for which
$(\nu_1', \nu_2') \in \sem{\alpha}$. Thereto, suppose that 
$(\nu_1', \nu_2') \in \sem{\alpha}$. Hence, $(\nu_1', \nu_2') \in \sem{\alpha_1}$ and 
$(\nu_1', \nu_2') \in \sem{\alpha_2}$. Indeed, $(\nu_1', \nu_2') \in \sem{\alpha_1}$ 
whence $\nu_1' = \nu_2'$. Clearly, $(\nu_1', \nu_1') \in \sem{\alpha_2}$ showing that
$(\nu_1', \nu_1') \not \in \sem{\alpha}$ as desired. Therefore, $v \in \semI(\alpha)$.

\paragraph{When $v \in (\outt1 \symdif \outt2) - (\inn1 \cup \inn2)$}
Let $v \in (\outt1 \symdif \outt2) - (\inn1 \cup \inn2)$. Consider an \interpretation 
$\inst$ where $\inst(M_1) = \{(1, \ldots, 1; 1, \ldots, 1)\}$ and 
$\inst(M_2) = \{(1, \ldots, 1; 1, \ldots, 1)\}$. Let $\nu_1$ be the valuation that is $2$ 
on $v$ and $1$ elsewhere. Also let $\nu_2$ be the valuation that is $1$ on $\outt1$ and
agrees with $\nu_1$ everywhere else. Clearly, $(\nu_1, \nu_2) \in \sem{\alpha_1}$.
Furthermore, $(\nu_1, \nu_2) \not \in \sem{\alpha_2}$. In particular, when 
$v \in \outt1 - (\inn1 \cup \inn2 \cup \outt2)$, we know that $\nu_1(v) = 2$ and 
$\nu_2(v) = 1$. Since $v \not \in \outt2$, then $(\nu_1, \nu_2) \not \in \sem{\alpha_2}$.
In the other case, when $v \in \outt2 - (\inn1 \cup \inn2 \cup \outt1)$, we know that
$\nu_1(v) = \nu_2(v) = 2$ since $(\nu_1, \nu_2) \in \sem{\alpha_1}$. Consequently, 
$(\nu_1, \nu_2) \not \in \sem{\alpha_2}$ since $v \in \outt2$ but $\nu_2(v) = 2$. We 
verify that $(\nu_1, \nu_2) \in \sem{\alpha}$. Take $\nu_1' := \nu_1[v:1]$. We establish 
that $v \in \semI(\alpha)$ by arguing that there is no $\nu_2'$ for which 
$(\nu_1', \nu_2') \in \sem{\alpha}$. Thereto, suppose that 
$(\nu_1', \nu_2') \in \sem{\alpha}$. Hence, $(\nu_1', \nu_2') \in \sem{\alpha_1}$ and
$(\nu_1', \nu_2') \in \sem{\alpha_2}$. Indeed, $(\nu_1', \nu_2') \in \sem{\alpha_1}$
whence $\nu_1' = \nu_2'$. Clearly, $(\nu_1', \nu_1') \in \sem{\alpha_2}$ showing that
$(\nu_1', \nu_1') \not \in \sem{\alpha}$ as desired. Therefore, $v \in \semI(\alpha)$.

\subsection{Composition}
Let $\alpha$ be of the form $\alpha_1 \comp \alpha_2$, where $\alpha_1$ is 
$M_1(\bar{x_1}; \bar{y_1})$ and $\alpha_2$ is $M_2(\bar{x_2}; \bar{y_2})$. We
distinguish different cases based on whether $M_1$ or $M_2$ is nullary. If $M_1$ and 
$M_2$ are both nullary there is nothing to prove.

\subsubsection{$M_1$ is nullary, $M_2$ is not} Clearly, $\sem{\alpha} = \sem{\alpha_2}$ for 
any \interpretation $\inst$ where $\inst(M_1)$ is not empty. In this case, $\inn1$ and 
$\outt1$ are both empty, then we observe that:
\begin{itemize}
	\item $\synO(\alpha) = \outt2$.
	\item $\synI(\alpha) = \inn2$.
\end{itemize}
First, we verify $\synO(\alpha) \subseteq \semO(\alpha)$. Let $v \in \outt2$. We know 
that $\synO(\alpha_2) \subseteq \semO(\alpha_2)$ by induction, then $v \in \semO(\alpha_2)$.
By definition, we know that there is an \interpretation $\inst'$ and
$(\nu_1, \nu_2) \in \semm{\alpha_2}{\inst'}$ such that $\nu_1(v) \neq \nu_2(v)$. Take
$\inst''$ to be the interpretation where $\inst''(M) = \inst'(M)$ for any $M \neq M_1$
while $\inst''(M_1)$ is not empty. Clearly, $(\nu_1, \nu_2) \in \semm{\alpha}{\inst''}$, 
whence, $M_1 \neq M_2$, $(\nu_1, \nu_1) \in \semm{\alpha_1}{\inst''}$ by 
Lemma \ref{lem:null_id}, and $\semm{\alpha}{\inst''} = \semm{\alpha_2}{\inst'}$. It 
follows then that $v \in \semO(\alpha)$.

Similarly, we proceed to verify $\synI(\alpha) \subseteq \semI(\alpha)$. Let $v \in \inn2$.
We know that $\synI(\alpha_2) \subseteq \semI(\alpha_2)$ by induction, then 
$v \in \semI(\alpha_2)$. By definition, we know that there is an \interpretation
$\inst'$, $(\nu_1, \nu_2) \in \semm{\alpha_2}{\inst'}$, and $\nu_1'(v) \neq \nu_1(v)$
such that $(\nu_1', \nu_2') \not \in \semm{\alpha_2}{\inst'}$ for every valuation 
$\nu_2'$ that agrees with $\nu_2$ on $\semO(\alpha_2)$.

Take $\inst''$ to be the interpretation where $\inst''(M) = \inst'(M)$ for any 
$M \neq M_1$ while $\inst''(M_1)$ is not empty. Clearly, 
$\semm{\alpha}{\inst''} = \semm{\alpha_2}{\inst'}$, whence, $M_1 \neq M_2$. Therefore, 
$\semO(\alpha_2) \subseteq \semO(\alpha)$. Hence, $v \in \semI(\alpha)$. Indeed, 
$(\nu_1, \nu_2) \in \semm{\alpha}{\inst''}$ and for any valuation $\nu_2'$ if $\nu_2'$ 
agrees with $\nu_2$ on $\semO(\alpha)$, then $\nu_2'$ agrees with $\nu_2$ on 
$\semO(\alpha_2)$.

\subsubsection{$M_2$ is nullary, $M_1$ is not} Clearly, $\sem{\alpha} = \sem{\alpha_1}$ for 
any \interpretation $\inst$ where $\inst(M_2)$ is not empty. In this case, $\inn2$ and
$\outt2$ are both empty, then we observe that:
\begin{itemize}
	\item $\synO(\alpha) = \outt1$.
	\item $\synI(\alpha) = \inn1$.
\end{itemize}
First, we verify $\synO(\alpha) \subseteq \semO(\alpha)$. Let $v \in \outt1$. We know that
$\synO(\alpha_1) \subseteq \semO(\alpha_1)$ by induction, then $v \in \semO(\alpha_1)$. 
By definition, we know that there is an \interpretation $\inst'$ and 
$(\nu_1, \nu_2) \in \semm{\alpha_1}{\inst'}$ such that $\nu_1(v) \neq \nu_2(v)$. Take 
$\inst''$ to be the interpretation where $\inst''(M) = \inst'(M)$ for any $M \neq M_2$
while $\inst''(M_2)$ is not empty. Clearly, $(\nu_1, \nu_2) \in \semm{\alpha}{\inst''}$,
whence, $M_1 \neq M_2$, $(\nu_2, \nu_2) \in \semm{\alpha_2}{\inst''}$ by 
Lemma \ref{lem:null_id}, and $\semm{\alpha}{\inst''} = \semm{\alpha_1}{\inst'}$. It
follows then that $v \in \semO(\alpha)$.

Similarly, we proceed to verify $\synI(\alpha) \subseteq \semI(\alpha)$. Let $v \in \inn1$.
We know that $\synI(\alpha_1) \subseteq \semI(\alpha_1)$ by induction, then 
$v \in \semI(\alpha_1)$. By definition, we know that there is an \interpretation
$\inst'$, $(\nu_1, \nu_2) \in \semm{\alpha_1}{\inst'}$, and $\nu_1'(v) \neq \nu_1(v)$ 
such that $(\nu_1', \nu_2') \not \in \semm{\alpha_1}{\inst'}$ for every valuation
$\nu_2'$ that agrees with $\nu_2$ on $\semO(\alpha_1)$.

Take $\inst''$ to be the interpretation where $\inst''(M) = \inst'(M)$ for any 
$M \neq M_2$ while $\inst''(M_2)$ is not empty. Clearly,
$\semm{\alpha}{\inst''} = \semm{\alpha_1}{\inst'}$, whence, $M_1 \neq M_2$. Therefore,
$\semO(\alpha_1) \subseteq \semO(\alpha)$. Hence, $v \in \semI(\alpha)$. Indeed, 
$(\nu_1, \nu_2) \in \semm{\alpha}{\inst''}$ and for any valuation $\nu_2'$ if $\nu_2'$ 
agrees with $\nu_2$ on $\semO(\alpha)$, then $\nu_2'$ agrees with $\nu_2$ on 
$\semO(\alpha_1)$.

\subsubsection{Neither $M_1$ nor $M_2$ is nullary} Recall the definitions:
\begin{itemize}
	\item $\synO(\alpha) = \outt1 \cup \outt2$.
	\item $\synI(\alpha) = \inn1 \cup (\inn2 - \outt1)$.
\end{itemize}
We first proceed to verify $\synO(\alpha) \subseteq \semO(\alpha)$.  Let 
$v \in \outt1 \cup \outt2$. Consider an \interpretation $\inst$ such that 
\[
\inst(M_1) = \{(1, \ldots, 1; 2, \ldots, 2), (i_1, \ldots, i_m; 3, \ldots, 3)\}
\]
, where $i_1, \ldots, i_m$ are all the combinations of $\{1, 2\}$. 
Similarly, 
\[
\inst(M_2) = \{(1, \ldots, 1; 2, \ldots, 2), (i_1, \ldots, i_n; 3, \ldots, 3)\}
\]
, where $i_1, \ldots, i_n$ are all the combinations of $\{1, 2\}$.

Let $\nu_1$ be the valuation that is $1$ everywhere. Also, let $\nu$ be the valuation 
that is $2$ on $\outt1$ and $1$ elsewhere. Clearly, $(\nu_1, \nu) \in \sem{\alpha_1}$. 
Let $\nu_2$ be the valuation that is $3$ on $\outt2$, $2$ on $\outt1 - \outt2$, and 
$1$ elsewhere. Clearly, $(\nu_1, \nu_2) \in \sem{\alpha}$, whence 
$(\nu, \nu_2) \in \sem{\alpha_2}$. Hence, $v \in \semO(\alpha)$. Indeed, 
$\nu_2(v) \neq \nu_1(v)$ for $v \in \synO(\alpha)$.

Now we proceed to verify $\synI(\alpha) \subseteq \semI(\alpha)$. Let 
$v \in \inn1 \cup (\inn2 - \outt1)$. Consider an \interpretation $\inst$ 
where $\inst(M_1) = \{(1, \ldots, 1; 1, \ldots, 1)\}$ and similarly 
$\inst(M_2) = \{(1, \ldots, 1; 1, \ldots, 1)\}$. Let $\nu_1$ be the valuation that 
is $1$ everywhere. Clearly, $(\nu_1, \nu_1) \in \sem{\alpha}$, whence 
$(\nu_1, \nu_1) \in \sem{\alpha_1}$ and $(\nu_1, \nu_1) \in \sem{\alpha_2}$.

Take $\nu_1' := \nu_1[v:2]$. We establish that $v \in \semI(\alpha)$ by arguing that 
there is no valuation $\nu_2'$ for which $(\nu_1', \nu_2') \in \sem{\alpha}$.
In particular, when $v \in \inn1$. Clearly, there is no $\nu_2'$ such that 
$(\nu_1', \nu_2') \in \sem{\alpha_1}$. On the other hand, when $v \in \inn2 - \outt1$. 
Clearly, $(\nu_1', \nu) \in \sem{\alpha_1}$, whence $\nu = \nu_1'$. However, there is 
no $\nu_2'$ such that $(\nu_1', \nu_2') \in \sem{\alpha_2}$. Thus, there is no $\nu_2'$ 
such that $(\nu_1', \nu_2') \not \in \sem{\alpha}$ as desired. We conclude that 
$v \in \semI(\alpha)$.

\subsection{Converse}
Let $\alpha$ be of the form $\conv{\alpha_1}$, where $\alpha_1 := M(\bar{x}; \bar{y})$. 
Recall the definitions:
\begin{itemize}
	\item $\synO(\alpha) = \outt 1$.
	\item $\synI(\alpha) = \inn 1 \cup \outt 1$.
\end{itemize}
We first proceed to verify $\synO(\alpha) \subseteq \semO(\alpha)$.  Let $v \in \outt1$.
Consider an \interpretation $\inst$ where 
\[
\inst(M) = \{(1, \ldots, 1; 2, \ldots, 2)\}.
\] 
Let $\nu_1$ be the valuation that is $2$ on $\outt1$ and $1$ elsewhere.  Also 
let $\nu_2$ be the valuation that is $1$ everywhere. Clearly, 
$(\nu_1, \nu_2) \in \sem{\alpha}$ since $(\nu_2, \nu_1) \in \sem{\alpha_1}$. 
Therefore, $v \in \semO(\alpha)$ since $\nu_1(v) \neq \nu_2(v)$.

Now we proceed to verify $\synI(\alpha) \subseteq \semI(\alpha)$. Let 
$v \in \inn1 \cup \outt1$. Consider the same \interpretation $\inst$ and the same 
valuations $\nu_1$ and $\nu_2$. We established that $(\nu_1, \nu_2) \in \sem{\alpha}$.
Take $\nu_1'  := \nu_1[v:3]$. We establish that $v \in \semI(\alpha)$ by arguing that 
there is no $\nu_2'$ for which $(\nu_1', \nu_2') \in \sem{\alpha}$. Indeed, when 
$v \in \outt1$, then $\nu_1$ have to be $2$ on $v$. In the other case, when 
$v \in \inn1 - \outt1$, then $\nu_1$ have to be $1$ on $v$. Thus, there is no 
$\nu_2'$ for which $(\nu_1', \nu_2') \in \sem{\alpha}$ as desired. Consequently, 
$v \in \semI(\alpha)$.

\subsection{Left Cylindrification}
Let $\alpha$ be of the form $\cyll{x}({\alpha_1})$, where $\alpha_1 := M(\bar{x}; \bar{y})$. 
Recall the definitions:
\begin{itemize}
	\item $\synO(\alpha) = \outt 1 \cup \{x\}$.
	\item $\synI(\alpha) = \inn 1 - \{x\}$.
\end{itemize}
We first proceed to verify $\synO(\alpha) \subseteq \semO(\alpha)$. Let 
$v \in \outt1 \cup \{x\}$. Consider an \interpretation $\inst$ where 
$\inst(M) = \{(1, \ldots, 1; 2, \ldots, 2)\}$. Let $\nu_1$ be the valuation that is 
$3$ on $x$ and $1$ elsewhere. Also let $\nu_2$ be the valuation that is $2$ on $\outt1$
and $1$ everywhere else. Clearly, $(\nu_1, \nu_2) \in \sem{\alpha}$ since 
$(\nu_1[x:1], \nu_2) \in \sem{\alpha_1}$. Therefore, $v \in \semO(\alpha)$ since 
$\nu_1(v) \neq \nu_2(v)$.

Now we proceed to verify $\synI(\alpha) \subseteq \semI(\alpha)$. Let $v \in \inn1 - \{x\}$. 
Consider the same \interpretation $\inst$ and the same valuations $\nu_1$ and $\nu_2$. We 
established that $(\nu_1, \nu_2) \in \sem{\alpha}$.
Take $\nu_1'  := \nu_1[v:2]$. We establish that $v \in \semI(\alpha)$ by arguing that there is 
no $\nu_2'$ for which $(\nu_1', \nu_2') \in \sem{\alpha}$. Indeed, this is true since
$v \in \inn1 - \{x\}$. Consequently, $v \in \semI(\alpha)$.

\subsection{Right Cylindrification}
Let $\alpha$ be of the form $\cylr{x}({\alpha_1})$, where $\alpha_1 := M(\bar{x}; \bar{y})$. 
Recall the definitions:
\begin{itemize}
	\item $\synO(\alpha) = \outt 1 \cup \{x\}$.
	\item $\synI(\alpha) = \inn 1$.
\end{itemize}
We first proceed to verify $\synO(\alpha) \subseteq \semO(\alpha)$. Let 
$v \in \outt1 \cup \{x\}$. Consider an \interpretation $\inst$ where
$\inst(M) = \{(1, \ldots, 1; 2, \ldots, 2)\}$. Let $\nu_1$ be the valuation that 
is $1$ everywhere. Also let $\nu_2$ be the valuation that is $2$ on $\outt1$ and on $x$ 
and $1$ everywhere else. Clearly, $(\nu_1, \nu_2) \in \sem{\alpha}$ since either 
$(\nu_1, \nu_2[x:1]) \in \sem{\alpha_1}$ or $(\nu_1, \nu_2) \in \sem{\alpha_1}$. 
Therefore, $v \in \semO(\alpha)$ since $\nu_1(v) \neq \nu_2(v)$.

Now we proceed to verify $\synI(\alpha) \subseteq \semI(\alpha)$. Let $v \in \inn1$. 
Consider the same \interpretation $\inst$ and the same valuations $\nu_1$ and $\nu_2$.
We established that $(\nu_1, \nu_2) \in \sem{\alpha}$.
Take $\nu_1'  := \nu_1[v:2]$. We establish that $v \in \semI(\alpha)$ by arguing that 
there is no $\nu_2'$ for which $(\nu_1', \nu_2') \in \sem{\alpha}$. Indeed, this is true
since $v \in \inn1$. Consequently, $v \in \semI(\alpha)$.

\subsection{Left Selection}
Let $\alpha$ be of the form $\sell{x = y}({\alpha_1})$, where 
$\alpha_1 := M(\bar{u}; \bar{w})$, $\bar{u} = u_1, \ldots, u_n$, and 
$\bar{w} = w_1, \ldots, w_m$. We distinguish different cases based on whether
$x \eqv y$.

\paragraph{When $x$ and $y$ are the same variable ($x \eqv y$)} Recall the definitions in this
case:
\begin{itemize}
	\item $\synO(\alpha) = \outt1$.
	\item $\synI(\alpha) = \inn1$.
\end{itemize}
We proceed to verify that $\synO(\alpha) \subseteq \semO(\alpha)$ and 
$\synI(\alpha) \subseteq \semI(\alpha)$. Indeed, this is true since 
$\sem{\alpha} = \sem{\alpha_1}$ for any \interpretation $\inst$ because of 
$x \eqv y$.

\paragraph{When $x$ and $y$ are different variables ($x \neqv y$)} Recall the definitions 
in this case:
\begin{itemize}
	\item $\synO(\alpha) := \outt1$.
	\item $\synI(\alpha) := \inn1 \cup \{x,y\}$.
\end{itemize}
We first proceed to verify $\synO(\alpha) \subseteq \semO(\alpha)$. Let $v \in \outt1$.
Consider an \interpretation $\inst$ where
\[\inst(M) = \{(1, \ldots, 1; 2, \ldots, 2)\}.\] 
Let $\nu_1$ be the valuation that is $1$ everywhere. Also let $\nu_2$ be the valuation that 
is $2$ on $\outt1$ and $1$ everywhere else. Clearly, $(\nu_1, \nu_2) \in \sem{\alpha}$ 
since $(\nu_1, \nu_2) \in \sem{\alpha_1}$ and $\nu_1(x) = \nu_1(y)$. Therefore, 
$v \in \semO(\alpha)$ since $\nu_1(v) \neq \nu_2(v)$.

Now we proceed to verify $\synI(\alpha) \subseteq \semI(\alpha)$. Let 
$v \in \inn1 \cup \{x, y\}$. Consider an \interpretation $\inst$ where 
$\inst(M_1) = \{(1, \ldots, 1; 1, \ldots, 1)\}$. Let $\nu_1$ be the valuation that
is $1$ everywhere. Clearly, $(\nu_1, \nu_1) \in \sem{\alpha}$ since 
$(\nu_1, \nu_1) \in \sem{\alpha_1}$ and $\nu_1(x) = \nu_1(y)$. Take 
$\nu_1' := \nu_1[v:2]$. We establish that $v \in \semI(\alpha)$ by arguing that there 
is no $\nu_2'$ for which $(\nu_1', \nu_2') \in \sem{\alpha}$. In particular, when 
$v \in \inn1$, it is clear that there is no $\nu_2'$ such that
$(\nu_1', \nu_2') \in \sem{\alpha_1}$. In the other case, when $v$ is either $x$ or $y$,
there is no $\nu_2'$ such that $(\nu_1', \nu_2') \in \sem{\alpha}$. Indeed, this is true
since $x \neqv y$ and $\nu_1'(x) \neq \nu_1'(y)$. Consequently, $v \in \semI(\alpha)$.

\subsection{Right Selection}
Let $\alpha$ be of the form $\selr{x = y}({\alpha_1})$, where
$\alpha_1 := M(\bar{u}; \bar{w})$, $\bar{u} = u_1, \ldots, u_n$, and 
$\bar{w} = w_1, \ldots, w_m$. We distinguish different cases based on whether
$x \eqv y$.

\paragraph{When $x$ and $y$ are the same variable ($x \eqv y$)} Recall the definitions in 
this case:
\begin{itemize}
	\item $\synO(\alpha) = \outt1$.
	\item $\synI(\alpha) = \inn1$.
\end{itemize}
We proceed to verify that $\synO(\alpha) \subseteq \semO(\alpha)$ and 
$\synI(\alpha) \subseteq \semI(\alpha)$. Indeed, this is true since 
$\sem{\alpha} = \sem{\alpha_1}$ for any \interpretation $\inst$ because of
$x \eqv y$.

\paragraph{When $x$ and $y$ are different variables ($x \neqv y$)} Recall the definitions 
in this case:
\begin{itemize}
	\item $\synO(\alpha) := \outt1$.
	\item $\synI(\alpha) := \inn1 \cup (\{x,y\} - \outt1)$.
\end{itemize}
We first proceed to verify $\synO(\alpha) \subseteq \semO(\alpha)$. Let $v \in \outt1$.
Consider an \interpretation $\inst$ where 
\[\inst(M) = \{(i_1, \ldots, i_n; 2, \ldots, 2)\}\]
such that $i_j = 2$ if $u_j$ is either $x$ or $y$ and $u_j \not \in \outt1$, otherwise, 
$u_j = 1$. Let $\nu_1$ be the valuation that is $2$ on $x$ if $x \not \in \outt1$, $2$ on 
$y$ if $y \not \in \outt1$, and $1$ everywhere. Also let $\nu_2$ be the valuation that is 
$2$ on $\outt1$ and agrees with $\nu_1$ everywhere else. Clearly, 
$(\nu_1, \nu_2) \in \sem{\alpha}$ since $(\nu_1, \nu_2) \in \sem{\alpha_1}$ and
$\nu_2(x) = \nu_2(y)$. Therefore, $v \in \semO(\alpha)$ since $\nu_1(v) \neq \nu_2(v)$.

Now we proceed to verify $\synI(\alpha) \subseteq \semI(\alpha)$. Let
$v \in \inn1 \cup \{x, y\}$. Consider an \interpretation $\inst$ where 
$\inst(M_1) = \{(1, \ldots, 1; 1, \ldots, 1)\}$. Let $\nu_1$ be the valuation that
is $1$ everywhere. Clearly, $(\nu_1, \nu_1) \in \sem{\alpha}$ since 
$(\nu_1, \nu_1) \in \sem{\alpha_1}$ and $\nu_1(x) = \nu_1(y)$. Take $\nu_1'  := \nu_1[v:2]$.
We establish that $v \in \semI(\alpha)$ by arguing that there is no $\nu_2'$ for which
$(\nu_1', \nu_2') \in \sem{\alpha}$. In particular, when $v \in \inn1$, it is clear that
there is no $\nu_2'$ such that $(\nu_1', \nu_2') \in \sem{\alpha_1}$. Now we need to verify
the same when $v$ is $x$ or $y$ and $v \not \in \inn1$. Thereto, suppose 
$(\nu_1', \nu_2') \in \sem{\alpha}$. In the case of $v$ is $x$ and $x \not \in \inn1$, this
is only possible when $x \not \in \outt1$. Therefore, $\nu_2'(x) = \nu_1'(x) = 2$ but
$\nu_2'(y) = 1$ whether $y \in \outt1$ or not. Hence, 
$(\nu_1', \nu_2') \not \in \sem{\alpha}$ since $x \neqv y$ and $\nu_2'(x) \neq \nu_2'(y)$.
The case when $v$ is $y$ and $y \not \in \inn1$ is symmetric. Consequently,
$v \in \semI(\alpha)$.

\subsection{Left-to-Right Selection}
Let $\alpha$ be of the form $\sellr{x = y}({\alpha_1})$, where
$\alpha_1 := M(\bar{u}; \bar{w})$, $\bar{u} = u_1, \ldots, u_n$, and 
$\bar{w} = w_1, \ldots, w_m$. We distinguish different cases based on whether
$x \eqv y$ and $y \in \outt1$.

\paragraph{When $x \eqv y$ and $y \in \outt1$} Recall the definitions in this case:
\begin{itemize}
	\item $\synO(\alpha) = \outt1 - \{x\}$.
	\item $\synI(\alpha) = \inn1 \cup \{x\}$.
\end{itemize}
In what follows, since $x \eqv y$ we will refer to both of them with $x$. We first 
proceed to verify $\synO(\alpha) \subseteq \semO(\alpha)$. Let $v \in \outt1 - \{x\}$. 
Consider an \interpretation $\inst$ such that 
$\inst(M) = \{(1, \ldots, 1; o_1, \ldots, o_m)\}$ where $o_j = 1$ if $w_j = y$, otherwise
$o_j = 2$ . Let $\nu_1$ be the valuation that is $1$ everywhere. Also let $\nu_2$ be
the valuation that is $2$ on $\outt1 - \{x\}$ and $1$ everywhere else. Clearly, 
$(\nu_1, \nu_2) \in \sem{\alpha}$ since $(\nu_1, \nu_2) \in \sem{\alpha_1}$ and 
$\nu_1(x) = \nu_2(x)$. Therefore, $v \in \semO(\alpha)$ since $\nu_1(v) \neq \nu_2(v)$.

Now we proceed to verify $\synI(\alpha) \subseteq \semI(\alpha)$. Let 
$v \in \inn1 \cup \{x\}$. Consider an \interpretation $\inst$ where 
$\inst(M) = \{(1, \ldots, 1; 1, \ldots, 1)\}$. Let $\nu_1$ be the valuation that is
$1$ everywhere. Clearly, $(\nu_1, \nu_1) \in \sem{\alpha}$ since
$(\nu_1, \nu_1) \in \sem{\alpha_1}$. Take $\nu_1'  := \nu_1[v:2]$. We establish that 
$v \in \semI(\alpha)$ by arguing that there is no $\nu_2'$ for which 
$(\nu_1', \nu_2') \in \sem{\alpha}$. Thereto, suppose that 
$(\nu_1', \nu_2') \in \sem{\alpha}$. In particular, when $v \in \inn1$ it is clear 
that $(\nu_1', \nu_2') \not \in \sem{\alpha_1}$. On the other hand, when $v = x$ and
$x \in \outt1 - \inn1$, clearly $\nu_1'(x) = 2 \neq 1 = \nu_2'(x)$. Consequently,
$v \in \semI(\alpha)$.

\paragraph{When $x \eqv y$ and $y \not \in \outt1$} Recall the definitions in this case:
\begin{itemize}
	\item $\synO(\alpha) = \outt1$.
	\item $\synI(\alpha) = \inn1$.
\end{itemize}
In what follows, since $x \eqv y$ we will refer to both of them with $x$. We proceed to 
verify $\synO(\alpha) \subseteq \semO(\alpha)$ and $\synI(\alpha) \subseteq \semI(\alpha)$.
Indeed, this is true since $\sem{\alpha} = \sem{\alpha_1}$ for any \interpretation 
$\inst$ because of $x \eqv y$ and $x \not \in \outt1$.

\paragraph{When $x \neqv y$ and $y \in \outt1$} Recall the definitions in this case:
\begin{itemize}
	\item $\synO(\alpha) = \outt1$.
	\item $\synI(\alpha) = \inn1 \cup \{x\}$.
\end{itemize}
We first proceed to verify $\synO(\alpha) \subseteq \semO(\alpha)$. Let $v \in \outt1$.
Consider an \interpretation $\inst$ such that 
$\inst(M) = \{(i_1, \ldots, i_n; o_1, \ldots, o_m)\}$ where $i_j = 2$ if $u_j = x$, 
otherwise $i_j = 1$. Also, $o_j = 3$ if $w_j = x$, otherwise $o_j = 2$ . Let $\nu_1$ be 
the valuation that is $2$ on $x$ and $1$ everywhere else. Also let $\nu_2$ be the 
valuation that is $2$ on $\outt1 - \{x\}$, $3$ on $x$ if $x \in \outt1$ and agrees with 
$\nu_1$ everywhere else. Clearly, $(\nu_1, \nu_2) \in \sem{\alpha}$ since 
$(\nu_1, \nu_2) \in \sem{\alpha_1}$ and $\nu_1(x) = \nu_2(y)$. Therefore,
$v \in \semO(\alpha)$. Indeed, in both cases whether $x \in \outt1$ or not, 
$\nu_1(v) \neq \nu_2(v)$.

Now we proceed to verify $\synI(\alpha) \subseteq \semI(\alpha)$. Let 
$v \in \inn1 \cup \{x\}$. Consider an \interpretation $\inst$ where 
$\inst(M) = \{(1, \ldots, 1; 1, \ldots, 1)\}$. Let $\nu_1$ be the valuation that is 
$1$ everywhere. Clearly, $(\nu_1, \nu_1) \in \sem{\alpha}$ since 
$(\nu_1, \nu_1) \in \sem{\alpha_1}$. Take $\nu_1'  := \nu_1[v:2]$. We establish that
$v \in \semI(\alpha)$ by arguing that there is no $\nu_2'$ for which 
$(\nu_1', \nu_2') \in \sem{\alpha}$. Thereto, suppose that 
$(\nu_1', \nu_2') \in \sem{\alpha}$. In particular, when $v \in \inn1$ it is clear
that $(\nu_1', \nu_2') \not \in \sem{\alpha_1}$. On the other hand, when $v = x$ and
$y \in \outt1$, clearly $\nu_1'(x) = 2 \neq 1 = \nu_2'(y)$. Consequently, 
$v \in \semI(\alpha)$.

\paragraph{When $x \neqv y$ and $y \not \in \outt1$} Recall the definitions in this case:
\begin{itemize}
	\item $\synO(\alpha) = \outt1$.
	\item $\synI(\alpha) = \inn1 \cup \{x, y\}$.
\end{itemize}
We first proceed to verify $\synO(\alpha) \subseteq \semO(\alpha)$. Let $v \in \outt1$.
Consider an \interpretation $\inst$ such that $\inst(M) = \{(1, \ldots, 1; 2, \ldots, 2)\}$.
Let $\nu_1$ be the valuation that is $1$ everywhere. Also let $\nu_2$ be the valuation that 
is $2$ on $\outt1$ and $1$ everywhere else. Clearly, $(\nu_1, \nu_2) \in \sem{\alpha}$ since
$(\nu_1, \nu_2) \in \sem{\alpha_1}$ and $\nu_1(x) = \nu_2(y)$. Indeed, this is true since 
$y \not \in \outt1$, then $\nu_1(y) = \nu_2(y)$. Therefore, $v \in \semO(\alpha)$ since
$\nu_1(v) \neq \nu_2(v)$.

Now we proceed to verify $\synI(\alpha) \subseteq \semI(\alpha)$. Let
$v \in \inn1 \cup \{x, y\}$. Consider an \interpretation $\inst$ where
$\inst(M) = \{(1, \ldots, 1; 1, \ldots, 1)\}$. Let $\nu_1$ be the valuation that is
$1$ everywhere. Clearly, $(\nu_1, \nu_1) \in \sem{\alpha}$ since 
$(\nu_1, \nu_1) \in \sem{\alpha_1}$. Take $\nu_1'  := \nu_1[v:2]$. We establish that
$v \in \semI(\alpha)$ by arguing that there is no $\nu_2'$ for which 
$(\nu_1', \nu_2') \in \sem{\alpha}$. Thereto, suppose that 
$(\nu_1', \nu_2') \in \sem{\alpha}$. In particular, when $v \in \inn1$ it is clear that
$(\nu_1', \nu_2') \not \in \sem{\alpha_1}$. On the other hand, when $v = x$ or $v = y$,
clearly $\nu_1'(x) \neq (\nu_1'(y) = \nu_2'(y))$ since $y \not \in \outt1$ and 
$x \neqv y$. Consequently, $v \in \semI(\alpha)$.

\section{Optimality Theorem Proof}\label{app:optimality}
In this section, we prove Theorem~\ref{thm:optimality}.  Thus, we would like to show 
that 
\[ \synI(\alpha) \subseteq I(\alpha)\text{ and }\synO(\alpha) \subseteq O(\alpha).\]
for any LIF expression $\alpha$, assuming that $(I,O)$ is a sound and compositional 
input--output definition.  The proof is by induction on the structure of $\alpha$.

\paragraph{Atomic Modules}
For atomic module expressions $\alpha$, this follows directly from 
Theorem~\ref{thm:precision}. 

\paragraph{Identity}
For $\alpha=\id$, this is immediate since $\synI(\id)=\synO(\id)=\emptyset$.

\paragraph{Binary Operators}
For $\alpha = \alpha_1 \boxdot \alpha_2$, where $\boxdot$ is a binary operator,
we define two atomic module expressions $\alpha_1' = M_1(\bar x; \bar y)$ and 
$\alpha_2'=M_2(\bar u, \bar v)$ where $\bar x = I(\alpha_1)$, 
$\bar y = O(\alpha_1)$, $\bar u = I(\alpha_2)$, and $\bar v = O(\alpha_2)$ with 
$M_i$ distinct module names of the right arity.

Since $(I,O)$ is sound, we know that the following holds for $i \in \{1,2\}$:
\begin{equation}\label{eq:optimal:one} 
I(\alpha_i) = I(\alpha'_i) = \synI(\alpha'_i) \text{ and } 
O(\alpha_i) = O(\alpha'_i) = \synO(\alpha'_i). 
\end{equation}
Moreover by soundness and Proposition~\ref{prop:inputApp}, we know that:
\begin{equation}\label{eq:optimal:two}
\semI(\alpha'_1 \boxdot \alpha'_2) \subseteq I(\alpha'_1 \boxdot \alpha'_2) 
\text{ and }
\semO(\alpha'_1 \boxdot \alpha'_2) \subseteq O(\alpha'_1 \boxdot \alpha'_2). 
\end{equation}
From the Precision Theorem, we know that:
\begin{equation}\label{eq:optimal:three}
\synI(\alpha'_1 \boxdot \alpha'_2) = \semI(\alpha'_1 \boxdot \alpha'_2)
\text{ and }
\synO(\alpha'_1 \boxdot \alpha'_2) = \semO(\alpha'_1 \boxdot \alpha'_2). 
\end{equation}
From the compositionality of $(I,O)$, we know that
\begin{equation}\label{eq:optimal:four} 
I(\alpha_1' \boxdot \alpha_2') = I(\alpha_1 \boxdot \alpha_2) \text{ and }
O(\alpha_1' \boxdot \alpha_2') = O(\alpha_1 \boxdot \alpha_2).
\end{equation}
By combining Equations (\ref{eq:optimal:two}--\ref{eq:optimal:four}), we find that:
\begin{equation}\label{eq:optimal:five} 
\synI(\alpha'_1 \boxdot \alpha'_2) \subseteq I(\alpha_1 \boxdot \alpha_2) \text{ and }
\synO(\alpha'_1 \boxdot \alpha'_2) \subseteq O(\alpha_1 \boxdot \alpha_2).
\end{equation}
We now \textbf{claim} the following
\begin{equation} \label{eq:optimal:six} 
\synI(\alpha_1 \boxdot \alpha_2) \subseteq \synI(\alpha'_1 \boxdot \alpha'_2)
\text{ and }
\synO(\alpha_1 \boxdot \alpha_2) \subseteq \synO(\alpha'_1 \boxdot \alpha'_2). 
\end{equation}
If we prove our \textbf{claim}, then combining Equations 
(\ref{eq:optimal:five}--\ref{eq:optimal:six}) establishes our theorem
for binary operators.

First, we prove our claim for the inductive cases for outputs of the 
different binary operators.
From the inductive hypothesis and Equation~(\ref{eq:optimal:one}), we know that 
for $i\in\{1,2\}$:
\[\synO(\alpha_i) \subseteq O(\alpha_i) = \synO(\alpha_i')\]
Hence, it is clear that:
\begin{itemize}
    \item $\synO(\alpha_1) \cup \synO(\alpha_2)$ is a subset of
    $\synO(\alpha'_1) \cup \synO(\alpha'_2)$, which settles the cases 
    when $\boxdot \in \{\cup, \comp\}$ since 
    $\synO(\beta \boxdot \gamma) = \synO(\beta) \cup \synO(\gamma)$ for any
    LIF expressions $\beta$ and $\gamma$.
    \item $\synO(\alpha_1) \cap \synO(\alpha_2)$ is a subset of 
    $\synO(\alpha'_1) \cap \synO(\alpha'_2)$, which settles the case 
    when $\boxdot$ is $\cap$ since 
    $\synO(\beta \boxdot \gamma) = \synO(\beta) \cap \synO(\gamma)$ for any
    LIF expressions $\beta$ and $\gamma$.
    \item $\synO(\alpha_1)$ is a subset of $\synO(\alpha'_1)$, which
    settles the case when $\boxdot$ is $\setminus$ since 
    $\synO(\beta \boxdot \gamma) = \synO(\beta)$ for any
    LIF expressions $\beta$ and $\gamma$.
\end{itemize}

Now, we consider the inductive cases for the inputs of the different binary
operators.  Similar to the outputs, we know that for $i\in\{1,2\}$:
\[\synI(\alpha_i) \subseteq I(\alpha_i) = \synI(\alpha_i')\] 
Consequently,
\begin{itemize}
    \item when $x \in \synI(\alpha_1) \cup (\synI(\alpha_2) - \synO(\alpha_1))$,
    we consider the following cases:
    \begin{itemize}
        \item if $x \in \synI(\alpha_1)$, then it is clear that 
        $x \in \synI(\alpha'_1)$;
        \item if $x \in (\synI(\alpha_2) - \synO(\alpha_1))$, then we know that
        $x \in \synI(\alpha'_2)$.  Moreover, since $x \not \in \synO(\alpha_1)$, 
        we know by soundness of $(\synI, \synO)$ that $x \not \in \semO(\alpha_1)$.  
        Now, we differentiate between two cases
        \begin{itemize}
            \item when $x \not \in \synO(\alpha'_1)$, it is clear that
            $x \in (\synI(\alpha'_2) - \synO(\alpha'_1))$.
            \item when $x \in \synO(\alpha'_1)$, we know from 
            Equation~(\ref{eq:optimal:one}) that $x \in O(\alpha_1)$.
            From Lemma~\ref{lem:sound-input} and Equation~(\ref{eq:optimal:one}), 
            it follows that $x \in I(\alpha_1)$ and $x \in \synI(\alpha'_1)$.
        \end{itemize}
    \end{itemize}
    In all cases, we verify that 
    $x \in \synI(\alpha'_1) \cup (\synI(\alpha'_2) - \synO(\alpha'_1))$.
    This settles the case when $\boxdot$ is $\comp$ since 
    $\synI(\beta \boxdot \gamma) = \synI(\beta) \cup (\synI(\gamma) - \synO(\beta))$ 
    for any LIF expressions $\beta$ and $\gamma$.
    
    \item when $x \in \synI(\alpha_1) \cup \synI(\alpha_2) \cup 
    (\synO(\alpha_1) \symdif \synO(\alpha_2))$, we consider the following cases:
    \begin{itemize}
        \item if $x \in \synI(\alpha_i)$ for some $i$, then it is clear that
        $x \in \synI(\alpha'_i)$.
        \item if $x \in \synO(\alpha_i) - \synO(\alpha_j)$ for $i \neq j$, we
        know that $x \in \synO(\alpha'_i)$.  Since $x \not \in \synO(\alpha_j)$, 
        we know by soundness that $x \not \in \semO(\alpha_j)$.
        Now, we differentiate between two cases
        \begin{itemize}
            \item when $x \not \in \synO(\alpha'_j)$, it is clear that
            $x \in (\synO(\alpha'_i) \symdif \synO(\alpha'_j))$.
            \item when $x \in \synO(\alpha'_j)$, we know from 
            Equation~(\ref{eq:optimal:one}) that $x \in O(\alpha_j)$.
            From Lemma~\ref{lem:sound-input} and Equation~(\ref{eq:optimal:one}), 
            it follows that $x \in I(\alpha_j)$ and $x \in \synI(\alpha'_j)$.
        \end{itemize}
    \end{itemize}
    In all cases, we verify that $x \in \synI(\alpha'_1) \cup \synI(\alpha'_2) \cup 
    (\synO(\alpha'_1) \symdif \synO(\alpha'_2))$.  This settles the cases when 
    $\boxdot \in \{\cup, \cap, \setminus\}$ since $\synI(\beta \boxdot \gamma) = 
    \synI(\beta) \cup \synI(\gamma) \cup (\synO(\beta) \symdif \synO(\gamma))$ 
    for any LIF expressions $\beta$ and $\gamma$.
\end{itemize}

\paragraph{Unary Operators} We follow a similar approach for unary operators.
For $\alpha = \square \alpha_1$, where $\square$ is a unary operator,
we define one atomic module expression $\alpha_1' = M_1(\bar x; \bar y)$ where 
$\bar x = I(\alpha_1)$, and $\bar y = O(\alpha_1)$.

Since $(I,O)$ is sound, we know that the following holds:
\begin{equation}\label{eq:optimal:one1} 
I(\alpha_1) = I(\alpha'_1) = \synI(\alpha'_1) \text{ and } 
O(\alpha_1) = O(\alpha'_1) = \synO(\alpha'_1). 
\end{equation}
Moreover, we know that:
\begin{equation}\label{eq:optimal:two1}
\semI(\square \alpha'_1) \subseteq I(\square \alpha'_1) \text{ and }
\semO(\square \alpha'_1) \subseteq O(\square \alpha'_1). 
\end{equation}
From the precision theorem, we know that:
\begin{equation}\label{eq:optimal:three1}
\synI(\square \alpha'_1) = \semI(\square \alpha'_1) \text{ and }
\synO(\square \alpha'_1) = \semO(\square \alpha'_1). 
\end{equation}
From the compositionality of $(I,O)$, we know that
\begin{equation}\label{eq:optimal:four1} 
I(\square \alpha_1') = I(\square \alpha_1) \text{ and }
O(\square \alpha_1') = O(\square \alpha_1).
\end{equation}
By combining Equations (\ref{eq:optimal:two1}--\ref{eq:optimal:four1}), 
we find that:
\begin{equation}\label{eq:optimal:five1} 
\synI(\square \alpha'_1) \subseteq I(\square \alpha_1) \text{ and }
\synO(\square \alpha'_1) \subseteq O(\square \alpha_1).
\end{equation}
We now \textbf{claim} the following
\begin{equation} \label{eq:optimal:six1} 
\synI(\square \alpha_1) \subseteq \synI(\square \alpha'_1) \text{ and }
\synO(\square \alpha_1) \subseteq \synO(\square \alpha'_1). 
\end{equation}
If we prove our \textbf{claim}, then combining Equations 
(\ref{eq:optimal:five1}--\ref{eq:optimal:six1}) establishes our theorem
for unary operators.

Proving our claim for the inductive cases for outputs of the different 
unary operators follows directly from the inductive hypothesis and 
Equation~(\ref{eq:optimal:one1}), which states that:
\[\synO(\alpha_1) \subseteq O(\alpha_1) = \synO(\alpha_1').\]
Indeed, $\synO(\square \alpha_1)$ and $\synO(\square \alpha_1')$, respectively,
simply equal $\synO(\alpha_1)$ and $\synO(\alpha_1')$, except for the possible
addition or removal of some fixed variable that depends only on $\square$.

Now, we consider the inductive cases for inputs.  Similar to the outputs, 
we know that: 
\[\synI(\alpha_1) \subseteq I(\alpha_1) = \synI(\alpha_1').\]
Here, we only discuss the cases for $\sellr{x=y}$ and $\selr{x=y}$
as all the other cases again follow directly from the above inclusion
and the definition of $\synI$.

We begin by the cases for $\sellr{x=y}$. The cases are:
\begin{itemize}
    \item when $y \in \synO(\alpha_1)$, we have
    \[\synI(\sellr{x=y}(\alpha_1)) = \synI(\alpha_1) \cup \{x\} \subseteq 
    \synI(\alpha'_1) \cup \{x\} = \synI(\sellr{x=y}(\alpha'_1)).\]
    
    \item when $y \not \in \synO(\alpha_1)$ and $x \eqv y$, we have
    \[\synI(\sellr{x=y}(\alpha_1)) = \synI(\alpha_1) \subseteq \synI(\alpha'_1) 
    \subseteq \synI(\sellr{x=y}(\alpha'_1)).\]
    
    \item when $y \not \in \synO(\alpha_1)$ and $x \neqv y$, by definition
    \[\synI(\sellr{x=y}(\alpha_1)) = \synI(\alpha_1) \cup \{x, y\}.\]
    In case $y \not \in \synO(\alpha'_1)$, we are done since 
    $\synI(\alpha_1) \cup \{x, y\} \subseteq 
    \synI(\alpha_1') \cup \{x, y\} = \synI(\sellr{x=y}(\alpha_1'))$.  
    Otherwise, $y \in \synO(\alpha'_1)$ in which case 
    $\synI(\sellr{x=y}(\alpha_1')) = \synI(\alpha_1') \cup \{x\}$.  
    What remains to show is that $y \in \synI(\alpha'_1)$.  
    By Equation~\ref{eq:optimal:one1}, we have $y \in O(\alpha_1)$.
    Moreover, $y \not \in \semO(\alpha_1)$ since $y \not \in \synO(\alpha_1)$.
    By Lemma~\ref{lem:sound-input} and Equation~\ref{eq:optimal:one1}, we 
    obtain $y \in I(\alpha_1) = \synI(\alpha'_1)$ as desired.  
\end{itemize}
Finally, we consider the case for $\selr{x=y}$ when $x \neqv y$. The case 
when $x \eqv y$ follows directly.  By definition,
\[\synI(\selr{x=y}(\alpha_1)) = \synI(\alpha_1) \cup 
(\{x, y\} - \synO(\alpha_1)).\]
We can focus on $z \in \{x, y\}$.  If $z \in \synO(\alpha_1)$ or 
$z \not \in \synO(\alpha'_1)$, we are done.  Now, consider the case when 
$z \not \in \synO(\alpha_1)$, but $z \in \synO(\alpha'_1)$.  Similar to our 
reasoning for the last case in $\sellr{x=y}$, we can show that 
$z \in \synI(\alpha'_1)$, whence, $z \in \synI(\sellr{x=y}(\alpha'_1))$
by definition.

\section{Primitivity of Composition}\label{sec:comp}
We now turn our attention to the study of composition in LIF\@.  Indeed, LIF has a rich
set of logical operators already, plus an explicit operator ($\comp$) for sequential
composition.  This begs the question whether $\comp$ is not already definable in terms
of the other operators.

We begin by showing that for ``well-behaved'' expressions (all subexpressions have 
disjoint inputs and outputs) composition is redundant in LIF: every well-behaved LIF 
expression is equivalent to a LIF expression that does not use composition.  As a 
corollary, we will obtain that composition is generally redundant if there is an 
infinite supply of variables.  In contrast, in the bounded variable case, we will 
show that composition is primitive in LIF\@.  Here, we use LIFnc to denote the fragment 
of LIF without composition.

\subsection{When Inputs and Outputs are Disjoint, Composition is Non-Primitive}
Our first non-primitivity result is based on inputs and outputs.
We say that a LIF expression $\beta$ is \emph{io-disjoint} if
$\semO(\beta)\cap\semI(\beta)=\emptyset$.
The following theorem implies that if
$\alpha$, $\beta$, and all their subexpressions are io-disjoint, we can rewrite 
$\alpha\comp\beta$ into a LIFnc expression.
Of course, this property also holds in case $\synO(\beta)\cap\synI(\beta)=\emptyset$ 
since the syntactic inputs and outputs overapproximate the semantic ones.
\begin{theorem}\label{theo:IOdisjointcompnotprim}
Let $\alpha$ and $\beta$ be LIF expressions such that $\beta$ is io-disjoint.
Then, $\alpha \comp{} \beta$ is equivalent to
\[
\gamma := \compEqSem{\alpha}{\beta}.
\]
\end{theorem}
Intuitively, the reason why this expression works is as follows: we cylindrify $\alpha$
on the right. In general, this might result in a loss of information, but since we are 
only cylindrifying outputs of $\beta$, this means we only forget the information that
would be overwritten by $\beta$ anyway.
Since the inputs and outputs of $\beta$ are disjoint, $\beta$ does not need to know what 
$\alpha$ did to those variables in order to determine its own outputs.
We also cylindrify $\beta$ on the left on the outputs of $\alpha$, since these values will
be set by $\alpha$.  One then still needs to be careful in showing that the intersection 
indeed removes all artificial pairs, by exploiting the fact that expressions are 
inertial outside their output.
\begin{proof}
Let $\inst{}$ be an \interpretation.
First, we show that $\sem{\alpha \comp{} \beta}\subseteq \sem{\gamma}$.
If $(\nu_1,\nu_2) \in \sem{\alpha \comp{} \beta}$, then there is a $\nu_3$ such 
that $(\nu_1,\nu_3)\in \sem\alpha$ and $(\nu_3,\nu_2)\in \sem \beta$.
By definition of the outputs of $\beta$, $\nu_3$ and $\nu_2$ agree outside 
$\semO(\beta)$. Hence, $(\nu_1,\nu_2) \in \sem{\cylr{\semO(\beta)}(\alpha)}$. 
Similarly, we can show that $(\nu_1,\nu_2)\in \sem{\cyll{\semO(\alpha)}(\beta)}$.

For the other inclusion, assume that $(\nu_1,\nu_2)\in \sem{\gamma}$.  Using the 
definition of the semantics of cylindrification, we find $\nu_2'$ such that 
$(\nu_1,\nu_2')\in \sem{\alpha}$ and $\nu_2$ agrees with $\nu_2'$ outside 
$\semO{(\beta)}$ and we find a $\nu_1'$ such that  $\nu_1'$ agrees with 
$\nu_1$ outside $\semO(\alpha)$ and $(\nu_1',\nu_2)\in\sem{\beta}$.
Using the definition of output of $\beta$, we know that also $\nu_1'$ agrees with 
$\nu_2$ outside the outputs of $\beta$, thus $\nu_1'$ and $\nu_2'$ agree outside 
the outputs of $\beta$, and hence definitely on the inputs of $\beta$.
We can apply Proposition~\ref{prop:determined2} thanks to the $(\synI, \synO)$
soundness, $\synI(\alpha)$ is finite and determines $\synO(\alpha)$, which
contains $\semO(\alpha)$.  So we guarantee that $\beta$ is determined by its 
inputs, whence, there exists a $\nu_2''$ such that 
$(\nu_2',\nu_2'')\in\sem{\beta}$ where $\nu_2''=\nu_2$ on the outputs of 
$\beta$ and, since $\beta$ is inertial outside its outputs, $\nu_2''=\nu_2'$ 
outside the outputs of $\beta$. But we previously established that $\nu_2'$ 
agrees with $\nu_2$ outside the outputs of $\beta$, whence $\nu_2''=\nu_2$.  
Summarized we now found that $(\nu_1,\nu_2')\in \sem{\alpha}$ and 
$(\nu_2',\nu_2)\in\sem\beta$, whence, $(\nu_1,\nu_2)\in \sem{\alpha\comp\beta}$ 
as desired.
\end{proof}
\ignore{
The intuition underlying the expression Theorem~\ref{theo:IOdisjointcompnotprim} 
is that for each interpretation \inst, $\sem{\alpha\comp\beta}$ is
\begin{proof}
  Let $\inst{}$ be an \interpretation. First, we show that 
  $\sem{\alpha \comp{} \beta}\subseteq \sem{\gamma}$. To this end, assume that 
  $(\nu_1,\nu_2) \in \sem{\alpha\comp{} \beta}$, whence there exists $\nu_3$ such 
  that $(\nu_1,\nu_3)\in \sem{\alpha}$ and $(\nu_3,\nu_2)\in \sem{\beta}$. By the 
  Inertia Property, $\nu_2$ and $\nu_3$ agree outside $\synO(\beta)$. Therefore, 
  $(\nu_1,\nu_2)\in \sem{\cylr{\synO(\beta)}(\alpha)}$. Analogously, we can establish 
  that $(\nu_1,\nu_2)\in \sem{\cyll{\synO(\alpha)}(\beta)}$, whence 
  $(\nu_1,\nu_2)\in \sem{\gamma}$.

For the other inclusion, suppose that $(\nu_1,\nu_2) \in \sem{\gamma}$. We thus have
\begin{enumerate}
  \item Since $(\nu_1,\nu_2) \in \sem{\cylr{\synO(\beta)}(\alpha)}$ there exists $\nu_2'$ 
  such that
  \begin{enumerate}
    \item $(\nu_1,\nu_2') \in \sem{\alpha}$;
    \item $\nu_2$ and $\nu_2'$ agree everywhere outside $\synO(\beta)$.
    % \item $\nu_1$ and $\nu_2'$ agree everywhere outside $\synO(\alpha)$ by Lemma~\ref{lem:inertia}.
  \end{enumerate}
  \item Since $(\nu_1,\nu_2) \in \sem{\cyll{\synO(\alpha)}(\beta)}$ there exists $\nu_1'$ 
  such that:
    \begin{enumerate}
      \item $(\nu_1',\nu_2)\in \sem{\beta}$;
      \item $\nu_1'$ and $\nu_1$ agree everywhere outside $\synO(\alpha)$.
      \item By applying the Inertia Property to (2a) we also know that $\nu_1'$ and $\nu_2$ 
      agree outside $\synO(\beta)$.
    \end{enumerate}
\end{enumerate}
We now show that $(\nu_2',\nu_2) \in \sem{\beta}$. By (1b) and (2c) we know that 
$\nu_1'=\nu_2'$ outside $\synO(\beta)$, since $\synI(\beta)\cap \synO(\beta)=\emptyset$. 
Applying the Input-Output Determinacy property to $(\nu_1',\nu_2)\in \sem{\beta}$ we get 
that $(\nu_2',\nu_2)\in \sem{\beta}$. Combining this with (1a) we get that 
$(\nu_1,\nu_2)\in\sem{\alpha\comp{} \beta}$.
\end{proof}
}

Given the undecidability results of Section~\ref{sec:inout}, 
Theorem~\ref{theo:IOdisjointcompnotprim} is not effective.  We can however give the 
following syntactic variant.
\begin{theorem}\label{theo:IOdisjointcompnotprimSYN}
Let $\alpha$ and $\beta$ be LIF expressions such that 
$\synO(\beta)\cap\synI(\beta)=\emptyset$.
Then, $\alpha \comp{} \beta$ is equivalent to
\[
\compEqSyn{\alpha}{\beta}.
\]
\end{theorem}
\begin{proof}
Since $\disSyn \beta$, we obtain by Lemma~\ref{lem:sound-input} that 
$\semO(\beta) = \synO(\beta)$.  Thus, we alternatively 
show that $\alpha \comp \beta$ is equivalent to the expression
\[\compEqGen{\alpha}{\beta}{\semO(\beta)}{\synO(\alpha)}.\]

We can also see that $\beta$ is io-disjoint, since $\disSyn \beta$ 
and $(\synI, \synO)$ is sound.  
Thus, if we show that 
\[
\sem{\compEqGen{\alpha}{\beta}{\semO(\beta)}{\synO(\alpha)}} = 
\sem{\compEqSem{\alpha}{\beta}}
\] 
for any \interpretation $\inst$, we can apply 
Theorem~\ref{theo:IOdisjointcompnotprim} and we are done.

Thereto, let $\inst$ be an \interpretation.  By soundness, it is clear that 
\[\sem{\cyll{\semO(\alpha)}(\beta)} \subseteq \sem{\cyll{\synO(\alpha)}(\beta)}
\text{, so } \sem{\compEqSem{\alpha}{\beta}} \subseteq 
\sem{\compEqGen{\alpha}{\beta}{\semO(\beta)}{\synO(\alpha)}}.\] 

What remains to show is that the other inclusion also holds.  Thereto, let 
$(\nu_1, \nu_2) \in \sem{\compEqGen{\alpha}{\beta}{\semO(\beta)}{\synO(\alpha)}}$.
Clearly, $(\nu_1, \nu_2) \in \sem{\cylr{\semO(\beta)}(\alpha)}$ and 
$(\nu_1, \nu_2) \in \sem{\cyll{\synO(\alpha)}(\beta)}$.
From $(\nu_1, \nu_2) \in \sem{\cylr{\semO(\beta)}(\alpha)}$, we can see that
$\nu_1 = \nu_2$ outside $\semO(\alpha) \cup \semO(\beta)$.
From $(\nu_1, \nu_2) \in \sem{\cyll{\synO(\alpha)}(\beta)}$, we can see that there is 
a valuation $\nu_1'$ such that $(\nu_1', \nu_2) \in \sem{\beta}$ and $\nu_1' = \nu_1$ 
outside $\synO(\alpha)$.  
Define $\nu_1''$ to be the valuation $\nu_1'[\restr{\nu_1}{\semO(\beta)}]$.
By construction and io-disjointness of $\beta$, we see that $\nu_1'' = \nu_1'$
on $\semI(\beta)$ and outside $\semO(\beta)$.  By Proposition~\ref{prop:semDetToAD},
we obtain that $(\nu_1'', \nu_2) \in \sem{\beta}$.
Define $\nu$ to be the valuation $\nu_1''[\restr{\nu_1}{\semO(\alpha)}]$.
By the semantics of cylindrification, we see that 
$(\nu, \nu_2) \in \sem{\cyll{\semO(\alpha)}(\beta)}$.  
Consequently, $\nu = \nu_2$ outside $\semO(\alpha) \cup \semO(\beta)$.
Before, we established that $\nu_1$ and $\nu_2$ agree outside the same
set of variables.  So we obtain that $\nu = \nu_2 = \nu_1$ outside 
$\semO(\alpha) \cup \semO(\beta)$.  Moreover, we know by construction that 
$\nu = \nu_1'' = \nu_1$ on $\semO(\beta) \cup \semO(\alpha)$.  Then, $\nu$ is 
the same valuation as $\nu_1$.  So we obtain that 
$(\nu_1, \nu_2) \in \sem{\compEqSem{\alpha}{\beta}}$ as desired. 
\end{proof}
\begin{example}[Example~\ref{ex:inc:basic} continued]
Consider the expression
\[
\alpha = P_1(x;x) \comp P_1(x;y).
\]
with the \interpretation $\inst$ in Example~\ref{ex:inc:basic}.  In that case, 
$\alpha $ first increments $x$ by one and subsequently sets the value of $y$ to 
one higher than $x$.  Stated differently,
\[
\sem{\alpha} = \bigl\{ (\nu_1,\nu_2) \mid \nu_2(x) = \nu_1(x)+1 \wedge \nu_2(y)=\nu_2(x)+1 
\text{ and } \nu_1(z)=\nu_2(z) \text{ for }z\not\in\{x,y\} \bigr\}
\]
Theorem~\ref{theo:IOdisjointcompnotprim} tells us that $\alpha$ is equivalent to
\[ 
\cylr{y}(P_1(x;x))\cap \cyll{x}(P_1(x;y)).
\]
We see that
\begin{align*}
\sem{\cylr{y}(P_1(x;x))} &= \bigl\{(\nu_1,\nu_2)\mid 
\nu_2(x) = \nu_1(x)+1 \text{and } \nu_1(z)=\nu_2(z) \text{ for } z \not\in\{x,y\}\bigl\}, \\
\sem{\cyll{x}(P_1(x;y))} &= \bigl\{(\nu_1,\nu_2)\mid  
\nu_2(y)=\nu_2(x)+1 \text{and } \nu_1(z)=\nu_2(z) \text{ for }z\not\in\{x,y\}\bigl\}.
\end{align*}
The intersection of these indeed equals $\sem{\alpha}$.
\end{example}
Theorem~\ref{theo:IOdisjointcompnotprim} no longer holds in general if 
$\beta$ can have overlapping inputs and outputs, as the following example 
illustrates.
\begin{example}
Consider the expression \[\alpha := P_1(x;x) \comp P_1(x;x).\]
with the \interpretation $\inst$ as in the example above.
In this case, $\alpha$ increments the value of $x$ by two.
However, $\sem{\cylr{x}(P_1(x;x))}$ and $\sem{\cyll{x}(P_1(x;x))}$ are both equal to
\begin{align*}
 \{(\nu_1,\nu_2)\mid \nu_1(z) = \nu_2(z)\text{ for all } z\neq x\}.
\end{align*}
Hence, indeed, in this case $\alpha$ is not equivalent to
\[ 
\cylr{x}(P_1(x;x))\cap \cyll{x}(P_1(x;x)).
\]
\end{example}

\subsection{If \texorpdfstring{$\V$}{V} is Infinite, Composition is Non-Primitive}
We know from Theorem~\ref{theo:IOdisjointcompnotprim} that if $\beta$ is io-disjoint, 
$\alpha$ and $\beta$ can be composed without using the composition operator.
If $\V$ is sufficiently large, we can force any expression $\beta$ to be io-disjoint by 
having $\beta$ write its outputs onto unused variables instead of its actual outputs. 
The composition can then be eliminated following Theorem~\ref{theo:IOdisjointcompnotprim}, 
after which we move the variables back so that the ``correct'' outputs are used.
What we need to show is that ``moving the variables around'', as described above, 
is expressible without composition.
As before, we define the operators on $\bbar$s but their definition is  lifted to LIF 
expressions in a straightforward way.
\begin{definition}\label{def:mvr}
Let $B$ be a $\bbar$ and let $\bar{x}$ and $\bar{y}$ be disjoint tuples of distinct 
variables of the same length. The \emph{right move} is defined as follows:
\[
\mvr{\bar{x}\to\bar{y}}(B): = 
\{(\nu_1,\nu_2')\mid \nu_2'(\bar{x}) = \nu_1(\bar{x})\text{ and }
\exists \nu_2: (\nu_1,\nu_2)\in B \text{ and } \nu_2'(\bar{y}) = \nu_2(\bar{x})
\text{ and }\nu_2 =\nu_2'\text{ outside }\bar{x}\cup\bar{y}\}.
\]
\end{definition}
This operation can be expressed without composition, which we show in the 
following lemma:
\ignore{
\begin{proposition}\label{prop:moveEq}
Let $B$ be a $\bbar$ and $x$ and $y$ be distinct variables and let 
$(x)\cdot\bar{x}$ and $(y)\cdot\bar{y}$ be disjoint tuples of distinct 
variables of the same length.  Then, the following two expressions are 
equivalent:
\begin{itemize}
    \item $\mvr{(x)\cdot\bar{x}\to(y)\cdot\bar{y}}(B)$; and
    \item $\mvr{x\to y}(\mvr{\bar{x}\to\bar{y}}(B))$.
\end{itemize}
\end{proposition}
\begin{proof}
First we show that $\mvr{x\to y}(\mvr{\bar{x}\to\bar{y}}(B)) \subseteq 
\mvr{(x)\cdot\bar{x}\to(y)\cdot\bar{y}}(B)$. Let $(\nu_1, \nu_2) \in 
\mvr{x\to y}(\mvr{\bar{x}\to\bar{y}}(B))$.  We know that:
\begin{enumerate}
    \item $\nu_2(x)=\nu_1(x)$;
    \item $(\nu_1,\nu_2')\in \mvr{\bar{x}\to\bar{y}}(B)$;
    \item $\nu_2'(x) = \nu_2(y)$; and
    \item $\nu_2 =\nu_2'$ outside $\{x,y\}$.
\end{enumerate}
for some valuation $\nu_2'$.  From (2), we know that:
\begin{enumerate} \setcounter{enumi}{4}
    \item $\nu_2'(\bar{x}) = \nu_1(\bar{x})$;
    \item $(\nu_1,\nu_2'')\in B$;
    \item $\nu_2''(\bar x) = \nu_2'(\bar y)$; and
    \item $\nu_2' =\nu_2'' \text{ outside } \bar x \cup \bar y$.
\end{enumerate}
for some valuation $\nu_2''$.  From (4) and (5), we can see that
$\nu_2(\bar x) = \nu_1(\bar x)$.  Along with (1), we get that
\[\nu_2((x)\cdot\bar x) = \nu_1((x)\cdot\bar x) \tag{9}\]
From (7) and (4), we can see that $\nu_2''(\bar x) = \nu_2(\bar y)$.
Along with (3) and (8), we get that
\[\nu_2''((x)\cdot\bar{x}) = \nu_2((y)\cdot\bar{y}) \tag{10}\]
From (4) and (8), we can also see that
\[\nu_2 = \nu_2'' \text{ outside } \{x,y\}\cup\bar{x}\cup\bar{y} \tag{11}\]
Hence, we can see that 
$(\nu_1, \nu_2) \in \mvr{(x)\cdot\bar{x}\to(y)\cdot\bar{y}}(B)$
from (6) and (9-11).  This establishes that 
$\mvr{x\to y}(\mvr{\bar{x}\to\bar{y}}(B)) \subseteq 
\mvr{(x)\cdot\bar{x}\to(y)\cdot\bar{y}}(B)$.

Now, we prove the other inclusion.  Let $(\nu_1, \nu_2) \in 
\mvr{(x)\cdot\bar{x}\to(y)\cdot\bar{y}}(B)$.  Then,
\begin{enumerate}
    \item $\nu_2((x)\cdot\bar{x})=\nu_1((x)\cdot\bar{x})$;
    \item $(\nu_1,\nu_2')\in B$;
    \item $\nu_2'((x)\cdot\bar{x}) = \nu_2((y)\cdot\bar{y})$; and
    \item $\nu_2 =\nu_2'$ outside $\{x,y\} \cup \bar{x} \cup \bar{y}$.
\end{enumerate}
for some valuation $\nu_2'$.  Take $\nu_2''$ to be the valuation constructed 
as follows:
\begin{enumerate} \setcounter{enumi}{4}
    \item $\nu_2'' = \nu_2'$ outside $\bar{x} \cup \bar{y}$;
    \item $\nu_2''(\bar{x}) = \nu_1(\bar{x})$; and
    \item $\nu_2''(\bar{y}) = \nu_2'(\bar{x})$.
\end{enumerate}
From (2) and (5-7), we get that
\[(\nu_1, \nu_2'') \in \mvr{\bar{x}\to\bar{y}}(B) \tag{8}\]
Now, take $\nu_2''$ to be the valuation constructed as follows:
\begin{enumerate} \setcounter{enumi}{8}
    \item $\nu_2''' = \nu_2''$ outside $\{x,y\}$;
    \item $\nu_2'''(x) = \nu_1(x)$; and
    \item $\nu_2'''(y) = \nu_2''(x)$.
\end{enumerate}
From (8) and (9-11), we get that $(\nu_1,\nu_2''') \in 
\mvr{x\to y}(\mvr{\bar{x}\to\bar{y}}(B))$.  What remains to show is that
$\nu_2''' = \nu_2$.  Indeed,
\begin{itemize}
    \item $\nu_2''' = \nu_2$ outside $\{x,y\} \cup \bar{x} \cup \bar{y}$ 
    which we get from (4-5) and (9);
    \item $\nu_2'''(x) = \nu_1(x) = \nu_2(x)$ which we get from (10) and (1);
    \item $\nu_2'''(y) = \nu_2''(x) = \nu_2'(x) = \nu_2(y)$ which we get 
    from (11), (5), and (3);
    \item $\nu_2'''(\bar x) = \nu_2''(\bar x) = \nu_1(\bar x) = \nu_2(\bar x)$ 
    which we get from (9), (6), and (1); and
    \item $\nu_2'''(\bar y) = \nu_2''(\bar y) = \nu_2'(\bar x) = \nu_2(\bar y)$ 
    which we get from (9), (7), and (3).
\end{itemize}
\end{proof}
}
\begin{lemma}\label{lem:derived}
Let $\bar{x}$ and $\bar{y}$ be disjoint tuples of distinct variables of the 
same length.  Then, for any $\bbar$ $B$, we have
\[
\mvr{\bar{x}\to\bar{y}}(B) =  
\sellr{\bar{x}=\bar{x}} \cylr{\bar{x}} \selr{\bar{x}=\bar{y}} \cylr{\bar{y}}(B).
\]
\end{lemma}
\begin{proof}
We give a ``proof by picture''.  Consider an arbitrary $(\nu_1, \nu_2) \in B$: 
\[\begin{array}{ccccccc}
\toprule
\bar x & \bar y & rest & & 
\bar x & \bar y & rest  \\
\toprule
\bar a & \bar b & \bar c & &
\bar d & \bar e & \bar f \\
\bottomrule
\end{array}\]

% \begin{center}
%     \includegraphics[scale=0.28]{KR-Journal/figures/input.png}
% \end{center}

We will verify that when we apply the LHS and the RHS on this pair of valuations, 
we obtain identical results.

For the LHS,  we see that $\mvr{\bar{x}\to\bar{y}}(B)$ yields 
the following pair of valuations when applied on $(\nu_1, \nu_2)$:
\[\begin{array}{ccccccc}
\toprule
\bar x & \bar y & rest & &
\bar x & \bar y & rest  \\
\toprule
\bar a & \bar b & \bar c & &
\bar a & \bar d & \bar f \\
\bottomrule
\end{array}\]

Now, we check the RHS.  We see that the following set of pairs of valuations is
the result of $\cylr{\bar{y}}(B)$ when applied on $(\nu_1, \nu_2)$:
\[\begin{array}{ccccccc}
\toprule
\bar x & \bar y & rest & &
\bar x & \bar y & rest  \\
\toprule
\bar a & \bar b & \bar c & &
\bar d & * & \bar f \\
\bottomrule
\end{array}\]
Here the asterisk denotes a ``wildcard'', i.e., any valuation on $\bar y$ is
allowed.

Then, we see that $\selr{\bar{x}=\bar{y}} \cylr{\bar{y}}(B)$ yields:
\[\begin{array}{ccccccc}
\toprule
\bar x & \bar y & rest & &
\bar x & \bar y & rest  \\
\toprule
\bar a & \bar b & \bar c & &
\bar d & \bar d & \bar f \\
\bottomrule
\end{array}\]
Next, we see that $\cylr{\bar{x}} \selr{\bar{x}=\bar{y}} \cylr{\bar{y}}(B)$ 
yields:
\[\begin{array}{ccccccc}
\toprule
\bar x & \bar y & rest & &
\bar x & \bar y & rest  \\
\toprule
\bar a & \bar b & \bar c & &
* & \bar d & \bar f \\
\bottomrule
\end{array}\]

Finally, we see that 
$\sellr{\bar{x}=\bar{x}} \cylr{\bar{x}} \selr{\bar{x}=\bar{y}} \cylr{\bar{y}}(B)$ 
yields the following pair of valuations which is the same as the result of the LHS.
\[\begin{array}{ccccccc}
\toprule
\bar x & \bar y & rest & &
\bar x & \bar y & rest  \\
\toprule
\bar a & \bar b & \bar c & &
\bar a & \bar d & \bar f \\
\bottomrule
\end{array}\]
\end{proof}
\ignore{\begin{proof}
The intuition for this expression is as follows.  We first cylindrify the 
$\bar y$ variables on the right so that the subsequent selection effectively 
copies $\bar x$ to $\bar y$.  The final two operations make sure that the 
$\bar x$ variables are inertial.

To formally prove our intuition, we can see that when $\bar{x}$ and $\bar{y}$ are 
disjoint tuples of distinct variables of the same length, then the following 
expression
\[\mvr{x_1\to y_1}(\mvr{x_2\to y_2}(\ldots (\mvr{x_n\to y_n}(B))))\]
is equivalent to $\mvr{\bar{x}\to\bar{y}}(B)$, where $\bar x = x_1, \ldots, x_n$
and $\bar y = y_1, \ldots, y_n$.  This equivalence can be shown by Proposition
\ref{prop:moveEq}.  In that case, it suffices to show that the following holds 
for $i = 1\ldots n$, and valuations $\nu_1$ and $\nu_2$:
\[(\nu_1, \nu_2) \in \mvr{x_i \to y_i}(B) \text{ iff } 
(\nu_1, \nu_2) \in \sellr{x_i=x_i} \cylr{x_i} \selr{x_i=y_i} \cylr{y_i}(B)\]

In what follows, we refer to $x_i$ and $y_i$ simply as $x$ and $y$; respectively. 
We begin with the ``if''-direction.  Let 
$(\nu_1, \nu_2) \in \sellr{x=x} \cylr{x} \selr{x=y} \cylr{y}(B)$.  We can see 
that:
\begin{enumerate}
    \item $(\nu_1, \nu_2) \in \cylr{x} \selr{x=y} \cylr{y}(B)$; and 
    \item $\nu_1(x) = \nu_2(x)$.
\end{enumerate}
From (1), we can see that there exists $\nu'_2$ such that:
\begin{enumerate} \setcounter{enumi}{2}
    \item $(\nu_1, \nu'_2) \in \selr{x=y} \cylr{y}(B)$; and
    \item $\nu'_2 = \nu_2$ outside $\{x\}$.
\end{enumerate}
From (3), we can see that:
\begin{enumerate} \setcounter{enumi}{4}
    \item $(\nu_1, \nu'_2) \in \cylr{y}(B)$; and
    \item $\nu'_2(x) = \nu'_2(y)$.
\end{enumerate}
From (5), we can then see that there exists $\nu''_2$ such that:
\begin{enumerate} \setcounter{enumi}{6}
    \item $(\nu_1, \nu''_2) \in B$; and
    \item $\nu''_2 = \nu'_2$ outside $\{y\}$.
\end{enumerate}
Moreover, from (8), (6), and (4), we know that 
$\nu''_2(x) = \nu'_2(x) = \nu'_2(y) = \nu_2(y)$. 
What remains to show is that $\nu_2 = \nu''_2$ outside $\{x,y\}$.  Indeed, this 
follows from (4) and (8).

We now prove the ``only if''-direction.  
Let $(\nu_1, \nu_2) \in \mvr{x \to y}(B)$.  We know that:
\begin{enumerate}
    \item $\nu_2(x) = \nu_1(x)$;
    \item there exists $\nu'_2$ such that $(\nu_1, \nu'_2) \in B$;
    \item $\nu'_2(x) = \nu_2(y)$; and
    \item $\nu_2 = \nu'_2$ outside $\{x, y\}$.
\end{enumerate}
Take $\nu''_2$ to be the valuation $\nu'_2[y \mapsto \nu'_2(x)]$.  By 
construction, we know that $\nu''_2(x) = \nu''_2(y)$.  Along with (2), we can 
see that $(\nu_1, \nu''_2) \in \cylr{y}(B)$ and 
$(\nu_1, \nu''_2) \in \selr{x=y} \cylr{y}(B)$.  Now, take $\nu'''_2$ to be the 
valuation $\nu''_2[x \mapsto \nu_1(x)]$.  By construction, we know that
$\nu_1(x) = \nu'''_2(x)$.  Hence, 
$(\nu_1, \nu'''_2) \in \cylr{x} \selr{x=y} \cylr{y}(B)$ and
$(\nu_1, \nu'''_2) \in \sellr{x=x} \cylr{x} \selr{x=y} \cylr{y}(B)$.  
What remains to show is that $\nu_2 = \nu'''_2$.  Indeed, by construction and 
(4), we can see that $\nu_2 = \nu'''_2$ outside $\{x,y\}$.  Also, by construction 
and (1), we can see that $\nu_2(x) = \nu_1(x) = \nu'''_2(x)$.  Finally, by 
construction and (3), we can see that 
$\nu_2(y) = \nu'_2(x) = \nu''_2(y) = \nu'''_2(y)$.
\end{proof}
}

\begin{lemma}\label{lem:mvr-comp}
Let $A$ and $B$ be $\bbar$s and let $\bar{x}$ and $\bar{y}$ be disjoint tuples 
of distinct variables of the same length such that all variables in $\bar{y}$ 
are inertially cylindrified in $A$ and $B$. In that case:
\[A\comp B = \mvr{\bar{y}\to\bar{x}} ( A\comp \mvr{\bar{x}\to\bar{y}}(B) )\]
\end{lemma}
What this lemma shows is that we can temporarily move certain variables (the $\bar{x}$) 
to unused variables (the $\bar{y}$) and then move them back.  The proof of this lemma is:
\begin{proof}[Proof of Lemma~\ref{lem:mvr-comp}]
Again we give a proof by picture.  Let the left be a generic pair of valuations that 
belongs to $A$, while the one on the right be a generic one that belongs to $B$.  
The ``$-$'' here represents inertial cylindrification.
\begin{center}
$\begin{array}{ccccccc}
\toprule
\bar x & \bar y & rest & &
\bar x & \bar y & rest  \\
\toprule
\bar a & - & \bar b & &
\bar c & - & \bar d \\ 
\bottomrule
\end{array}$
\qquad
$\begin{array}{ccccccc}
\toprule
\bar x & \bar y & rest & &
\bar x & \bar y & rest  \\
\toprule
\bar e & - & \bar f & &
\bar g & - & \bar h \\
\bottomrule
\end{array}$
\end{center}

For the LHS, we see that composition can only be applied if $\bar c = \bar e$ 
and $\bar d = \bar f$.  Under this assumption, we get that $A \comp B$ yields the 
following:
\[\begin{array}{ccccccc}
\toprule
\bar x & \bar y & rest & &
\bar x & \bar y & rest  \\
\toprule
\bar a & - & \bar b & &
\bar g & - & \bar h \\
\bottomrule
\end{array}\]

Now, we check the RHS.  We see that $\mvr{\bar x \to \bar y}(B)$ yields the
following when applied on the generic pair belonging to $B$:
\[\begin{array}{ccccccc}
\toprule
\bar x & \bar y & rest & &
\bar x & \bar y & rest  \\
\toprule
\bar e & * & \bar f & &
\bar e & \bar g & \bar h \\
\bottomrule
\end{array}\]

To apply the composition in the RHS, we must have $\bar c = \bar e$ and 
$\bar d = \bar f$, which are the same restrictions we had in applying the composition 
in the LHS, so the expression $A\comp \mvr{\bar{x}\to\bar{y}}(B)$ yields:
\[\begin{array}{ccccccc}
\toprule
\bar x & \bar y & rest & &
\bar x & \bar y & rest  \\
\toprule
\bar a & * & \bar b & &
\bar e & \bar g & \bar h \\
\bottomrule
\end{array}\]

Finally, applying the last move operation,
$\mvr{\bar{y}\to\bar{x}}(A\comp \mvr{\bar{x}\to\bar{y}}(B))$ yields:
\[\begin{array}{ccccccc}
\toprule
\bar x & \bar y & rest & &
\bar x & \bar y & rest  \\
\toprule
\bar a & - & \bar b & &
\bar g & - & \bar h \\
\bottomrule
\end{array}\]
which is clearly identical to what we had from the LHS.
\end{proof}
\ignore{
\begin{proof}
We first show that 
$A\comp B \subseteq \mvr{\bar{y}\to\bar{x}}(A\comp \mvr{\bar{x}\to\bar{y}}(B))$.
Let $(\nu_1, \nu_3) \in A \comp B$.  We know then that there is a valuation 
$\nu_2$ such that:
\begin{enumerate}
    \item $(\nu_1, \nu_2) \in A$; and
    \item $(\nu_2, \nu_3) \in B$.
\end{enumerate}
We also know that:
\begin{enumerate} \setcounter{enumi}{2}
    \item $\nu_1 = \nu_2$ on $\bar y$; and
    \item $\nu_2 = \nu_3$ on $\bar y$.
\end{enumerate}
since all variables in $\bar{y}$ are inertially cylindrified in $A$ and $B$.
Take $\nu_3'$ be the valuation:
\begin{enumerate} \setcounter{enumi}{4}
    \item $\nu_3'(\bar x) = \nu_2(\bar x)$;
    \item $\nu_3'(\bar y) = \nu_3(\bar x)$; and
    \item $\nu_3' = \nu_3$ outside $\bar x \cup \bar y$.
\end{enumerate}
From the structure of $\nu_3'$, we can see that 
$(\nu_2, \nu_3') \in \mvr{\bar{x}\to\bar{y}}(B)$.  Along with (1), we can see 
that $(\nu_1, \nu_3') \in A\comp\mvr{\bar{x}\to\bar{y}}(B)$.  Now, take 
$\nu_3''$ to be the valuation:
\begin{enumerate} \setcounter{enumi}{7}
    \item $\nu_3''(\bar y) = \nu_1(\bar y)$;
    \item $\nu_3''(\bar x) = \nu_3'(\bar y)$; and
    \item $\nu_3'' = \nu_3'$ outside $\bar x \cup \bar y$.
\end{enumerate}
From the structure of $\nu_3''$, we can see that $(\nu_1, \nu_3'') \in 
\mvr{\bar{y}\to\bar{x}}(A\comp \mvr{\bar{x}\to\bar{y}}(B))$.
Moreover, we can see that:
\begin{itemize}
    \item $\nu_3''(\bar y) = \nu_1(\bar y) = \nu_2(\bar y) = \nu_3(\bar y)$ 
    from (8), and (3-4);
    \item $\nu_3''(\bar x) = \nu_3'(\bar y) = \nu_3(\bar x)$ 
    from (9) and (6); and
    \item $\nu_3'' = \nu_3' = \nu_3$ outside $\bar x \cup \bar y$ 
    from (10) and (7).
\end{itemize}
Thus, $\nu_3'' = \nu_3$, whence, $(\nu_1, \nu_3) \in 
\mvr{\bar{y}\to\bar{x}}(A\comp \mvr{\bar{x}\to\bar{y}}(B))$.

Now, we show the other inclusion.  Let $(\nu_1, \nu_5) \in
\mvr{\bar{y}\to\bar{x}}(A\comp \mvr{\bar{x}\to\bar{y}}(B))$.
Thus,
\begin{enumerate}
    \item $\nu_5(\bar y) = \nu_1(\bar y)$;
    \item $(\nu_1, \nu_4) \in A\comp \mvr{\bar{x}\to\bar{y}}(B)$;
    \item $\nu_5(\bar x) = \nu_4(\bar y)$; and
    \item $\nu_5 = \nu_4$ outside $\bar{x} \cup \bar{y}$.
\end{enumerate}
for some valuation $\nu_4$.  From (2), we know that there is a valuation
$\nu_2$ such that:
\begin{enumerate} \setcounter{enumi}{4}
    \item $(\nu_1, \nu_2) \in A$; and
    \item $(\nu_2, \nu_4) \in \mvr{\bar{x}\to\bar{y}}(B)$.
\end{enumerate}
From (6), we know that:
\begin{enumerate} \setcounter{enumi}{6}
    \item $\nu_4(\bar x) = \nu_2(\bar x)$;
    \item $(\nu_2, \nu_3) \in B$;
    \item $\nu_4(\bar y) = \nu_3(\bar x)$; and
    \item $\nu_4 = \nu_3$ outside $\bar{x} \cup \bar{y}$.
\end{enumerate}
for some valuation $\nu_3$.  From (5) and (8), we know that 
$(\nu_1, \nu_3) \in A\comp B$.  Moreover, we know that:
\begin{enumerate} \setcounter{enumi}{10}
    \item $\nu_1 = \nu_2$ outside $\bar{x} \cup \bar{y}$; and
    \item $\nu_2 = \nu_3$ outside $\bar{x} \cup \bar{y}$.
\end{enumerate}
since all variables in $\bar{y}$ are inertially cylindrified in $A$ and $B$.
We can establish that $(\nu_1, \nu_5) \in A \comp B$ if we show that 
$\nu_3 = \nu_5$.  Indeed,
\begin{itemize}
    \item $\nu_3(\bar{x}) = \nu_4(\bar{y}) = \nu_5(\bar{x})$ from (9) and (3);
    \item $\nu_3(\bar{y}) = \nu_2(\bar{y}) = \nu_1(\bar{y}) = \nu_5(\bar{y})$
    from (12), (11), and (1); and
    \item $\nu_3 = \nu_4 = \nu_5$ outside $\bar{x} \cup \bar{y}$
    from (10) and (4).\qedhere
\end{itemize}
\end{proof}
}

This finally brings us to the main result of the current subsection.
\begin{theorem}\label{thm:infiniteNonPrim}
If\/ $\V$ is infinite, then every LIF expression is equivalent to a LIFnc 
expression.
\end{theorem}
\ignore{
\begin{proof}
For a LIF expression $\alpha\comp\beta$ we can choose a tuple of variables 
$\bar{y}$ of the same length as $\synO{(\beta)}$, such that $\bar{y}$ neither
occurs in $\alpha$ not in $\beta$.  In that case, $\bar{y}$ is inertially 
cylindrified in $\alpha$ and in $\beta$ and hence Lemma~\ref{lem:mvr-comp} 
yields that 
$\alpha \comp \beta$ is equivalent to
\[ \mvr{\bar{y}\to \synO(\beta)} 
(\alpha \comp \mvr{\synO(\beta)\to\bar{y}}(\beta)).
\]
We will next show that $ \mvr{\synO(\beta)\to\bar{y}}(\beta)$ is  
io-disjoint.  Indeed, from the equivalence in Lemma~\ref{lem:derived} and the 
soundness of our definitions, we can see that
\begin{itemize}
    \item $\semO(\mvr{\synO(\beta)\to\bar{y}}(\beta)) =
    \semO(\sellr{\bar{x}=\bar{x}} \cylr{\bar{x}} \selr{\bar{x}=\bar{y}} \cylr{\bar{y}}(\beta)) 
    \subseteq \synO(\sellr{\bar{x}=\bar{x}} \cylr{\bar{x}} \selr{\bar{x}=\bar{y}} \cylr{\bar{y}}(\beta)) 
    = \synO(\beta) \cup \bar{y}$.  
    Moreover, we see that 
    $\semO(\mvr{\bar x = \bar y}(\gamma)) \cap \bar x = \emptyset$ for any 
    $\bar x$ and any LIF expression $\gamma$ in which $\bar y$ is inertially
    cylindrified.  Clearly, 
    $\semO(\mvr{\synO(\beta)\to\bar{y}}(\beta)) \subseteq \bar y$;
    \item $\semI(\mvr{\synO(\beta)\to\bar{y}}(\beta)) = 
    \semI(\sellr{\bar{x}=\bar{x}} \cylr{\bar{x}} \selr{\bar{x}=\bar{y}} \cylr{\bar{y}}(\beta)) 
    \subseteq \synI(\sellr{\bar{x}=\bar{x}} \cylr{\bar{x}} \selr{\bar{x}=\bar{y}} \cylr{\bar{y}}(\beta)) 
    = \synI(\beta) \cup \synO(\beta)$.
\end{itemize}
Since $\bar y$ does not occur in $\beta$, we see that 
\[\semO(\mvr{\synO(\beta)\to\bar{y}}(\beta)) 
\subseteq \bar y \cap (\synI(\beta) \cup \synO(\beta)) = \emptyset.\]
By establishing that $\mvr{\synO(\beta)\to\bar{y}}(\beta)$ is io-disjoint,
we can apply Theorem~\ref{theo:IOdisjointcompnotprimSYN} to eliminate the
composition yielding the expression
\[\mvr{\bar y \to \synO(\beta)}
(\compEqSyn{\alpha}{\mvr{\synO(\beta)\to\bar{y}}(\beta)}).\]
This construction can be applied recursively to eliminate all 
occurrences of composition.
\end{proof}
}
\begin{proof}
We prove this theorem by induction on the number of compositions operators in 
a LIF expression $\gamma$.  The base case (no composition operators), is
trivial.  For the inductive case, consider an expression $\eta$ containing
at least one composition operator.  We show how to rewrite $\eta$
equivalently with one composition operator less.  Thereto, take any
subexpression $\alpha \comp \beta$ such that $\alpha$ and
$\beta$ are LIFnc expressions.  We eliminate this composition as follows.
Choose a tuple of variables $\bar{y}$ of the same length as $\synO{(\beta)}$, 
such that $\bar{y}$ does not occur in $\gamma$.  In that case, $\bar{y}$ is 
inertially cylindrified in $\alpha$ and in $\beta$, and hence, 
Lemma~\ref{lem:mvr-comp} yields that $\alpha \comp \beta$ is equivalent to
\[ 
\mvr{\bar{y}\to \synO(\beta)} 
(\alpha \comp \mvr{\synO(\beta)\to\bar{y}}(\beta)).
\]
We will next show that $\mvr{\synO(\beta)\to\bar{y}}(\beta)$ is  
io-disjoint.  Indeed, from the equivalence in Lemma~\ref{lem:derived} and the 
soundness of our definitions, we can see that
\[
\semO(\mvr{\synO(\beta)\to\bar{y}}(\beta)) =
\semO(\sellr{\bar{x}=\bar{x}} \cylr{\bar{x}} \selr{\bar{x}=\bar{y}} \cylr{\bar{y}}(\beta)) 
\subseteq \synO(\sellr{\bar{x}=\bar{x}} \cylr{\bar{x}} \selr{\bar{x}=\bar{y}} \cylr{\bar{y}}(\beta)) 
= \synO(\beta) \cup \bar{y}. 
\]
Moreover, we generally have $\semO(\mvr{\bar x \to \bar y}(\gamma)) \cap \bar x = \emptyset$ 
for any $\bar x$ and any LIF expression $\gamma$ in which $\bar y$ is inertially
cylindrified.  As a consequence, $\semO(\mvr{\synO(\beta)\to\bar{y}}(\beta)) \subseteq \bar y$.

Also, we can see that
\[\semI(\mvr{\synO(\beta)\to\bar{y}}(\beta)) = 
\semI(\sellr{\bar{x}=\bar{x}} \cylr{\bar{x}} \selr{\bar{x}=\bar{y}} \cylr{\bar{y}}(\beta)) 
\subseteq \synI(\sellr{\bar{x}=\bar{x}} \cylr{\bar{x}} \selr{\bar{x}=\bar{y}} \cylr{\bar{y}}(\beta)) 
= \synI(\beta) \cup \synO(\beta).\]
% \begin{itemize}
%     \item $\semO(\mvr{\synO(\beta)\to\bar{y}}(\beta)) =
%     \semO(\sellr{\bar{x}=\bar{x}} \cylr{\bar{x}} \selr{\bar{x}=\bar{y}} \cylr{\bar{y}}(\beta)) 
%     \subseteq \synO(\sellr{\bar{x}=\bar{x}} \cylr{\bar{x}} \selr{\bar{x}=\bar{y}} \cylr{\bar{y}}(\beta)) 
%     = \synO(\beta) \cup \bar{y}$.  
%     Moreover, we generally have
%     $\semO(\mvr{\bar x \to \bar y}(\gamma)) \cap \bar x = \emptyset$ for any 
%     $\bar x$ and any LIF expression $\gamma$ in which $\bar y$ is inertially
%     cylindrified.  As a consequence, 
%     $\semO(\mvr{\synO(\beta)\to\bar{y}}(\beta)) \subseteq \bar y$.
%     \item $\semI(\mvr{\synO(\beta)\to\bar{y}}(\beta)) = 
%     \semI(\sellr{\bar{x}=\bar{x}} \cylr{\bar{x}} \selr{\bar{x}=\bar{y}} \cylr{\bar{y}}(\beta)) 
%     \subseteq \synI(\sellr{\bar{x}=\bar{x}} \cylr{\bar{x}} \selr{\bar{x}=\bar{y}} \cylr{\bar{y}}(\beta)) 
%     = \synI(\beta) \cup \synO(\beta)$.
% \end{itemize}
Since $\bar y$ does not occur in $\beta$, we indeed obtain that is 
$\mvr{\synO(\beta)\to\bar{y}}(\beta)$ io-disjoint. 
\[\semO(\mvr{\synO(\beta)\to\bar{y}}(\beta)) 
\subseteq \bar y \cap (\synI(\beta) \cup \synO(\beta)) = \emptyset.\]
We can now apply Theorem~\ref{theo:IOdisjointcompnotprimSYN} to eliminate the
composition yielding the LIFnc expression
\[\mvr{\bar y \to \synO(\beta)}
(\compEqSyn{\alpha}{\mvr{\synO(\beta)\to\bar{y}}(\beta)}).\]
\qedhere
\end{proof}

\subsection{If \texorpdfstring{$\V$}{V} is Finite, Composition is Primitive}
The case that remains is when $\V$ is finite.
We will show that in this case, composition is indeed primitive by relating 
bounded-variable LIF to bounded-variable first-order logic.

Assume $\V=\{x_1,\ldots,x_n\}$. Since $\bbar$s involve pairs of $\V$-valuations, we 
introduce a copy  $\V_y = \{y_1,\ldots,y_n\}$ disjoint from $\V$.
% and $\V_z= \{z_1,\ldots,z_n\}$.
For clarity, we also write $\V_x$ for $\V$.
As usual, by $\FO{k}$ we denote the fragment of first-order logic that uses only $k$ 
distinct variables.
% If $V$ is a set of variables, we define $\FO{V}$ as the set of first-order formulas that only use variables in $V$.
%
We observe:
\begin{proposition}\label{prop:LIF-BoundedFO}
For every LIF expression $\alpha$, there exists an $\FO{3n}$ formula $\varphi_\alpha$ 
with free variables in $\V_x\cup \V_y$ such that
\[
(\nu_1,\nu_2)\in \sem{\alpha} \text{ iff } 
\inst{},(\nu_1\cup\nu_2')\models \varphi_\alpha,
\]
where $\nu_2'$ is the $\V_y$-valuation such that $\nu_2'(y_i) = \nu_2(x_i)$ 
for each $i$.  Furthermore, if $\alpha$ is a LIFnc expression, $\varphi_\alpha$ can be 
taken to be a $\FO{2n}$ formula.
%
%
%
%   and for every $u\neq v \in \{x,y,z\}$ there exists a formula $\varphi^{uv}$ with free variables in $\V_u\cup \V_v$ in $\FO{\V_x\cup \V_y\cup \V_z}$ such that for every \interpretation $\inst{}$:
%   \[(\nu_1,\nu_2)\in \sem{\alpha}\text{ \ iff \ } \inst{},(\nu_1\circ \rho_{ux}\cup\nu_2\circ \rho_{vx})\models \varphi^{uv}.\]
  % Then,there is a (x_1,\ldots,x_n,y_1,\ldots,y_n)$ such that \[\{(\nu_1,\nu_2)\mid I\models \varphi[\nu_1,\nu_2'] \text{ and } \nu_2(x_i)=\nu_2'(y_i)\}= \sem{\alpha}\]
\end{proposition}
\begin{proof}
 The proof is by induction on the structure of $\alpha$ (using Lemma \ref{lem:redundant}, 
 we omit redundant operators).

 We introduce a third copy $\V_z= \{z_1,\ldots,z_n\}$ of $\V$.
 For every $u,v\in \{x,y,z\}$ we define $\rho_{uv}$ as follows:
\begin{align*}
   \rho_{uv} :\V_u \rightarrow \V_v: u_i \mapsto v_i
 \end{align*}
Using these functions, we can translate a valuation $\nu$ on $\V=\V_x$ to a corresponding 
valuation on $\V_u$ with $u \in \{y,z\}$. Clearly, $\nu \circ \rho_{ux}$ does this job. %, where $\circ$ denotes function composition as usual.

In the proof, we actually show a stronger statement by induction, namely that for each 
$\alpha$ and for every $u\neq v \in \{x,y,z\}$ there is a formula $\varphi^{uv}_\alpha$ with 
free variables in $\V_u\cup \V_v$ in $\FO{\V_x\cup \V_y\cup \V_z}$ such that for every  $\inst{}$:
\[
(\nu_1,\nu_2)\in \sem{\alpha}
\text{ \ iff \ } \inst{},(\nu_1\circ \rho_{ux}\cup\nu_2\circ \rho_{vx})\models \varphi^{uv}_\alpha.
\]
Since the notations $x$, $y$, $z$, $u$ and $v$ are taken, 
we use notations $a$, $b$ and $c$ for variables.
\begin{itemize}
  \item $\alpha = \id$.  
  Take $\varphi^{uv}_\alpha$ to be $\bigwedge_{i=1}^n u_i=v_i$.
  \item $\alpha = M(\overline{a};\overline{b})$.
  Take $\varphi^{uv}_\alpha$ to be 
  $M(\rho_{xu}(\overline{a}),\rho_{xv}(\overline{b}))\land 
  \bigwedge_{c \not \in \overline{b}} \rho_{xu}(c)= \rho_{xv}(c)$
  \item $\alpha = \alpha_1 \cup \alpha_2$.
  Take $\varphi_\alpha^{uv}$ to be 
  $\varphi^{uv}_{\alpha_1} \lor \varphi^{uv}_{\alpha_2}$.
  \item $\alpha = \alpha_1 \setminus \alpha_2$.
  Take $\varphi_\alpha^{uv}$ to be 
  $\varphi_{\alpha_1}^{uv}\land\neg \varphi_{\alpha_2}^{uv}$.
  \item $\alpha = \alpha_1\comp{}\alpha_2$. Let $w \in \{x,y,z\}\setminus \{u,v\}$. 
  Take $\varphi_\alpha^{uv}$ to be
  $\exists w_1\ldots\exists w_n\;(\varphi_{\alpha_1}^{uw}\land \varphi_{\alpha_2}^{wv})$.
  \withconverse{\item $\alpha = \conv{\alpha_1}$.  
  By induction, $\varphi^{vu}_{\alpha_1}$ exists.  
  This formula can serve as $\varphi_\alpha^{uv}$.}
  \item $\alpha = \sellr{a=b}(\alpha_1)$.
  Take $\varphi_\alpha^{uv}$ to be 
  $\varphi_{\alpha_1}^{uv}\land \rho_{xu}(a) = \rho_{xv}(b)$.
  \item $\alpha = \cyll{a}(\alpha_1)$.
  Take $\varphi_\alpha^{uv}$ to be 
  $\exists \rho_{xu}(a)\; \varphi_{\alpha_1}^{uv}$.\qedhere
\end{itemize}
\end{proof}%
\fullproof{
\begin{proof}
 The proof is by induction on the structure of $\alpha$ (using Lemma~\ref{lem:redundant}, 
 we omit redundant operators). Below, we write $\varphi^{uv}\equiv \alpha$ to indicate the 
 equivalence stated in the Proposition. Since the notations $x$, $y$, $z$, $u$ and $v$ are 
 taken, we use notations $a$, $b$ and $c$ for variables.
\begin{itemize}
  \item $\alpha = \id$. For $\varphi^{uv}$ we can take $\bigwedge_{i=1}^n u_i=v_i$.
  \item $\alpha = M(\overline{a};\overline{b})$. For $\varphi^{uv}$ we can take 
  $M(\rho_{xu}(\overline{a}),\rho_{xv}(\overline{b}))\land 
  \bigwedge_{c \not \in \overline{b}} \rho_{xu}(c)= \rho_{xv}(c)$.
   % Then $\varphi(\ox,\oy) = M(\overline{u},f(\overline{v})) \land \bigwedge_{x_i \in  \V \setminus \overline{v}} x_i = y_i$. Clearly, $\varphi\equiv \alpha$ and $\varphi \in \FO{3n}$.
  \withconverse{\item $\alpha = \conv{\alpha_1}$. By induction, there exists 
  $\varphi_1^{vu}\equiv \alpha_1$. This formula can serve as $\varphi^{uv}$.}
  \item $\alpha = \alpha_1 \cup \alpha_2$. By induction there exists $\varphi_1^{uv}$ 
  and $\varphi_2^{uv}$ in $\FO{\V_x\cup\V_y\cup \V_z}$ such that $\varphi_1^{uv}\eqq \alpha_1$ 
  and $\varphi_2^{uv}\eqq \alpha_2$. For $\varphi^{uv}$ we can take 
  $\varphi_1^{uv}\lor \varphi_2^{uv}$.
%   \item $\alpha =\alpha_1 \cap \alpha_2$. By induction there exists $\varphi_1^{uv}$ and $\varphi_2^{uv}$ in $\FO{\V_x\cup\V_y\cup \V_z}$ such that $\varphi_1^{uv}\eqq \alpha_1$ and $\varphi_2^{uv}\eqq \alpha_2$. For $\varphi^{uv}$ we can take $\varphi_1^{uv}\land \varphi_2^{uv}$.
  \item $\alpha = \cyll{a}(\alpha_1)$. By induction there exists $\varphi_1^{uv} \in \FO{\V_x\cup\V_y\cup \V_z}$ 
  such that $\varphi_1^{uv}\eqq \alpha_1$. For $\varphi^{uv}$ we can take $\exists \rho_{xu}(a)\; \varphi_1^{uv}$.
  \item $\alpha = \alpha_1 \setminus \alpha_2$. By induction there exists $\varphi_1^{uv},\varphi_2^{uv}\in \FO{\V_x\cup\V_y\cup \V_z}$ 
  such that $\varphi_1^{uv}\eqq \alpha_1$ and $\varphi_2^{uv}\eqq \alpha_2$. For $\varphi^{uv}$ 
  we can take $\varphi_1^{uv}\land\neg \varphi_2^{uv}$.
  \item $\alpha = \sellr{a=b}(\alpha_1)$. By induction there exists $\varphi_1^{uv} \in \FO{\V_x\cup\V_y\cup \V_z}$ 
  such that $\varphi_1^{uv}\equiv \alpha_1$. For $\varphi^{uv}$ we can take 
  $\varphi_1^{uv}\land \rho_{xu}(x) = \rho_{xv}(b)$.
%   \item $\alpha = \selb{a=b}(\alpha_1)$. By induction there exists $\varphi_1^{uv} \in \FO{\V_x\cup\V_y\cup \V_z}$ such that $\varphi_1^{uv}\equiv \alpha_1$. For $\varphi^{uv}$ we can take $\varphi_1^{uv}\land \rho_{xu}(a) = \rho_{xu}(b) \land \rho_{xv}(a) = \rho_{xv}(b)$.
  \item $\alpha = \alpha_1\comp{}\alpha_2$. Let $w \in \{x,y,z\}\setminus \{u,v\}$. By induction there exists 
  $\varphi_1^{uw},\varphi_1^{wv}\in \FO{\V_x\cup\V_y\cup \V_z}$ such that $\varphi_1^{uw}\eqq \alpha_1$ and 
  $\varphi_2^{wv}\eqq \alpha_2$. Define $\varphi^{uv} = \exists w_1,\ldots,w_n\;(\varphi_1^{uw}\land \varphi_2^{wv})$. 
  Clearly, $\varphi\equiv \alpha$. Moreover, since $\FV(\varphi_1^{uw})\subseteq \V_u\cup\V_w$, $\FV(\varphi_2^{wv})\subseteq \V_w\cup\V_v$ 
  and both are in $\FO{\V_x\cup\V_y\cup \V_z}$, we know that $\varphi^{uv} \in \FO{\V_x\cup\V_y\cup \V_z}$ as desired.\qedhere
\end{itemize}
\end{proof}}
Now that we have established that LIFnc can be translated into $\FO{2n}$, all that is left to do is find a 
Boolean query that can be expressed in LIF with $n$ variables, but not in $\FO{2n}$.
We find such a query in the existence of a $3n$-clique.
We will first show that we can construct a LIFnc expression $\alpha_{2n}$ such that, given an \interpretation 
$\inst$ interpreting a binary relation $R$, $\sem{\alpha_{2n}}$ consists of all $2n$-cliques of $R$.
Next, we show how $\alpha_{2n}$ can be used (with composition)
to construct an expression $\alpha_{\exists3n}$ such that $\sem{\alpha_{\exists 3n}}$ is non-empty 
if and only if $R$ has a $3n$-clique. Since this property cannot be expressed in $\FO{2n}$, we can conclude 
that composition must be primitive.

To avoid confusion, we recall that a set $L$ of $k$ data elements is a $k$-clique in a binary relation $R$, 
if any two distinct $a$ and $b$ in $L$, we have $(a,b)\in R$ (and also $(b,a)\in R$).
%
% To show that composition is primitive in $n$-variable LIF we proceed as follows. 
% First, we show that we can compute $2n$-cliques in graphs, which is then used to 
% show that we can compute $3n$-cliques. Finally, we show that this query cannot 
% be expressed in LIF without composition.
\begin{proposition}\label{lem:2ncliques}
Suppose that $|\V| = n$ with $n \geq 2$ and let $\Sch = \{R\}$ with 
$\ar{(R)}=\iar{(R)}=2$. There exists a LIF expression $\alpha_{2n}$ such that
\[\sem{\alpha_{2n}} = \{(\nu_1,\nu_2)\mid \nu_1(\V)\cup \nu_2(\V) 
\text{ is a $2n$-clique in }\inst(R)\}.\]
\end{proposition}
\begin{proof}
Throughout this proof, we identify a pair $(\nu_1,\nu_2)$ of two valuations with 
the $2n$ tuple of data elements 
\[\nu_1(x_1,\dots,x_n)\cdot\nu_2(x_1,\dots,x_n).\]
%   Let $x,y \in \V$ be different variables.
Before coming to the actual expression for $\alpha_{2n}$, we introduce some 
auxiliary concepts.  First, we define
\[\all := \cyll{\V}\,\cylr{\V}(\id).\]
It is clear that
\[\sem{\all} = \{(\nu_1,\nu_2)\in \sval{}\times\sval{}\}.\]
A first condition for being a $2n$-clique is that all data elements are different.
It is clear that the expression
\[\alpha_{=} := \bigcup_{x\neq y \in \V} \bigl(\sell{x=y}(\all)\cup 
\selr{x=y}(\all) \bigr)\cup \bigcup_{x,y\in \V}\sellr{x=y}(\all)\]
has the property that $\sem{\alpha_{=}}$ consists of all $2n$-tuples where at 
least one data element is repeated.  Hence, $\sem{\alpha_{\neq}}$ consists of 
all $2n$-tuples of distinct data elements, where
\[\alpha_{\neq} := \all \setminus \alpha_{=}\]

The second condition for being a $2n$-clique is that each two distinct 
elements are connected by $R$.  In order to check this, we define the 
following expressions for each two variables $x$ and $y$:
%   For the left-hand side, this is straightforwardly expressed using $R$, e.g., 
% $R(x_1,x_2)$ expresses that the data element on the first and 
% second position need to be connected.
%   For the ones on the right-hand side or the left-right connection, 
% some shuffling around of variables needs to happen.
%   We define
%
\begin{align*}
R^{l}_{x,y} &:= \cyll{\V\setminus\{x,y\}}\cylr{\V} (R(x,y;)\cap R(y,x;))\\
R^{r}_{x,y} &:= \cyll{\V}\cylr{\V\setminus\{x,y\}} (R(x,y;)\cap R(y,x;))\\
R^{lr}_{x,y} &:= \cyll{\V\setminus\{x\}}\cylr{\V\setminus\{y\}} (R(x,y;) \cap R(y,x;))
\end{align*}
With these definitions, for instance $\sem{R^{lr}_{x_i,x_j}}$ consists of all 
$2n$-tuples such that the $i^\text{th}$ and the $n+j^\text{th}$ element are 
connected (in two directions) in $R$, and similar properties hold for $R^l$ 
and $R^r$.
From this, it follows that the expression
\[\alpha_{2n} = \alpha_{\neq}\cap  \bigcap_{x\neq y \in \V} 
\bigl( R^l_{x,y}\cap R^r_{x,y}\bigr) \cap \bigcap_{x, y \in \V} R^{lr}_{x,y}\]
satisfies the proposition; it intersects $\alpha_{\neq}$ with all the 
expressions stating that each two data elements must be (bidirectionally) 
connected by $R$.
\end{proof}

Notice that $\alpha_{2n}$ can be used to compute \emph{all} the $2n$-cliques of the 
input interpretation.  We now use $\alpha_{2n}$ to check for existence of $3n$-cliques.

\begin{proposition}\label{prop:3nclique}
Suppose that $|\V| = n$ with $n \geq 2$ and let $\Sch = \{R\}$ with $\ar{(R)} = \iar{(R)}=2$. Define
\[
\alpha_{\exists3n} := (\alpha_{2n}\comp{}\alpha_{2n})\cap \alpha_{2n}.
\]
Then, for every \interpretation $\inst{}$, $\sem{\alpha_{\exists3n}}$ is non-empty 
if and only if $\inst{}(R)$ has a $3n$-clique.
\end{proposition}
It is well known that existence of a $3n$-clique is not expressible in 
$\FO{2n}$ \cite{ef_fmt2}.  The above proposition thus immediately implies primitivity of composition.
\begin{theorem}
Suppose that $|\V|=n\geq 2$. Then, composition is primitive in LIF\@. 
Specifically, no LIFnc expression is equivalent to the LIF expression  $\alpha_{\exists3n}$.
\end{theorem}

\section{Related Work}
\label{sec:related}
\newcommand\emphrel[1]{\textbf{#1}}

LIF grew out of the \emphrel{Algebra of Modular Systems} \cite{T:GCAI:2015},
which was developed to provide foundations for programming from
available components. That paper mentions information flows, in connection with 
input--output behavior in classical logic, for the first time. The paper also 
surveys earlier work from the author's group, as well as other closely related work.

%abstractly specified as Model Expansion
%tasks \cite{MitchellT05}.

In a companion paper \cite{flifexfo}, we report on an application
of LIF to \emphrel{querying under limited access patterns}, as for instance offered 
by web services \cite{McilraithSonZeng01}.
That work also
involves inputs and outputs, but only of a syntactic nature, and
for a restricted variant of LIF (called ``forward'' LIF) only.
The property of io-disjointness turned also to be important in
that work, albeit for a quite different purpose.

Our results also relate to
\emphrel{the evaluation problem for LIF}, which takes as input a LIF expression $\alpha$, 
an interpretation $D$, and a valuation $\nu_1$, and where the task is to find all $\nu_2$ 
such that $(\nu_1,\nu_2) \in \sem\alpha$.
From our results, it follows that only the value of $\nu_1$ on the input variables is important, 
and similarly we are only interested in the values of each $\nu_2$ on the output variables. 
A subtle point, however, is that $D$ may be infinite, and moreover, even if $D$ itself is not infinite, 
the output of the evaluation problem may still be. In many cases, it is still possible to obtain a 
finite representation, for instance by using quantifier elimination techniques as done in 
Constraint Databases \cite{cdbbook}.

We have defined the semantics of LIF algebraically, in the style
of \emphrel{cylindric set algebra} \cite{hmt_cylindric,il_cylindric}.
An important difference is the dynamic nature of BRVs which are sets of
\emph{pairs} of valuations, as opposed to sets of valuations
which are the basic objects in cylindric set algebra.

Our optimality theorem was inspired by work on \emphrel{controlled FO}
\cite{fan_bounded_pods}, which had as aim to infer boundedness properties of the 
outputs of first-order queries, given boundedness properties of the input
relations.  Since this inference task is undecidable, the authors
defined syntactic inferences similar in spirit to our syntactic
definition of inputs and outputs.  They show (their Proposition
4.3) that their definitions are, in a sense, sharp.  Note that our optimality 
theorem is stronger in that it shows that no other compositional and sound 
definition can be better than ours.  Of course, the comparison between the 
two results is only superficial as the inference tasks at hand are
very different.

The Logic of Information Flows is similar to \emphrel{dynamic predicate logic} (DPL)
\cite{dplogic}, in the sense that formulas are also evaluated with respect to pairs 
of valuations.  There is, however a key difference in philosophy between the two logics.
LIF starts from the idea that well-known operators from first-order logic can be used to 
describe combinations and manipulations of dynamic systems, and as such provides a means 
for procedural knowledge in a declarative language.
The dynamics in LIF are dynamics of the described system.
Dynamic predicate logic, on the other hand starts from the observation that, in natural language, 
operators such as conjunction and existential quantification are dynamic, where the dynamics are 
in the process of parsing a sentence, often related to coreference analysis. 
\ignore{\orange{E: I would just stop here and say I/O are not invesigated in DPL} 
An example of such dynamics can be found in the sentence
%\begin{quote}
 ``Every farmer who owns a donkey, beats it.''
%\end{quote}
Ideally, sentences beginning with ``Every farmer who owns a donkey'' would be translated into
\[\forall f: Farm(f) \land (\exists d: Don(d)\land Owns(f,d))\rightarrow \phi.\]
However, the standard semantics of first-order logic does not allow for referring back to 
the donkey whose existence was established outside the scope of the existential 
quantification;\marginpar{\orange{E: seems like too much spece for DPL...I am not 
sure it is related to the current investigation}} dynamic predicate logic allows 
this and views first-order formulas as expressions modifying variable assignments. 
The complete sentence would be translated to
 \[\forall f: Farm(f) \land (\exists d: Don(d)\land Owns(f,d))\rightarrow Beat(f,d)\]
in DPL.}
To the best of our knowledge, inputs and outputs of expressions have not been studied in DPL\@.
% Another important difference between LIF and DPL to be noticed is that  
% disjunction in DPL is always interpreted as a test.
% Because of this, expressions such as $R(x;y) \cup S(u;v)$
% seem inexpressible in DPL.

Since we developed a large part of
our work in the general setting of $\bbar$s, and thus of \emphrel{transition systems}, 
we expect several of our results to be applicable in the context of other formalisms 
where specifying inputs and outputs is important, such as API-based programming 
\cite{DBLP:conf/kr/CalvaneseGLV16} and synthesis \cite{DBLP:journals/ai/GiacomoPS13,DBLP:conf/aaai/AlechinaBGFLV19}, 
privacy and security, business process modeling \cite{DBLP:journals/debu/CalvaneseGLMP08}, and model 
combinators in Constraint Programming \cite{DBLP:conf/cp/FontaineMH13}.

% \todo{More relevant related work?}
\section{Conclusion and Future Work}\label{sec:concl}
%
% \marginpar{JVDB: I dislike conclusions that repeat and summarize
% the results.  It looks gratuitous.  I would limit conclusion to
% additional remarks, insights, perspective, and directions for
% further work. \orange{E:does not seem like we have too much of that...}}
%
Declarative modeling is of central importance in the area of Knowledge 
Representation and Reasoning.  The Logic of Information Flows provides a framework 
to investigate how, and to what degree, dynamic or imperative
features can be modeled declaratively.  In this paper we have
focused on inputs, outputs, and sequential composition, as these
three concepts are fundamental to modeling dynamic systems.
There are many directions for further research.

Inputs and outputs are not just relevant from a theoretic perspective, but can also 
have ramifications on computation. Indeed, they form a first handle to parallelize 
computation of complex LIF expressions, or to decompose problems.

In this paper, we have worked with a basic set of operations
motivated by the classical logic connectives.  In order to
provide a fine control of computational complexity, or to
increase expressiveness, it makes sense to consider other
operations. % , such as a fixed point construct used in  \cite{T:GCAI:2015,lif_frocos}.
%BART: I REMOVED FIXPOINT

The semantic notions developed in this paper
(inputs, outputs, soundness) apply to global BRVs in
general, and hence are robust under varying the set of operations.
Moreover, our work delineates and demonstrates a methodology for
adapting syntactic input--output definitions to other operations.

A specific operation that is natural to consider is
\emph{converse}. The converse of a BRV $A$ is defined to be
$\{(\nu_2,\nu_1) \mid (\nu_1,\nu_2) \in A\}$.  In the context of
LIF \cite{lif_frocos} it can model constraint solving by
searching for an input to a module that produces a desired
outcome.  When we add converse to LIF with only a single
variable ($|\V|=1$), and the vocabulary has only binary
relations of input arity one, then we obtain the classical
calculus of relations \cite{tarski_relcalc}.  There, converse is known to 
be primitive \cite{rafragments_ins}.  When the number of variables
is strictly more than half of the maximum arity of relations in the 
vocabulary, converse is redundant in LIF, as can be shown using similar
techniques as used in this paper to show redundancy of
composition. Investigating the exact number of variables needed for 
non-primitivity is an interesting question for further research.

Another direction for further research is to examine fragments of LIF for which the 
semantic input or output problem may be decidable, or even for which the syntactic 
definitions coincide with the semantic definitions.

Finally, an operation that often occurs in dynamic systems is the fixed point construct 
used by \cite{lif_frocos}. It remains to be seen how our work, and the further research 
directions mentioned above, can be extended to include the fixpoint operation.

\begin{acks}
This research received funding from the Flemish Government 
under the ``Onderzoeksprogramma Artifici\"ele Intelligentie (AI) Vlaanderen'' 
programme,
%@JAN: according to Leander, this is the correct way to refer to the programme
from FWO Flanders project G0D9616N,
and from Natural Sciences and Engineering Research Council of Canada (NSERC).  
Jan Van den Bussche is partially supported by the National Natural Science Foundation 
of China (61972455).
Heba Aamer is supported by the Special Research Fund (BOF) (BOF19OWB16).
\end{acks}

\bibliographystyle{ACM-Reference-Format}
\bibliography{database,krrlib}

\end{document}